\documentclass[11pt]{amsart}

\usepackage{amsthm,amsfonts,amsmath,amssymb,verbatim,rotating} 
\usepackage[left=2.54cm,right=2.54cm,top=3.5cm,bottom=3.5cm]{geometry}
\usepackage{tikz,graphicx}
\usepackage{mathrsfs}
\usepackage{braket}

\usepackage{fouridx}

\usepackage[normalem]{ulem}

\usepackage{hyperref}

\allowdisplaybreaks

\numberwithin{equation}{section}
\newtheorem{theorem}{Theorem}[section]
\newtheorem{corollary}[theorem]{Corollary}

\newtheorem{proposition}[theorem]{Proposition}
\newtheorem{lemma}[theorem]{Lemma}

\theoremstyle{definition}

\newtheorem*{remark}{Remark}

\renewcommand{\leq}{\leqslant}
\renewcommand{\geq}{\geqslant}
\newcommand{\GEQ}{\hspace{-3pt}\geq\hspace{-3pt}}
\newcommand{\LE}{\hspace{-3pt}<\hspace{-3pt}}

\renewcommand{\Re}{\textup{Re}}
\renewcommand{\Im}{\textup{Im}}
\newcommand{\dup}{\mathrm{d}}
\newcommand{\eup}{\mathrm{e}}
\newcommand{\iup}{\hspace{1pt}\mathrm{i}\hspace{1pt}}
\newcommand{\Int}{\int\limits}
\newcommand{\bla}{\boldsymbol{\la}}
\newcommand{\bmu}{\boldsymbol{\mu}}
\newcommand{\bnu}{\boldsymbol{\nu}}

\newcommand{\lar}[1]{\lambda^{(#1)}}
\newcommand{\nur}[1]{\nu^{(#1)}}
\newcommand{\mur}[1]{\mu^{(#1)}}

\newcommand{\tar}[1]{t^{(#1)}}

\newcommand{\Symm}{\mathfrak{S}}
\newcommand{\abs}[1]{\lvert#1\rvert}
\newcommand{\Abs}[1]{\big\lvert#1\big\rvert}
\newcommand{\la}{\lambda}
\newcommand{\La}{\Lambda}

\newcommand{\obinomE}{\genfrac\langle\rangle{0pt}{}}

\newcommand{\spec}[1]{\langle #1\rangle}

\newcommand{\A}{\mathrm A}

\DeclareMathOperator{\interior}{int}
\DeclareMathOperator{\exterior}{ext}

\begin{document}

\title[Selberg integrals]{AFLT-type Selberg integrals}

\author{Seamus P. Albion}
\address{Fakult\"at f\"ur Mathematik,
Universit\"at Wien, Oskar-Morgenstern-Platz 1, A-1090, Vienna,
Austria}
\email{seamus.albion@univie.ac.at}

\author{Eric M. Rains}
\address{Department of Mathematics, 
California Institute of Technology,
Pasadena, CA 91125, USA}
\email{rains@caltech.edu}

\author{S. Ole Warnaar}
\address{School of Mathematics and Physics,
The University of Queensland, Brisbane, QLD 4072,
Australia}
\email{o.warnaar@maths.uq.edu.au}

\dedicatory{Dedicated to the memory of Richard (Dick) Askey.}

\thanks{Work supported by the Australian Research Council Discovery
Grant DP170102648.}


\begin{abstract}
In their 2011 paper on the AGT conjecture, Alba, Fateev, Litvinov 
and Tarnopolsky (AFLT) obtained a closed-form evaluation for a Selberg
integral over the product of two Jack polynomials, thereby unifying 
the well-known Kadell and Hua--Kadell integrals.
In this paper we use a variety of symmetric functions and symmetric function
techniques to prove generalisations of the AFLT integral.
These include 
(i) an $\A_n$ analogue of the AFLT integral, containing two Jack 
polynomials in the integrand;
(ii) a generalisation of (i) for $\gamma=1$ (the Schur or GUE case),
containing a product of $n+1$ Schur functions; 
(iii) an elliptic generalisation of the AFLT integral in which the role
of the Jack polynomials is played by a pair of elliptic interpolation 
functions;
(iv) an AFLT integral for Macdonald polynomials.
\smallskip

\end{abstract}

\maketitle

\section{Introduction}
In 2010 Alday, Gaiotto and Tachikawa \cite{AGT10} conjectured a deep 
relationship between $\mathcal{N}=2$ superconformal field theory  
in four dimensions and Liouville conformal field theory on a punctured
Riemann surface.
Their correspondence provides a dictionary
between correlation functions in Liouville field theory 
\cite{Nakayama04,Teschner01} and the Nekrasov partition function 
in $\mathcal{N}=2$ superconformal field theory~\cite{Nekrasov04,NO06}.
One entry of this dictionary relates the instanton part of the Nekrasov 
partition function to conformal blocks in Liouville field theory.
This relationship allowed Alday et al. to derive an explicit
combinatorial expansion for conformal blocks.

One particularly promising approach to the AGT conjecture was developed
by Alba, Fateev, Litvinov and Tarnopolsky~\cite{AFLT11}.
Let $\textrm{Vir}$ and $\mathscr{A}$ denote the Virasoro and Heisenberg
algebras respectively.
Then Alba et al.\ considered representations $L(P,Q)$ of central charge and 
conformal dimension
\begin{equation}\label{Eq_central-charge}
c=1+6Q^2 \quad\text{and}\quad \Delta(P)=Q^2/4-P^2 
\end{equation} 
of $\textrm{Vir}\oplus\mathscr{A}$, and showed that $L(P,Q)$ has a unique
orthogonal basis $\{\ket{P_{\bla}}\}$ indexed by bipartitions $\bla$, 
such that in this basis the matrix element between $L(P,Q)$ and $L(P',Q)$
corresponding to the primary field indexed by $\alpha$ coincides with 
\[
Z_{\textrm{bifund}}\big((P',-P'),\bla;(P,-P),\bmu;\alpha\big).
\]
Here $Z_{\textrm{bifund}}$ is the key building block --- corresponding to
the `bifundamental hypermultiplet' --- of the instanton part of the Nekrasov
partition function, which admits the following explicit combinatorial 
expression \cite{FP03,FMP04,Shadchin06}
\begin{multline}\label{Eq_Zbifund}
Z_{\textrm{bifund}}\big((u_1,u_2),\bla;(v_1,v_2),\bmu;m\big)
=\prod_{i,j=1}^2 \bigg(\,\prod_{s\in \lar{i}}
\Big(E\big(u_i-v_j,\lar{i},\mur{j},s\big)-m\Big) \\
\times \prod_{s\in \mur{j}}
\Big(Q-m-E\big(v_j-u_i,\mur{j},\lar{i},s\big)\Big)\bigg),
\end{multline}
where $\bla=(\lar{1},\lar{2})$, $\bmu=(\mur{1},\mur{2})$
and, for $s=(i,j)\in \la$,
\begin{equation}\label{Eq_E}
E(u,\la,\mu,s)=u-b(\mu'_j-i)+b^{-1}(\la_i-j+1)
\end{equation}
with $Q$ in \eqref{Eq_central-charge} parametrised as $Q=b+b^{-1}$.
In the above $\mu'$ is the conjugate of the partition $\mu$, so that 
$\mu'_j-i$ and $\la_i-j$ may be recognised as the (generalised) leg-length
$l_{\mu}(s)$ and arm-length $a_{\la}(s)$ of the square $s\in\la$,
see \eqref{Eq_arm}.
Lifting the isomorphism \cite{FF82,TK86} between Verma modules for 
$\mathrm{Vir}$ and Fock space representations of $\mathscr{A}$ to the level
of $L(P,Q)$, Alba et al.\ obtained a closed-form expression in 
the spirit of \cite{MY95,SSAFR05,Yanagida11} for the states 
$\ket P_{\la,0}$ in terms of Jack polynomials $P_{\la}^{(-1/b^2)}$.
The more general states $\ket P_{\bla}$ then follow recursively
from $\ket P_{\la,0}$.

For $\mathscr{O}$ a symmetric Laurent polynomial in $k$ variables and
$\alpha,\beta,\gamma\in\mathbb{C}$, define the Selberg average of 
$\mathscr{O}$ as
\[
\big\langle \mathscr{O}\big\rangle_{\alpha,\beta;\gamma}^k
:=\frac{1}{S_k(\alpha,\beta;\gamma)} 
\Int_{[0,1]^k}\! \mathscr{O}(t_1,\dots,t_k)
\prod_{i=1}^k t_i^{\alpha-1}(1-t_i)^{\beta-1}
\prod_{1\leq i<j\leq k}\abs{t_i-t_j}^{2\gamma} \,
\dup t_1\cdots \dup t_k,
\]
where the normalisation $S_k(\alpha,\beta;\gamma)$ is given by the
classical Selberg integral \cite{Forrester10,FW08,Selberg44,TV19}
\begin{align}\label{Eq_Selberg}
S_k(\alpha,\beta;\gamma)&:=
\Int_{[0,1]^k}\! \prod_{i=1}^k t_i^{\alpha-1}(1-t_i)^{\beta-1}
\prod_{1\leq i<j\leq k}\abs{t_i-t_j}^{2\gamma}
\,\dup t_1\cdots \dup t_k \\
&\hphantom{:}=\prod_{i=1}^k
\frac{\Gamma(\beta+(i-1)\gamma)\Gamma(\alpha+(i-1)\gamma)\Gamma(1+i\gamma)}
{\Gamma(\alpha+\beta+(2k-i-1)\gamma)\Gamma(1+\gamma)}, \notag 
\end{align}
for $\mathrm{Re}(\alpha)>0$, $\mathrm{Re}(\beta)>0$ and
\[
\mathrm{Re}(\gamma)>-\min\{1/k,\mathrm{Re}(\alpha)/(k-1),
\mathrm{Re}(\beta)/(k-1)\}.
\]
A crucial step taken by Alba et al.\ was to show that
for $P+P'+\alpha+kb=0$ 
\begin{multline}\label{Eq_Z-Selb}
Z_{\textrm{bifund}}\big((P',-P'),(\la,0);(P,-P),(\mu,0);\alpha\big) \\[1mm]
=\kappa_{\la}(P)\kappa_{\mu}(P') \Big\langle 
P_{\la}^{(-1/b^2)}\big(t^{-1}\big)
P_{\mu}^{(-1/b^2)}\big[t+(2b\alpha-1-b^2)/b^2\big]
\Big\rangle_{1-b(Q+2P),1-2b\alpha;-b^2}^k.
\end{multline}
In the above,
\[
\kappa_{\la}(P):=
\prod_{(i,j)\in\la} \Big(\big(b(i-\la'_j-1)+b^{-1}(\la_i-j)\big)
\big(2P+bi+b^{-1}j\big)\Big),
\]
$f(t^{-1}):=f(t_1^{-1},\dots,t_k^{-1})$ and, for $f$ a symmetric function
which expands in terms of the power sums $p_{\la}$ as
$f=\sum_{\la} c_{\la} p_{\la}$, $f[t+z]$ is plethystic notation
(see Section~\ref{Sec_plethystic} for details) for 
\[
\sum_{\la} c_{\la} \prod_{i=1}^{l(\la)}
\Big( p_{\la_i}(t_1,\dots,t_k)+z \Big).
\]
The verification of \eqref{Eq_Z-Selb} boils down to computing 
the Selberg average
\[
\Big\langle P^{(1/\gamma)}_{\la}\big(t^{-1}\big) 
P^{(1/\gamma)}_{\mu}\big[t+\beta/\gamma-1\big] 
\Big\rangle_{\alpha,\beta;\gamma}^k
\]
and comparing this against the explicit form of $Z_{\textrm{bifund}}$
provided by \eqref{Eq_Zbifund}.
By the complementation symmetry
\[
P^{(1/\gamma)}_{\la}(t^{-1})=
(t_1\cdots t_k)^{-N} P^{(1/\gamma)}_{(N-\la_k,\dots,N-\la_1)}(t),
\]
for $\la\in\mathscr{P}_k$ (the set of partitions of length at most $k$) 
and $N$ an arbitrary integer such that $N\geq \la_1$, this is achieved by 
the following integral evaluation~\cite{AFLT11} (see also \cite{MMSS12}).
Let $\mathscr{P}$ denote the set of partitions.

\begin{theorem}[AFLT integral]\label{Thm_AFLT}
Let $k$ be a positive integer, $\la,\mu\in\mathscr{P}$ 
and $\alpha,\beta,\gamma\in\mathbb{C}$.
Then
\begin{align}\label{Eq_AFLT}
&\Int_{[0,1]^k}\!
P_{\la}^{(1/\gamma)}[t]P_{\mu}^{(1/\gamma)}[t+\beta/\gamma-1]\,
\prod_{i=1}^k t_i^{\alpha-1}(1-t_i)^{\beta-1}
\prod_{1\leq i<j\leq k} \abs{t_i-t_j}^{2\gamma}
\,\dup{t_1}\cdots \dup t_k \\ 
\notag
&\quad=P_{\la}^{(1/\gamma)}[k]P_{\mu}^{(1/\gamma)}[k+\beta/\gamma-1]
\prod_{i=1}^k
\frac{\Gamma(\beta+(i-1)\gamma)\Gamma(\alpha+(k-i)\gamma+\la_i)
\Gamma(1+i\gamma)}
{\Gamma(\alpha+\beta+(2k-m-i-1)\gamma+\la_i)\Gamma(1+\gamma)} \\\notag
&\qquad\times\prod_{i=1}^k\prod_{j=1}^m
\frac{\Gamma(\alpha+\beta+(2k-i-j-1)\gamma+\la_i+\mu_j)}
{\Gamma(\alpha+\beta+(2k-i-j)\gamma+\la_i+\mu_j)},
\end{align}
where $m$ is an arbitrary integer such that $m\geq l(\mu)$, and
\[
\Re(\alpha)>-\la_k,\quad \Re(\beta)>0,\quad
\Re(\gamma)>-\min_{1\leq i\leq k-1}
\bigg\{\frac{1}{k},\frac{\Re(\alpha)+\la_i}{k-i},\frac{\Re(\beta)}{k-1}
\bigg\}.
\]
\end{theorem}

The arguments of the Jack polynomials on the right-hand side are again
expressed using plethystic notation, see Section~\ref{Sec_plethystic} for
details.
Alternatively, both $P_{\la}^{(1/\gamma)}[k]$ and
$P_{\mu}^{(1/\gamma)}[k+\beta/\gamma-1]$ may be written in fully factorised
form by equation \eqref{Eq_Jack-z} below.
When $\mu=0$ the AFLT integral simplifies to Kadell's integral
\cite{Kadell97} (see also \cite{Macdonald87}) and 
for $\beta=\gamma$ it yields the Hua--Kadell integral~\cite{Hua79,Kadell93}.
Theorem~\ref{Thm_AFLT} is proved in \cite{AFLT11} by generalising the
Anderson-style recursive proof of Kadell's integral given in \cite{Warnaar08a}.
Key input in both these proofs is the Okounkov--Olshanski integral 
formula for Jack polynomials \cite{Okounkov98a,OO97}.

In the conclusion to their paper \cite{AFLT11}, Alba et al.\ remark that 
the generalisation of their construction to $\mathrm{WA}_n$ 
requires a generalisation of the $\A_n$ Selberg integral of \cite{Warnaar09}
with two Jack polynomials included in the integrand.\footnote{See also
\cite{CPT20,IOY13} for the relation between $\A_n$ Selberg integrals and
the AGT conjecture.}
For $\A_2$ such an AFLT-type integral was considered by Fateev and Litvinov
in \cite{FL12}, where they again used a recursion based on the 
Okounkov--Olshanski integral to obtain a closed-form evaluation.
They also claim a more general $\A_n$ AFLT integral, but the evaluation
of this integral is only implicit and, unfortunately, the stated recursion
relation that needs to be solved to obtain the evaluation is incorrect
for $n\geq 3$ (see also \eqref{Eq_recR} in Section~\ref{Sec_open}).
 
The first main result of the present paper is the evaluation of 
the $\A_n$ AFLT integral for all $n$.
Before stating this result we introduce some notation.
For $k,\ell$ nonnegative integers and $t=(t_1,\dots,t_k)$, 
$s=(s_1,\dots,s_{\ell})$ we define the Vandermonde products
$\Delta(t)$ and $\Delta(t,s)$ as
\begin{equation}\label{Eq_Vandermonde}
\Delta(t):=\prod_{1\leq i<j\leq k} (t_i-t_j)\quad\text{and}\quad
\Delta(t,s):=\prod_{i=1}^k\prod_{j=1}^{\ell}(t_i-s_j).
\end{equation}
We also abbreviate $\dup t_1\cdots \dup t_k$ as $\dup t$.

For $n$ a positive integer, let $k_1,\dots,k_n$ be integers such that
$0\leq k_1\leq\cdots\leq k_n$ and let $\tar{1},\dots,\tar{n}$ 
be a sets of variables (or alphabets) such that $\tar{r}$ has
cardinality $k_r$.
Further let $\alpha_1,\dots,\alpha_n,\beta\in\mathbb{C}$ such that
\begin{subequations}\label{Eq_conditions}
\begin{gather}\label{Eq_conditions-a}
\Re(\beta)>0,\quad
\abs{\Re(\gamma)}<\frac{1}{k_n}, \quad
\Re\big(\beta+(k_n-1)\gamma\big)>0, \\[2mm]
\Re\big(\alpha_r+\dots+\alpha_s+(r-s+i-1)\gamma\big)>0
\quad\text{for $1\leq r\leq s\leq n$ and $1\leq i\leq k_r-k_{r-1}$}, 
\end{gather}
\end{subequations}
where $k_0:=0$.\footnote{The condition $\Re(\gamma)<1/k_n$ may be dropped
when $n=1$.}
We then define the $\A_n$ Selberg average of a polynomial
$\mathscr{O}(\tar{1},\dots,\tar{n})$, symmetric in each of the alphabets
$\tar{r}$, as
\begin{equation}\label{Eq_An-Selberg-average}
\big\langle \mathscr{O}\big\rangle_{\alpha_1,\dots,\alpha_n,\beta;\gamma}
^{k_1,\dots,k_n}:=\frac{I^{\A_n}_{k_1,\dots,k_n}
(\mathscr{O};\alpha_1,\dots,\alpha_n,\beta;\gamma)}
{I^{\A_n}_{k_1,\dots,k_n}(1;\alpha_1,\dots,\alpha_n,\beta;\gamma)},
\end{equation}
where 
\begin{align}\label{Eq_I-An}
I^{\A_n}_{k_1,\dots,k_n}&(\mathscr{O};\alpha_1,\dots,\alpha_n,\beta;\gamma) \\
&:=\Int_{C_{\gamma}^{k_1,\dots,k_n}[0,1]}
\mathscr{O}\big(\tar{1},\dots,\tar{n}\big)
\prod_{r=1}^n \prod_{i=1}^{k_r}\big(\tar{r}_i\big)^{\alpha_r-1}
\big(1-\tar{r}_i\big)^{\beta_r-1} \notag \\[-1mm]
&\qquad\qquad\qquad\quad\times
\prod_{r=1}^n \Abs{\Delta\big(\tar{r}\big)}^{2\gamma}
\prod_{r=1}^{n-1} \Abs{\Delta\big(\tar{r},\tar{r+1}\big)}^{-\gamma}\, 
\dup\tar{1}\cdots\dup\tar{n}. \notag
\end{align}
Here
\begin{equation}\label{Eq_beta-r}
\beta_1=\dots=\beta_{n-1}:=1, \quad \beta_n:=\beta 
\end{equation}
and 
$C_{\gamma}^{k_1,\dots,k_n}[0,1]$ is a real domain
of integration described in Section~\ref{Sec_chain}.
The normalisation 
$I^{\A_n}_{k_1,\dots,k_n}(1;\alpha_1,\dots,\alpha_n,\beta;\gamma)$ in 
\eqref{Eq_An-Selberg-average} is the $\A_n$ Selberg integral of 
\cite[Theorem~1.2]{Warnaar09} (see also \cite[Theorem~3.3]{TV03} 
for the $\A_2$ case), which admits the evaluation
\begin{align}\label{Eq_An-Selberg}
&I^{\A_n}_{k_1,\dots,k_n}(1;\alpha_1,\dots,\alpha_n,\beta;\gamma) \\
&\quad\quad=\prod_{r=1}^n\prod_{i=1}^{k_r}
\frac{\Gamma(\beta_r+(i-k_{r+1}-1)\gamma)\Gamma(i\gamma)}
{\Gamma(\gamma)} \notag \\
&\quad\qquad\times \prod_{1\leq r\leq s\leq n} \prod_{i=1}^{k_r-k_{r-1}} 
\frac{\Gamma(\alpha_r+\dots+\alpha_s+(r-s+i-1)\gamma)}
{\Gamma(\alpha_r+\dots+\alpha_s+\beta_s+(k_s-k_{s+1}+r-s+i-2)\gamma)},
\notag 
\end{align}
where, once again, $k_0=k_{n+1}:=0$.

For $n$ a nonnegative integer let $(a)_n:=\Gamma(a+n)/\Gamma(a)=
a(a+1)\cdots(a+n-1)$ denote the Pochhammer symbol,
and let $\delta_{u,v}$ be the usual Kronecker delta.
Then the $\A_n$ analogue of the AFLT integral is given by the
following identity for the Selberg average of the product of
two Jack polynomials.

\begin{theorem}[$\A_n$ AFLT integral]\label{Thm_An-AFLT}
For $n$ a positive integer, let $k_1,\dots,k_n$ be integers 
such that $0\leq k_1\leq\cdots\leq k_n$.
Then for $\alpha_1,\dots,\alpha_n,\beta,\gamma\in\mathbb{C}$
such that \eqref{Eq_conditions} holds and 
$\la,\mu\in\mathscr{P}$, we have
\begin{align}\label{Eq_An-AFLT}
&\Big\langle P_{\la}^{(1/\gamma)}\big[\tar{1}\big]
P_{\mu}^{(1/\gamma)}\big[\tar{n}+\beta/\gamma-1\big]
\Big\rangle_{\alpha_1,\dots,\alpha_n,\beta;\gamma}^{k_1,\dots,k_n} \\[2mm]
&\qquad=P_{\la}^{(1/\gamma)}[k_1]
P_{\mu}^{(1/\gamma)}[k_n+\beta/\gamma-1] \notag \\
&\qquad\quad\times
\prod_{r=1}^n \prod_{i=1}^{\ell}
\frac{(\alpha_1+\dots+\alpha_r+(k_1-r-i+1)\gamma)_{\la_i}}
{(\alpha_1+\dots+\alpha_r+\beta_r+
(k_1+k_r-k_{r+1}-r-m\delta_{r,n}-i)\gamma)_{\la_i}} \notag \\
&\qquad\quad\times
\prod_{r=1}^n\prod_{j=1}^m
\frac{(\alpha_r+\dots+\alpha_n+\beta+(k_n+r-n-j-1)\gamma)_{\mu_j}}
{(\alpha_r+\dots+\alpha_n+\beta+(k_r-k_{r-1}+k_n+r-n-\ell\delta_{r,1}-j-1)
\gamma)_{\mu_j}}\notag \\
&\qquad\quad\times
\prod_{i=1}^{\ell}\prod_{j=1}^m
\frac{(\alpha_1+\dots+\alpha_n+\beta+(k_1+k_n-n-i-j)\gamma)_{\la_i+\mu_j}}
{(\alpha_1+\dots+\alpha_n+\beta+(k_1+k_n-n-i-j+1)\gamma)_{\la_i+\mu_j}}.
\notag
\end{align}
In the expression on the right, $\ell$ and $m$ are arbitrary integers such
that $\ell\geq l(\la)$, $m\geq l(\mu)$, $k_0=k_{n+1}:=0$ and the $\beta_r$
are as in \eqref{Eq_beta-r}.
\end{theorem}

Note that both sides of \eqref{Eq_An-AFLT} trivially vanish unless
$l(\la)\leq k_1$ so that without loss of generality it may be assumed
that $\la\in\mathscr{P}_{k_1}$.
Then the $r=1$ term in the second double product on
the right simplifies to $1$ upon choosing $\ell=k_1$.

Since
\begin{equation}\label{Eq_simplex}
C_{\gamma}^{k_1}[0,1]=\{t\in \mathbb{R}^{k_1}: 0<t_1<t_2<\dots<t_{k_1}<1\},
\end{equation}
Theorem~\ref{Thm_An-AFLT} for $n=1$ is equivalent to Theorem~\ref{Thm_AFLT}.
For $\mu=0$ the theorem corresponds to the $\A_n$ analogue of Kadell's
integral \cite[Theorem~6.1]{Warnaar09}, and for $\beta=\gamma$ it gives an
$\A_n$ analogue of the Hua--Kadell integral \cite{Hua79,Kadell93}.
Our proof of Theorem~\ref{Thm_An-AFLT} is not reliant on the
Okounkov--Olshanski integral formula for Jack polynomials and instead
uses $\A_n$ Cauchy-type identities for Macdonald polynomials.
One advantage of our approach is that it immediately implies a companion to
the $\A_n$ AFLT integral, stated as Theorem~\ref{Thm_An-alt} in
Section~\ref{Sec_AFLT}. 

\medskip

Setting $k_1=0$ in the $\A_n$ Selberg integral yields the $\A_{n-1}$ Selberg
integral.
The same is not true, however, for the $\A_n$ AFLT integral.
Setting $k_1=0$ in \eqref{Eq_An-AFLT} forces $\la=0$ for nonvanishing, 
thus eliminating one of the two Jack polynomials in the integrand.
In their work on the AGT conjecture for $\mathrm{WA}_n$,
Matsuo and Zhang \cite{ZM11} formulated several conjectures for 
$\A_n$ Selberg integrals of AFLT type that do have the desired
reduction property, but unfortunately as stated their conjectures
appear to be false.
In the $\gamma\to 1$ limit of Theorem~\ref{Thm_An-AFLT}, in which case the
Jack polynomials $P^{(1/\gamma)}_{\la}\big[\tar{1}\big]$ and
$P^{(1/\gamma)}_{\mu}\big[\tar{n}+\beta/\gamma-1\big]$
simplify to the Schur functions
$s_{\la}\big[\tar{1}\big]$ and $s_{\mu}\big[\tar{n}+\beta-1\big]$, 
we have managed to prove a corrected Matsuo--Zhang-type AFLT integral,
containing a product of $n+1$ Schur functions in the integrand.
Because \eqref{Eq_conditions-a} implies that $\abs{\Re(\gamma)}<1$,
in this integral we must replace the real domain of integration
$C_{\gamma}^{k_1,\dots,k_n}[0,1]$ by the complex contour 
\begin{equation}\label{Eq_contour}
C^{k_1,\dots,k_n}=C_1^{k_1}\times\dots\times C_n^{k_n},
\quad \text{where}\quad 
C_r^{k_r}=\underbrace{C_r\times\dots\times C_r}_{k_r \textrm{ times}}
\end{equation}
defined as follows.
Each $C_r$ is a positively oriented Jordan curve which 
passes through the origin, contains the interval $(0,1]$ in its
interior, and has nonzero slope close to $0$.
Away from $0$, the contour $C_r$ for $r\geq 2$ is contained in the
interior of $C_{r-1}$ as shown in the figure below

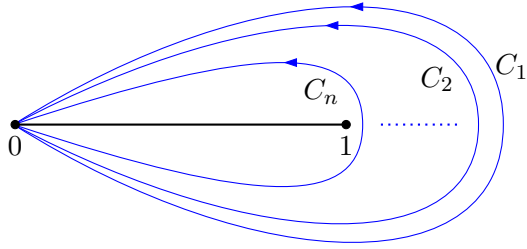
\tikzset{arro/.pic={
\filldraw[blue] (0,-0.05) -- (-0.15,0) -- (0,0.05) -- cycle;}}

\begin{center}
\begin{tikzpicture}[scale=1.1]
\draw (0,-0.25) node {$0$};
\draw (4,-0.25) node {$1$};
\draw[blue] (-0.02,0) .. controls (3,1) and (4.2,1) .. (4.2,0) 
.. controls (4.2,-1) and (3,-1) .. (-0.02,0);
\draw[blue,dotted,thick] (4.42,0) -- (5.38,0);
\draw[blue] (-0.02,0) .. controls (3,1.6) and (5.6,1.6) .. (5.6,0) 
.. controls (5.6,-1.6) and (3,-1.6) .. (-0.02,0);
\draw[blue] (-0.02,0) .. controls (3,1.9) and (5.9,1.9) .. (5.9,0) 
.. controls (5.9,-1.9) and (3,-1.9) .. (-0.02,0);
\draw (3.7,0.4) node {$C_n$};
\draw (6.0,0.7) node {$C_1$};
\draw (5.1,0.55) node {$C_2$};
\pic at (3.4,0.75) {arro};
\pic at (3.9,1.198) {arro};
\pic at (4.2,1.425) {arro};
\draw[thick] (0,0)--(4,0);
\draw[fill] (0,0) circle (0.05cm);
\draw[fill] (4,0) circle (0.05cm); 
\end{tikzpicture}
\end{center}

For $\gamma=1$ we now redefine the $\A_n$ Selberg average as follows.
In the complex $t_i^{(r)}$-plane fix the usual principal branch of 
the complex logarithm, with cut along the negative real axis and argument 
in $(-\pi,\pi]$.
Then for $0\leq k_1\leq\cdots\leq k_n$ and
$\alpha_1,\dots,\alpha_n,\beta\in\mathbb{C}$ such that
\begin{equation}\label{Eq_conv}
\Re(\alpha_r+\cdots+\alpha_s)>s-r \quad
\text{for $1\leq r\leq s\leq n$}
\end{equation}
we define
\[
\big\langle \mathscr{O}\big\rangle_{\alpha_1,\dots,\alpha_n,\beta}
^{k_1,\dots,k_n}:=
\frac{I^{\A_n}_{k_1,\dots,k_n}
(\mathscr{O};\alpha_1,\dots,\alpha_n,\beta)}
{I^{\A_n}_{k_1,\dots,k_n}(1;\alpha_1,\dots,\alpha_n,\beta)},
\]
where, assuming \eqref{Eq_beta-r}, 
\begin{align}\label{Eq_I-def}
&I^{\A_n}_{k_1,\dots,k_n}(\mathscr{O};\alpha_1,\dots,\alpha_n,\beta) \\
&\qquad :=\frac{1}{(2\pi\iup)^{k_1+\dots+k_n}} \Int_{C^{k_1,\dots,k_n}}
\mathscr{O}\big(\tar{1},\dots,\tar{n}\big)
\prod_{r=1}^n \prod_{i=1}^{k_r}\big(\tar{r}_i\big)^{\alpha_r-1}
\big(\tar{r}_i-1\big)^{\beta_r-1} \notag \\[-1mm]
&\qquad\qquad\qquad\qquad\qquad\qquad\qquad\times
\prod_{r=1}^n \Delta^2\big(\tar{r}\big)
\prod_{r=1}^{n-1} \Delta^{-1}\big(\tar{r},\tar{r+1}\big)\, 
\dup\tar{1}\cdots\dup\tar{n}. \notag
\end{align}
The integral \eqref{Eq_I-def} should be understood in the sense of 
indefinite integrals since the integrand is not defined at 
$t_i^{(r)}=0\in C_r$.
Due to the change in contour, the normalisation is now given by 
(see \eqref{Eq_An-one})
\begin{align}\label{Eq_An-One}
&I^{\A_n}_{k_1,\dots,k_n}(1;\alpha_1,\dots,\alpha_n,\beta) \\
&\quad\quad =\prod_{r=1}^n \bigg( (-1)^{\binom{k_r}{2}} \prod_{i=1}^{k_r}
\frac{i!}{\Gamma(k_{r+1}-\beta_r+2-i)} \bigg) \notag \\
&\quad\qquad\times \prod_{1\leq r\leq s\leq n} \prod_{i=1}^{k_r-k_{r-1}} 
\frac{\Gamma(\alpha_r+\dots+\alpha_s+r-s+i-1)}
{\Gamma(\alpha_r+\dots+\alpha_s+\beta_s+k_s-k_{s+1}+r-s+i-2)}, \notag
\end{align}
where $k_0=k_{n+1}:=0$.

Let $\tar{0}:=0$.
Then our next main result is a closed form evaluation of
\begin{equation}\label{Eq_gammaisone}
\bigg\langle 
\bigg(\prod_{r=1}^n s_{\lar{r}}\big[\tar{r}-\tar{r-1}\big]\bigg)
s_{\lar{n+1}}\big[\tar{n}+\beta-1\big]
\bigg\rangle_{\alpha_1,\dots,\alpha_n,\beta}^{k_1,\dots,k_n},
\end{equation}
generalising the $\gamma=1$ case of \eqref{Eq_An-AFLT}.
The most concise way to state this is by using the duality 
\cite[p.~43]{Macdonald95}
\begin{equation}\label{Eq_Schur-duality}
s_{\mu'}[X]=(-1)^{\abs{\mu}} s_{\mu}[-X],
\end{equation}
and to instead give the evaluation of
\[
\bigg\langle 
\prod_{r=1}^{n+1} s_{\lar{r}}\big[\tar{r}-\tar{r-1}\big]
\bigg\rangle_{\alpha_1,\dots,\alpha_n,\beta}^{k_1,\dots,k_n},
\]
where $\tar{n+1}:=1-\beta$.\footnote{For the evaluation of the average
\eqref{Eq_gammaisone} see equation \eqref{Eq_gamma-one} below.}
Before stating this evaluation we introduce the following
shorthand notation.
For $1\leq r\leq n+1$ let
\begin{subequations}\label{Eq_A-def}
\begin{equation}
A_r:=\alpha_r+\dots+\alpha_n+k_r-k_{r-1}+r
\end{equation}
and $A_{r,s}:=A_r-A_s$, so that $A_{r,s}=-A_{s,r}$.
In particular,
\begin{equation}
A_{r,s}=\alpha_r+\dots+\alpha_{s-1}+k_r-k_{r-1}-k_s+k_{s-1}+r-s
\end{equation}
\end{subequations}
for $1\leq r\leq s\leq n+1$.

\begin{theorem}\label{Thm_nplusone}
For $n$ a positive integer, let $0\leq k_1\leq\cdots\leq k_n$ be
integers, $\alpha_1,\dots,\alpha_n,\beta\in\mathbb{C}$ such that
\eqref{Eq_conv} holds, and $\lar{1},\dots,\lar{n+1}\in\mathscr{P}$.
Let $\tar{0}:=0$ and $\tar{n+1}:=1-\beta$.
Then
\begin{align}\label{Eq_nplusone}
&\bigg\langle 
\prod_{r=1}^{n+1} s_{\lar{r}}\big[\tar{r}-\tar{r-1}\big]
\bigg\rangle_{\alpha_1,\dots,\alpha_n,\beta}^{k_1,\dots,k_n} \\
&\qquad=
\prod_{r=1}^{n+1}\prod_{1\leq i<j\leq \ell_r} 
\frac{\lar{r}_i-\lar{r}_j+j-i}{j-i} 
\prod_{r,s=1}^{n+1}\prod_{i=1}^{\ell_r}
\frac{(A_{r,s}-k_{s-1}+k_s-i+1)_{\lar{r}_i}}
{(A_{r,s}+\ell_s-i+1)_{\lar{r}_i}} \notag \\
&\qquad\quad\times 
\prod_{1\leq r<s\leq n+1}\prod_{i=1}^{\ell_r}\prod_{j=1}^{\ell_s}
\frac{\lar{r}_i-\lar{s}_j+A_{r,s}+j-i}{A_{r,s}+j-i}, \notag
\end{align}
where $k_0:=0$ and $k_{n+1}:=1-\beta$, and where $\ell_r$ 
\textup{(}$1\leq r\leq n+1$\textup{)}
is an arbitrary nonnegative integer such that $\ell_r\geq l(\lar{r})$.
\end{theorem}

The reader is warned that in order to obtain the above compact form 
for the right-hand side we have used a different convention for $k_{n+1}$ 
than in the previous two theorems.
We also remark that \eqref{Eq_nplusone} displays a significant
amount of nontrivial cancellation.
For $s=r$ the second triple product on the right becomes
\[
\prod_{r=1}^{n+1}\prod_{i=1}^{\ell_r}
\frac{(k_r-k_{r-1}-i+1)_{\lar{r}_i}}
{(\ell_r-i+1)_{\lar{r}_i}}.
\]
Since $\ell_r\geq l(\lar{r})$, this shows that the right-hand side
vanishes unless $l(\lar{r})\leq k_r-k_{r-1}$ for all $1\leq r\leq n$.
The integrand, however, only vanishes for $l(\lar{1})\leq k_1-k_0=k_1$.
Finally we note that \eqref{Eq_nplusone} has the desired rank-reduction
property. 
If we denote either side of \eqref{Eq_nplusone} by
\[
I_{\lar{1},\dots,\lar{n+1}}^{k_1,\dots,k_n}
(\alpha_1,\dots,\alpha_n,\beta),
\]
then it is readily verified that
\[
I_{0,\lar{2},\dots,\lar{n+1}}^{0,k_2,\dots,k_n}
(\alpha_1,\alpha_2,\dots,\alpha_n,\beta)
=
I_{\lar{2},\dots,\lar{n+1}}^{k_2,\dots,k_n}
(\alpha_2,\dots,\alpha_n,\beta).
\]

Evaluating the Selberg integral \eqref{Eq_Selberg} for $\gamma=1$
is not at all hard; it follows from Heine's integral formula for 
the Hankel determinant of moments of orthogonal polynomials on the real 
line applied to the case of Jacobi polynomials, see e.g., 
\cite{LT03,Rosengren18}.
We have not found a similarly elementary proof of Theorem~\ref{Thm_nplusone},
and our proof hinges on a novel type of integration formula for Schur 
functions indexed by sequences of complex numbers, see Theorem~\ref{Thm_Schur}.

\medskip

Our third main result applies to $\A_1$, and is an elliptic
generalisation of the AFLT integral \eqref{Eq_AFLT}.
This integral, which also generalises the elliptic 
Selberg integral of \cite{vDS00,vDS01,Rains10}, 
contains two elliptic skew interpolation functions in the
integrand.
These play the role of the pair of Jack polynomials in the 
AFLT integral.

For $\bla,\bnu$ a pair of bipartitions, let 
$R^{\ast}_{\bla/\bnu}([v_1,\dots,v_{2m}];a,b;t,p,q)$ 
denote an elliptic skew interpolation function \cite{Rains12}, for which
we use the shorthand notation
\begin{multline}\label{Eq_Rastpm}
R^{\ast}_{\bla/\bnu}
\big([uz_1^{\pm},\dots,uz_n^{\pm},v_1,\dots,v_{2m}];a,b;t;p,q) \\
:=R^{\ast}_{\bla/\bnu}
\big([uz_1,uz_1^{-1},\dots,uz_n,uz_n^{-1},v_1,\dots,v_{2m}];a,b;t;p,q).
\end{multline}
Further let $(a_1,\dots,a_k;q)_{\infty}:=\prod_{r=1}^k \prod_{i=0}^{\infty}
(1-a_r q^i)$ and let $\Gamma(z;p,q)$ denote the elliptic gamma function 
\cite{Ruijsenaars97}, for which we adopt the usual multiplicative
plus-minus conventions
\begin{align*}
\Gamma(z^{\pm};p,q)&:=\Gamma(z;p,q)\Gamma(z^{-1};p,q), \\
\Gamma(z^{\pm}w^{\pm};p,q)&:=\Gamma(zw;p,q)\Gamma(zw^{-1};p,q)
\Gamma(z^{-1}w;p,q)\Gamma(z^{-1}w^{-1};p,q).
\end{align*}
For $\bla=(\lar{1},\lar{2})\in\mathscr{P}_n^2$
we define the following shorthand for the ratio of products of elliptic gamma 
functions,
\[
\Delta_{\bla}^0(a\vert b_1,\dots,b_k;t;p,q):=
\Delta_{\lar{1}}^0(a\vert b_1,\dots,b_k;q,t;p)
\Delta_{\lar{2}}^0(a\vert b_1,\dots,b_k;p,t;q),
\]
with
\[
\Delta_{\la}^0(a\vert b_1,\dots,b_k;q,t;p):=
\prod_{r=1}^k \prod_{i=1}^n 
\frac{\Gamma(t^{1-i}q^{\la_i}b_r;p,q)\Gamma(pqt^{1-i}a/b_r;p,q)}
{\Gamma(t^{1-i}b_r;p,q)\Gamma(pqt^{1-i}q^{\la_i}a/b_r;p,q)}.
\]
In particular, for $\mu\in\mathscr{P}_m$,
\[
\Delta_{\mu}^0\big(a\vert b\spec{\bla}_{n;t;p,q};q,t;p\big)=
\prod_{i=1}^n \prod_{j=1}^m 
\frac{\Gamma(t^{n-i-j+1}q^{\lar{1}_i+\mu_j}p^{\lar{2}_i}b;p,q)
\Gamma(pqt^{i-j-n+1}q^{-\lar{1}_i}p^{-\lar{2}_i}a/b;p,q)} 
{\Gamma(t^{n-i-j+1}q^{\lar{1}_i}p^{\lar{2}_i}b;p,q)
\Gamma(pqt^{i-j-n+1}q^{-\lar{1}_i+\mu_j}p^{-\lar{2}_i}a/b;p,q)},
\]
where
\begin{equation}\label{Eq_e-spec}
\spec{\bla}_{n;t;p,q}:=\big(q^{\lar{1}_1}p^{\lar{2}_1}t^{n-1}
,q^{\lar{1}_2}p^{\lar{2}_2}t^{n-2},\dots,
q^{\lar{1}_{n-1}}p^{\lar{2}_{n-1}}t,
q^{\lar{1}_n}p^{\lar{2}_n}\big),
\end{equation}
is the spectral vector for the bipartition $\bla$.
With the above notation we may state the following elliptic analogue of 
\eqref{Eq_AFLT}, where we note that instead of $k$ we use
$n$ to denote the number of integration variables.

\begin{theorem}[Elliptic AFLT integral]\label{Thm_eAFLT}
Let $n$ be a positive integer, $p,q,t,t_1,t_2,t_3,t_4,t_5,t_6\in\mathbb{C}$ 
such that the elliptic balancing condition
\begin{equation}\label{Eq_balancing}
t^{2n-2} t_1\cdots t_6=pq
\end{equation}
holds and such that $\abs{p},\abs{q}<1$.
Further let 
\[
\kappa_n:=\frac{(p;p)_{\infty}^n(q;q)_{\infty}^n \Gamma^n(t;p,q)}
{2^n n!(2\pi\iup)^n}.
\]
Then, for $\bla\in\mathscr{P}_n^2$ and $\bmu\in\mathscr{P}^2$,
\begin{align}\label{Eq_eAFLT}
\kappa_n & \Int_{C_{\bla\bmu}}
R^{\ast}_{\bla/\boldsymbol{0}}
\big([t^{1/2}z_1^{\pm},\dots,t^{1/2}z_n^{\pm}];
t^{n-1/2}t_1,t^{1/2}t_2;t;p,q\big) \\[-1mm]
& \qquad\times R^{\ast}_{\bmu/\boldsymbol{0}}
\big([t^{1/2}z_1^{\pm},\dots,t^{1/2}z_n^{\pm},
t^{-1/2}t_4,t^{-1/2}t_5];t^{n-3/2}t_3t_4t_5,t^{1/2}t_6;t;p,q\big) 
\notag \\[1mm]
& \qquad\times \prod_{1\leq i<j\leq n}
\frac{\Gamma(tz_i^{\pm}z_j^{\pm};p,q)}{\Gamma(z_i^{\pm}z_j^{\pm};p,q)}
\prod_{i=1}^n \frac{\prod_{r=1}^6 \Gamma(t_r z_i^{\pm};p,q)}
{\Gamma(z_i^{\pm 2};p,q)} \,
\frac{\dup z_1}{z_1}\cdots\frac{\dup z_n}{z_n} \notag \\
&=\prod_{i=1}^n 
\bigg(\Gamma(t^i;p,q) 
\prod_{1\leq r<s\leq 6} \Gamma(t^{i-1}t_rt_s;p,q)\bigg) \notag \\[1mm]
&\quad \times 
\Delta^0_{\bla}(t^{n-1}t_1/t_2\vert 
t^n,t^{n-1}t_1t_3,t^{n-1}t_1t_4,t^{n-1}t_1t_5,t^{n-1}t_1t_6;t;p,q) 
\notag \\[1.5mm]
&\quad \times 
\Delta^0_{\bmu}(t^{n-2}t_3t_4t_5/t_6\vert 
t^{n-1}t_3t_4,t^{n-1}t_3t_5,t^{n-1}t_4t_5;t;p,q) \notag \\[1mm]
&\quad \times 
\frac{\Delta^0_{\bmu}(t^{n-2}t_3t_4t_5/t_6\vert
t^{n-2}t_1t_3t_4t_5 \spec{\boldsymbol{\bla}}_{n;t;p,q})}
{\Delta^0_{\bmu}(t^{n-2}t_2t_3t_4/t_5\vert
t^{n-1}t_1t_3t_4t_5 \spec{\boldsymbol{\bla}}_{n;t;p,q})}, \notag 
\end{align}
where $C_{\bla\bmu}$ is a deformation of $\mathbb{T}^n$ 
\textup{(}with $\mathbb{T}$ the positively oriented unit circle\textup{)} 
separating sequences of poles of the integrand tending to zero from sequences
of poles tending to infinity.
\end{theorem}

For $\bmu=0$ this may be viewed as an analogue of Kadell's integral, and 
also follows by setting $\bmu=0$ in \eqref{Eq_thm92} below, a fact that was 
already noted by the second author in \cite[Remark 2]{Rains10}.
Imposing the constraint $t_4t_5=t$, the integral may be viewed as
an elliptic analogue of the Hua--Kadell integral. 
In this case the integral is invariant under the simultaneous 
substitutions
$\bla\leftrightarrow\bmu$ and $(t_1,t_2)\leftrightarrow(t_3,t_6)$.

By taking an appropriate $p\to 0$ limit, Theorem~\ref{Thm_eAFLT}
simplifies to the following AFLT integral over a pair of Macdonald
polynomials $P_{\la}(q,t)$.

\begin{corollary}\label{Cor_Mac-limit}
For $\la\in\mathscr{P}_n$, $\mu\in\mathscr{P}$,
and $a,b,q,t\in\mathbb{C}$ such that $\abs{b},\abs{q},\abs{t}<1$,
\begin{align*}
&\frac{1}{n!(2\pi\iup)^n} \Int_{\mathbb{T}^n} P_{\la}(z;q,t) 
P_{\mu}\Big(\Big[z+\frac{t-b}{1-t}\Big];q,t\Big) \\[-1mm]
&\qquad\qquad\quad \times 
\prod_{i=1}^n \frac{(a/z_i,qz_i/a;q)_{\infty}}
{(b/z_i,z_i;q)_{\infty}} 
\prod_{1\leq i<j\leq n} \frac{(z_i/z_j,z_j/z_i;q)_{\infty}}
{(tz_i/z_j,tz_j/z_i;q)_{\infty}} \, 
\frac{\dup z_1}{z_1}\cdots \frac{\dup z_n}{z_n} \\[1mm]
&\quad=b^{\abs{\la}} t^{\abs{\mu}} 
P_{\la}\Big(\Big[\frac{1-t^n}{1-t}\Big];q,t\Big)
P_{\mu}\Big(\Big[\frac{1-bt^{n-1}}{1-t}\Big];q,t\Big) \\[1mm]
&\quad\quad \times 
\prod_{i=1}^n 
\frac{(t,at^{n-m-i}q^{\la_i},at^{1-i}/b,qt^{i-1}b/a;q)_{\infty}}
{(q,t^i,bt^{i-1},at^{1-i}q^{\la_i}/b;q)_{\infty}}
\prod_{i=1}^n \prod_{j=1}^m
\frac{(at^{n-i-j+1}q^{\la_i+\mu_j};q)_{\infty}}
{(at^{n-i-j}q^{\la_i+\mu_j};q)_{\infty}},
\end{align*}
where $m$ is an arbitrary integer such that $m\geq l(\mu)$.
\end{corollary}

For $\la=\mu=0$ the above integral may also be found in
\cite[Section~6]{vdBR11} as a special case of the biorthogonality relation
for multivariable Pastro polynomials.
This special case also easily follows by combining the $_1\Psi_1$ summation
for Macdonald polynomials \cite{Kaneko98} with the orthogonality and
quadratic norm evaluation of the Macdonald polynomials with respect 
to the scalar product $\langle \cdot,\cdot\rangle'_n$, see 
\eqref{Eq_scalar-product}.

\bigskip
The remainder of this paper is organised as follows.
In the next section we review some basic material from the theory of
symmetric functions.
This includes a discussion of Schur, Jack and Macdonald polynomials
as well as the heavily-used plethystic notation.
In Section~\ref{Sec_Cauchy} we present some $\A_n$ generalisations of
the classical Cauchy identity for Macdonald polynomials. 
Then, in Section~\ref{Sec_AFLT}, we show that such Cauchy identities are
essentially discrete analogues of $\A_n$ AFLT integrals, leading to a
proof of Theorem~\ref{Thm_An-AFLT} and its companion given in 
Theorem~\ref{Thm_An-alt}.
In Section~\ref{Sec_An-AFLT-Schur} we use integral formulas of Cauchy-type
for complex Schur functions to prove Theorem~\ref{Thm_nplusone}.
In Section~\ref{Sec_E-Selberg} we review some of the theory of
elliptic interpolation functions and use this to prove the elliptic AFLT 
integral of Theorem~\ref{Thm_eAFLT}.
Finally, in Section~\ref{Sec_open}, we discuss some open problems stemming
from our work.

\section{Preliminaries}

\subsection{Partitions}

Throughout this paper $\mathbb{N}$ denotes the set of nonnegative integers.

A partition $\la$ is a weakly decreasing sequence of nonnegative
integers $(\la_1,\la_2,\dots)$ with only finitely many of the $\la_i$ nonzero.
The positive $\la_i$ are called the parts of $\la$, and the number of parts 
is called the length, denoted by $l(\la)$. 
The sum of the $\la_i$ is denoted by $\abs{\la}$, and if $\abs{\la}=n$
we say that $\la$ is a partition of $n$ and write $\la\vdash n$.
With the possible exception of finitely many zeros, we usually drop 
the infinite string of zeros of a partition so that
$\la=(\la_1,\dots,\la_n)$ denotes a partition of length at most $n$.
The set of all partitions and the set of partitions of length at most $n$
are denoted by $\mathscr{P}$ and $\mathscr{P}_n$ respectively.
In particular, $\mathscr{P}_0=\{0\}$ where, by mild abuse of notation, the
unique partition of zero is denoted by $0$.

We identify a partition with its Young diagram, which is the set of all
$(i,j)\in\mathbb{N}^2$ such that $1\leq i\leq l(\la)$ and $1\leq j\leq \la_i$. 
This may be visualised as a left-justified array of squares with $\la_i$ 
squares in row $i$.
For example the Young diagram of the partition $(6,4,3,1,1)$ is
\smallskip
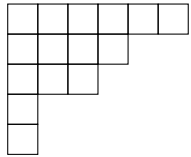
\begin{center}
\begin{tikzpicture}[scale=0.4]
\foreach \i [count=\ii] in {6,4,3,1,1}
\foreach \j in {1,...,\i}{\draw (\j,1-\ii) rectangle (\j+1,-\ii);}
\end{tikzpicture}
\end{center}
The conjugate of $\la$, which is denoted by $\la'$, is obtained by 
reflecting $\la$ in the main diagonal.
Hence $(6,4,3,1,1)'=(5,3,3,2,1,1)$. 
Given a square $s=(i,j)\in \mathbb{N}^2$ and a partition $\la$
we define the (generalised) arm-length, leg-length, arm-colength, 
and leg-colength of $s$ by
\begin{align}\label{Eq_arm}
a_{\la}(s)&:=\la_i-j, \hspace{-15mm} & l_{\la}(s)&:=\la'_j-i, \\
a'(s)&:=j-1, \hspace{-15mm} & l'(s)&:=i-1, \notag
\end{align}
respectively. 
If $s\in\la$ this reduces to the standard definition in
\cite[p.~337]{Macdonald95}.
Note that with the above notation the function $E$ in \eqref{Eq_E}
may also be written as
\[
E(u,\la,\mu,s)=u-b\, l_{\mu}(s)+b^{-1}(a_{\la}(s)+1),
\]
where $s\in\la$.

A frequently encountered statistic on partitions is
\[
n(\la):=\sum_{i\geq1}(i-1)\la_i=\sum_{i\geq1}\binom{\la_i'}{2}
=\sum_{s\in\la}l'(s).
\]

All of the previous notation regarding partitions is extended to
bipartitions $\bla\in\mathscr{P}^2$ in the obvious way. 
In particular, $\bmu\subseteq\bla$ will be used as shorthand for 
$\mur{1}\subseteq\lar{1}$ and $\mur{2}\subseteq\lar{2}$,
and $\boldsymbol{0}:=(0,0)$. 

\subsection{Generalised shifted factorials}

For $n$ a nonnegative integer and $b$ an indeterminate or 
complex number, the Pochhammer symbol $(b)_n$ is defined as
\[
(b)_n:=\prod_{i=0}^{n-1}(b+i),
\]
where an empty product is to be taken as $1$.
We generalise this to complex $z$ by
\begin{equation}\label{Eq_bz}
(b)_z:=\frac{\Gamma(b+z)}{\Gamma(b)},
\end{equation}
where now it is assumed that $b\in\mathbb{C}$ and neither $b$ nor
$b+z$ are nonpositive integers.
Similarly, for indeterminate or arbitrary complex $b$ and $q$,
the $q$-shifted factorial is defined as
\begin{equation}\label{Eq_q-Poch}
(b;q)_n:=\prod_{i=0}^{n-1}(1-bq^i),
\end{equation}
where $n\in\mathbb{N}\cup\{\infty\}$ and where
$\abs{q}<1$ in the infinite product case.
When $0<q<1$ this can again be extended to complex $z$
by
\[
(b;q)_z:=\frac{(b;q)_{\infty}}{(bq^z;q)_\infty}.
\]
This in particular implies that for $n$ a positive integer,
\begin{equation}\label{Eq_q-Poch-vanishing}
\frac{1}{(q;q)_{-n}}=0.
\end{equation}
In terms of the $q$-gamma function
\[
\Gamma_q(z):=\frac{(q;q)_{\infty}}{(q^z;q)_{\infty}}\,(1-q)^{1-z},
\]
and (for $z\in\mathbb{C}\setminus\{0,-1,-2,\dots\}$ and $0<q<1$),
\[
(q^b;q)_n=(1-q)^n \, \frac{\Gamma_q(b+n)}{\Gamma_q(b)}.
\]
A generalisation of \eqref{Eq_q-Poch} to partitions is given by
\begin{equation}\label{Eq_qt-Poch-def}
(b;q,t)_{\la}:=\prod_{s\in\la}\big(1-bq^{a'(s)}t^{-l'(s)}\big)
=\prod_{i\geq 1} (bt^{1-i};q)_{\la_i}.
\end{equation}
Setting $t=q^\gamma$ and replacing $b$ by $q^b$, and then letting $q$ tend
to $1$ we obtain an analogue of the Pochhammer symbol indexed by partitions 
\[
(b;\gamma)_{\la}:=\prod_{i\geq 1} \big(b+(1-i)\gamma\big)_{\la_i}.
\]
Also frequently used in this paper are the generalised hook polynomials
\begin{subequations}\label{Eq_hook}
\begin{align}
c_{\la}(q,t) &:= \prod_{s\in\la} \big(1-q^{a_{\la}(s)}t^{l_{\la}(s)+1}\big) 
=\prod_{i=1}^n(t^{n-i+1};q)_{\la_i}
\prod_{1\leq i<j\leq n}\frac{(t^{j-i};q)_{\la_i-\la_j}}
{(t^{j-i+1};q)_{\la_i-\la_j}}, \\
c'_{\la}(q,t)&:= \prod_{s\in\la}\big(1-q^{a_{\la}(s)+1}t^{l_{\la}(s)}\big)
=\prod_{i=1}^n(qt^{n-i};q)_{\la_i}
\prod_{1\leq i<j\leq n}\frac{(qt^{j-i-1};q)_{\la_i-\la_j}}
{(qt^{j-i};q)_{\la_i-\la_j}}, \\
b_{\la}(q,t)&:=\frac{c_{\la}(q,t)}{c'_{\la}(q,t)}. 
\end{align}
\end{subequations}
Note that the choice of $n$ on the right of the first two equations is 
irrelevant as long as $n\geq l(\la)$. 

Finally, at the top-level we have the elliptic shifted
factorials and gamma function.
To define the former we need the modified theta function
\[
\theta(z;p):=(z;p)_{\infty}(p/z;p)_{\infty},
\]
where $z\in\mathbb{C}^{\ast}$ and $p\in\mathbb{C}$ such that $\abs{p}<1$.
Then the elliptic shifted factorial is given by
\[
(b;q,p)_n:=\prod_{i=0}^{n-1}\theta(bq^i;p),
\]
so that $(b;q,0)_n=(b;q)_n$.
If $\Gamma(z;p,q)$ denotes the elliptic gamma function \cite{Ruijsenaars97}
\[
\Gamma(z;p,q):=\prod_{i,j=0}^{\infty}\frac{1-p^{i+1}q^{j+1}/z}{1-zp^iq^j},
\]
which satisfies the reflection formula $\Gamma(z;p,q)\Gamma(pq/z;p,q)=1$,
then, in analogy with \eqref{Eq_bz},
\[
(b;q,p)_n=\frac{\Gamma(bq^n;p,q)}{\Gamma(b;p,q)}.
\]
Similarly, the elliptic generalisation of \eqref{Eq_qt-Poch-def} is
\[
(b;q,t;p)_{\la}:=\prod_{s\in\la} \theta\big(bq^{a'(s)}t^{-l'(s)};p\big)=
\prod_{i\geq 1} (bt^{1-i};q,p)_{\la_i}.
\]
To shorten some of our expressions we also use the following shorthand 
for the ``well-poised'' ratio of products of elliptic shifted factorials:
\[
\Delta^0_{\la}(a\vert b_1,\dots,b_n;q,t;p)
:=\prod_{i=1}^n \frac{(b_i;q,t;p)_{\la}}{(pqa/b_i;q,t;p)_{\la}},
\]
from which it is clear that which it is clear that
\begin{equation}\label{Eq_Delta_sym}
\Delta^0_{\la}(a\vert b_1,\dots,b_n;q,t;p)=
\frac{1}{\Delta^0_{\la}(a\vert pqa/b_1,\dots,pqa/b_n;q,t;p)}.
\end{equation}
Finally we need 
\begin{align*}
C^{-}_{\la}(b;q,t;p)&:=
\prod_{s\in\la} \theta\big(bq^{a_{\la}(s)}t^{l_{\la}(s)};p\big) \\
C^{+}_{\la}(b;q,t;p)&:=
\prod_{(i,j)\in\la} \theta\big(bq^{\la_i+j-1}t^{2-\la'_j-i};p\big),
\end{align*}
so that $c_{\la}(q,t)=C^{-}_{\la}(t;q,t;0)$ and 
$c'_{\la}(q,t)=C^{-}_{\la}(q;q,t;0)$.

\subsection{Symmetric functions}

Let $X=\{x_1,x_2,\dots\}$ be an alphabet of countably many variables
and $X_n=\{x_1,\dots,x_n\}$ an alphabet of cardinality $n$.
Then we denote the ring of symmetric functions in $X$ (resp.\ $X_n$) 
over the field $\mathbb{F}$ by $\La$ (resp.\ $\La_n$), see \cite{Macdonald95}.
Typically, we will work with $\mathbb{F}=\mathbb{Q}$, or the extensions 
$\mathbb{Q}(\gamma)$ and $\mathbb{Q}(q,t)$. 

Given a sequence $\alpha=(\alpha_1,\alpha_2,\dots)$ of nonnegative
integers such that $\abs{\alpha}:=\alpha_1+\alpha_2+\cdots$ is finite, we
write $\alpha^{+}$ for the unique partition obtained by reordering the 
$\alpha_i$.
Then the monomial symmetric function indexed by the partition $\la$ is
defined as
\[
m_{\la}(X):=\sum_{\alpha^{+}=\la} X^{\alpha},
\]
where $X^\alpha:=x_1^{\alpha_1}x_2^{\alpha_2}\cdots$.
Further defining $m_{\la}(X_n):=m_{\la}(X)|_{x_i=0 \text{ for } i>n}$ it
follows that $m_{\la}(X_n)=0$ if $l(\la)>n$. 
The sets $\{m_{\la}(X)\}$ and $\{m_{\la}(X_n)\}_{l(\la)\leq n}$ form bases 
of $\La$ and $\La_n$ respectively.

For $k$ a nonnegative integer the $k$th complete and elementary symmetric
functions are defined by
\[
h_k(X):=\sum_{1\leq i_1\leq\cdots\leq i_k}x_{i_1}\cdots x_{i_k},
\]
and 
\[
e_k(X):=\sum_{1\leq i_1<\cdots<i_k} x_{i_1}\cdots x_{i_k},
\]
respectively.
For $k$ a positive integer we set $e_{-k}=h_{-k}=0$.
The generating functions for the complete symmetric functions is given by
\[
\sigma_z(X):=\sum_{k\geq 0} z^k h_k(X)=\prod_{i\geq 1} \frac{1}{1-zx_i}
\]
and
\begin{equation}\label{Eq_e-GF}
\sum_{k\geq 0} z^k e_k(X)=\prod_{i\geq 1} (1+zx_i)=\frac{1}{\sigma_{-z}(X)}.
\end{equation}
We also require Newton power sums
\[
p_k(X):=\sum_{i\geq 1}x_i^k,
\]
for $k$ a positive integer. 
These admit the generating function
\[
\psi_z(X):=\sum_{k\geq 1} \frac{z^k p_k(X)}{k}=
-\sum_{i\geq 1}\log(1-zx_i),
\]
so that
\begin{equation}\label{Eq_GF-rel}
\sigma_z(X)=\eup^{\psi_z(X)}.
\end{equation}

The most important family of classical symmetric functions are the 
Schur functions, defined by the usual ratio of alternants
\begin{equation}\label{Eq_Schur-def}
s_{\la}(X_n)=
\frac{\det_{1\leq i,j\leq n}\big(x_i^{\la_j+n-j}\big)}{\Delta(X_n)},
\end{equation}
for $l(\la)\leq n$ and $0$ otherwise.
The Schur functions indexed by partitions of length at most $n$
form a basis of $\Lambda_n$.
From \eqref{Eq_Schur-def} it is not hard to derive the specialisation
formula \cite[p.~44]{Macdonald95}
\begin{equation}\label{Eq_Schur-spec}
s_{\la}(\underbrace{1,\dots,1}_{\text{$n$ times}}) = 
\prod_{i=1}^k \frac{(n-i+1)_{\la_i}}{(k-i+1)_{\la_i}}
\prod_{1\leq i<j\leq k}\frac{\la_i-\la_j+j-i}{j-i},
\end{equation}
where $k$ is an arbitrary integer such that $k\geq l(\la)$.

\subsection{Plethystic notation}\label{Sec_plethystic}

We extensively use plethystic or $\la$-ring notation when dealing with 
symmetric functions, see e.g., \cite{Haglund08,Lascoux03,RW18}. 
For $X=\{x_1,x_2,\dots\}$ an alphabet and $f(X)$ a symmetric function
in $X$ we use the additive notation
\[
f(x_1,x_2,\dots)=f(X)=f[X]=f[x_1+x_2+\cdots].
\]
Hence $X+Y$, the sum of the alphabets $X$ and $Y$, is the disjoint union
of these sets.
The above notation forces
\begin{equation}\label{Eq_addition}
p_k[X+Y]=p_k[X]+p_k[Y].
\end{equation}
A symmetric function acting on the difference of two alphabets is then
defined as
\[
p_k[X-Y]:=p_k[X]-p_k[Y].
\]
Observe that
\[
p_k[(X+Y)-Y]=p_k[X]
\]
as it should.
Inside plethystic brackets we denote the empty alphabet by $0$.
By \eqref{Eq_GF-rel} it follows that
\begin{equation}\label{Eq_sigma-sum}
\sigma_z[X+Y]=\sigma_z[X]\sigma_z[Y]
\quad\text{and}\quad
\sigma_z[X-Y]=\frac{\sigma_z[X]}{\sigma_z[Y]},
\end{equation}
and hence $\sigma_z[-X]\sigma_z[X]=1$. 
Together with \eqref{Eq_e-GF} this implies
\[
h_k[-X]=(-1)^k e_k[X],
\]
which, by the dual Jacobi--Trudi identity \cite[p.~41]{Macdonald95},
extends to Schur functions as \eqref{Eq_Schur-duality}.

For the Cartesian product of two alphabets we have
\[
p_k[XY]=p_k[X]p_k[Y],
\]
which in particular implies that if $X=x$ is an alphabet containing a
single letter then
\[
p_k[xY]=x^k p_k[Y].
\]
However, from \eqref{Eq_addition} we also have
\[
p_k[nX]:=
p_k[\underbrace{X+\dots+X}_{\text{$n$ times}}]=np_k[X],
\]
so that
\[
f[n]=f[\underbrace{1+\dots+1}_{\text{$n$ times}}]=
f(\underbrace{1,\dots,1}_{\text{$n$ times}}).
\]
We extend the above to any $z\in\mathbb{F}$ via 
\[
p_k[zX]:=zp_k[X].
\]
For $X=1$ we more simply write the above left-hand side as 
$p_k[z]$.\footnote{In \cite[p.~32]{Lascoux03} Lascoux refers to a single
letter alphabet $x$ as a rank-$1$ element of a $\la$-ring and
$z\in\mathbb{F}$ as a binomial element, since for the latter 
$e_k[z]=\binom{z}{k}$ and $h_k[z]=\binom{z+k-1}{k}$.}
Note that this leads to some notational ambiguities, and whenever not
clear from the context we will indicate if a symbol such as $x$ or $z$
represents a letter or a binomial element.

Finally, since the Cartesian product of the alphabets
$\{1,t,t^2,\dots\}=1+t+t^2+\cdots$ and $1-t$ is $1$, we adopt the standard
convention of writing the former as $1/(1-t)$.
An often occurring composite alphabet  is $(a-b)/(1-t)$ for which
\[
p_k\bigg[\frac{a-b}{1-t}\bigg]=\frac{a^k-b^k}{1-t^k}
\]
and
\begin{equation}\label{Eq_sigma-kernel}
\sigma_1\bigg[\frac{a-b}{1-t}\bigg]
=\frac{(b;t)_{\infty}}{(a;t)_{\infty}}.
\end{equation}

\subsection{Macdonald polynomials}

In this section we work with the ring of symmetric functions $\La$ over
$\mathbb{F}=\mathbb{Q}(q,t)$.

Let $\langle \cdot,\cdot\rangle:\La\times\La\to\mathbb{F}$ be 
the $q,t$-Hall scalar product on $\La$ given by \cite[p.~306]{Macdonald95}
\[
\langle p_{\la},p_{\mu}\rangle:=\delta_{\la,\mu}z_\la
\prod_{i\geq 1}\frac{1-q^{\la_i}}{1-t^{\la_i}},
\]
where $z_\la:=\prod_{i\geq 1}i^{m_i}m_i!$.
The Macdonald polynomials $P_{\la}=P_{\la}(q,t)=P_{\la}(X;q,t)$ are the 
unique symmetric functions such that
\[
\langle P_{\la},P_{\mu}\rangle=0\quad\text{if}\quad\la\neq\mu
\]
and
\[
P_{\la}=m_{\la}+\sum_{\mu<\la}u_{\la\mu} m_{\mu},
\quad u_{\la\mu}\in\mathbb{F},
\]
with respect to the usual dominance order on partitions.
Like the Schur functions, $\{P_{\la}(X)\}$ and
$\{P_{\la}(X_n)\}_{l(\la)\leq n}$ are bases for $\La$ and $\La_n$
respectively, and $P_{\la}(X_n;q,t)=0$ if $l(\la)>n$. 

We also require the skew Macdonald polynomials defined by
\begin{equation}\label{Eq_Mac-skew}
P_{\la}([X+Y];q,t)=\sum_{\mu}P_{\la/\mu}(X;q,t)P_{\mu}(Y;q,t),
\end{equation}
where it should be noted that $P_{\la/0}=P_{\la}$. 
We also note that $P_{\la/\mu}$ is homogeneous of degree $\abs{\la/\mu}$
and vanishes unless $\mu\subseteq\la$ (cf.~\cite[p.~344]{Macdonald95}).

For $q=t$ the Macdonald polynomials reduce to the Schur functions
\eqref{Eq_Schur-def}.
More generally, if we set $t=q^\gamma$ and let $q$ tend to $1$ we obtain
the Jack polynomials~\cite{Stanley89}
\[
P_{\la}^{(1/\gamma)}(X)=\lim_{q\to 1}P_{\la}(X;q,q^\gamma).
\]

Setting $Q_{\la/\mu}(q,t):=b_\la(q,t)P_{\la/\mu}(q,t)/b_{\mu}(q,t)$,
the Macdonald polynomials satisfy the (skew) Cauchy identity
\cite[p.~345]{Macdonald95}
\begin{align}\label{Eq_Cauchy}
\sum_{\la} P_{\la}(X;q,t)Q_{\la/\mu}(Y;q,t)&=P_{\mu}(X;q,t)\prod_{i,j\geq 1}
\frac{(tx_iy_j;q)_{\infty}}{(x_iy_j;q)_{\infty}} \\
&=P_{\mu}(X;q,t)\sigma_1\bigg[\frac{1-t}{1-q}\,XY\bigg], \notag
\end{align}
where the second equality follows from \eqref{Eq_sigma-sum} and
\eqref{Eq_sigma-kernel}.

Let 
\[
{_2\phi_1}\bigg[\genfrac{}{}{0pt}{}{a,b}{c};q,z\bigg]
:=\sum_{k=0}^{\infty} \frac{(a;q)_k(b;q)_k}{(c;q)_k(q;q)_k}\, z^k
\]
denote the usual $q$-analogue of the $_2F_1$ Gauss hypergeometric 
function \cite{GR04}. 

\begin{lemma}
Let $x$ and $y$ be single-letter alphabets.
Then
\begin{equation}\label{Eq_Prxminy}
P_{(r)}([x-y];q,t)=x^r \, {_2\phi_1}
\bigg[\genfrac{}{}{0pt}{}{t^{-1},q^{-r}}{q^{1-r}t^{-1}};q,\frac{yq}{x}\bigg].
\end{equation}
\end{lemma}

\begin{proof}
If we set $\mu=0$ in \eqref{Eq_Cauchy} and then replace 
$(X,Y)\mapsto (x-y,1)$ this yields
\[
\sum_{r\geq 0} \frac{(t;q)_r}{(t;t)_r}\,  P_{(r)}([x-y];q,t)
=\sigma_1\bigg[\frac{1-t}{1-q}\,(x-y)\bigg]
=\frac{(tx;q)_{\infty}(y;q)_{\infty}}{(x;q)_{\infty}(ty;q)_{\infty}}.
\]
Using the $q$-binomial theorem \cite[Equation~(II.3)]{GR04} to expand
the right-hand side as a power series in $x$ and $y$ leads to 
\[
\sum_{r\geq 0} \frac{(t;q)_r}{(t;t)_r}\,P_{(r)}([x-y];q,t)
=\sum_{k,\ell\geq 0} \frac{(t;q)_{\ell}(t^{-1};q)_k}
{(q;q)_{\ell}(q;q)_k}\, x^{\ell} (ty)^k.
\]
Equating terms of homogeneous degree $r$ in $x,y$ gives
\[
\frac{(t;q)_r}{(t;t)_r}\, P_{(r)}([x-y];q,t)
=\sum_{k=0}^r \frac{(t;q)_{r-k}(t^{-1};q)_k}
{(q;q)_{r-k}(q;q)_k} \, x^{r-k} (ty)^k,
\]
which is equivalent to \eqref{Eq_Prxminy}.
\end{proof}

For $\la\in\mathscr{P}_n$ define the spectral vector
\[
\spec{\la}_n=\spec{\la}_{n;q,t}:=(q^{\la_1}t^{n-1},q^{\la_2}t^{n-2},\dots,
q^{\la_{n-1}}t,q^{\la_n}t^0),
\]
which, depending on the context, we will also interpret plethystically as
\[
\spec{\la}_n=q^{\la_1}t^{n-1}+q^{\la_2}t^{n-2}+\dots+q^{\la_{n-1}}t+
q^{\la_n}t^0.
\]
Then the principal specialisation formula for Macdonald polynomials
\cite[p.~337]{Macdonald95} can be written as 
\begin{equation}\label{Eq_princ-spec}
P_{\la}[\spec{0}_n]=P_{\la}\bigg[\frac{1-t^n}{1-t}\bigg]
=\frac{t^{n(\la)}(t^n;q,t)_\la}{c_{\la}(q,t)},
\end{equation}
where from hereon we mostly suppress the dependence of the Macdonald
polynomials on $q$ and $t$.
By a polynomial argument \eqref{Eq_princ-spec} is equivalent to 
\cite[p.~338]{Macdonald95}
\begin{equation}\label{Eq_a-spec}
P_{\la}\bigg[\frac{1-a}{1-t}\bigg]=
\frac{t^{n(\la)}(a;q,t)_{\la}}{c_{\la}(q,t)}.
\end{equation}

Replacing $(a,t)\mapsto (q^{z \gamma},q^{\gamma})$ and letting
$q$ tend to one in \eqref{Eq_a-spec} yields the following expression
for the Jack polynomial evaluated at a binomial element $z$:
\begin{equation}\label{Eq_Jack-z}
P_{\la}^{(1/\gamma)}[z]=
\frac{(z \gamma;\gamma)_{\la}}{(k\gamma;\gamma)_{\la}}
\prod_{1\leq i<j\leq k}\frac{((j-i+1)\gamma)_{\la_i-\la_j}}
{((j-i)\gamma)_{\la_i-\la_j}},
\end{equation}
where $k$ is an arbitrary integer such that $k\geq l(\la)$.
For $\gamma=1$ and $z=n$, with $n$ a nonnegative integer
this reduces to \eqref{Eq_Schur-spec}.

The evaluation symmetry of the Macdonald polynomials 
\cite[p.~332]{Macdonald95} may expressed in terms of spectral vectors as 
\begin{equation}\label{Eq_eval-sym}
P_{\mu}[\spec{0}_n]P_{\la}[\spec{\mu}_n]=
P_{\la}[\spec{0}_n]P_{\mu}[\spec{\la}_n],
\end{equation}
where $\la,\mu\in\mathscr{P}_n$.
We require a more general form of this symmetry, which is a nonsymmetric 
version of \cite[Proposition~2.1]{Warnaar10}.

\begin{lemma}\label{Lem_gen-eval-sym}
For $\la\in\mathscr{P}_n$ and $\mu\in\mathscr{P}_m$,
\begin{equation}\label{Eq_gen-eval-sym}
P_{\mu}\bigg[\frac{1-a}{1-t}\bigg]P_{\la}\bigg[at^{-m}\spec{\mu}_m
+\frac{1-at^{-m}}{1-t}\bigg]
=P_{\la}\bigg[\frac{1-a}{1-t}\bigg]P_{\mu}\bigg[at^{-n}\spec{\la}_n
+\frac{1-at^{-n}}{1-t}\bigg].
\end{equation}
\end{lemma}

\begin{proof}
For $m=n$ and $a\mapsto at^n$ this is \cite[Proposition~2.1]{Warnaar10}
\begin{equation}\label{Eq_gen-princ-spec-proof}
P_{\mu}\bigg[\frac{1-at^n}{1-t}\bigg]
P_{\la}\bigg[a\spec{\mu}_n+\frac{1-a}{1-t}\bigg]=
P_{\la}\bigg[\frac{1-at^n}{1-t}\bigg]
P_{\mu}\bigg[a\spec{\la}_n+\frac{1-a}{1-t}\bigg].
\end{equation}
Fixing an integer $m$ such that $l(\mu)\leq m\leq n$ we have
\[
a\spec{\mu}_n=at^{n-m}\spec{\mu}_m+\frac{a-at^{n-m}}{1-t}.
\]
Therefore \eqref{Eq_gen-princ-spec-proof} becomes
\[
P_{\mu}\bigg[\frac{1-at^n}{1-t}\bigg]
P_{\la}\bigg[at^{n-m}\spec{\mu}_m+\frac{1-at^{n-m}}{1-t}\bigg]=
P_{\la}\bigg[\frac{1-at^n}{1-t}\bigg]
P_{\mu}\bigg[a\spec{\la}_n+\frac{1-a}{1-t}\bigg].
\]
Scaling $a\mapsto at^{-n}$ results in \eqref{Eq_gen-eval-sym}.
Since this is symmetric in $m$ and $n$ the restriction $m\leq n$ may
be dropped.
\end{proof}

\section{Cauchy-type identities}\label{Sec_Cauchy}

An important special case of the Cauchy identity \eqref{Eq_Cauchy} ---
obtained by taking $\mu=0$, making the plethystic substitution
$Y\mapsto (1-a)/(1-t)$ and using the specialisation formula
\eqref{Eq_a-spec} --- is the Kaneko--Macdonald $q$-binomial theorem 
\cite{Kaneko96,Macdonald13}
\[
\sum_{\la}\frac{t^{n(\la)}(a;q,t)_{\la}}{c_{\la}'(q,t)}\,P_{\la}(X;q,t)
=\prod_{i\geq 1}\frac{(ax_i;q)_\infty}{(x_i;q)_\infty}.
\]
An $\A_n$ analogue of this formula is given in
\cite[Theorem~3.2]{Warnaar09} and then applied to prove the $\nu=0$ case of 
Theorem~\ref{Thm_An-AFLT}.
In order to prove Theorem~\ref{Thm_An-AFLT} in full we require an 
$\A_n$ Cauchy-type identity which simultaneously generalises 
\eqref{Eq_Cauchy} and \cite[Theorem~3.2]{Warnaar09}.
This is the content of Theorem~\ref{Thm_An-Cauchy} below.
 
\subsection{Identities for skew Macdonald polynomials}

For partitions $\la,\mu$ and $k\in\mathbb{N}$, 
$\ell\in\mathbb{N}\cup\{\infty\}$ such that
$k\geq l(\la)$ and $\ell\geq l(\mu)$, define \cite{Warnaar10}\footnote{For 
the relationship between $f_{\la,\mu}^{k,\ell}(a;q,t)$ and Nekrasov-type
functions, see \cite[Equation (B.4)]{IMP16}.}
\[
f_{\la,\mu}^{k,\ell}(a;q,t)
:=t^{-k\abs{\mu}} \prod_{i=1}^k\prod_{j=1}^\ell 
\frac{(aqt^{j-i-1};q)_{\la_i-\mu_j}}{(aqt^{j-i};q)_{\la_i-\mu_j}}.
\]
By $(a;q)_{-n}/(b;q)_{-n}=(b/a)^n (q/b;q)_n/(q/a;q)_n$
it follows that
\[
f_{\la,\mu}^{k,\ell}(a;q,t)=f_{\mu,\la}^{\ell,k}(t/aq;q,t)
\]
provided $\ell$ is finite.
For infinite $\ell$, let $t^{\ell}:=0$.
Our proof of Theorem~\ref{Thm_An-Cauchy}, given in 
Section~\ref{Sec_Cauchy-type-proof}, hinges on the following summation
formula for skew Macdonald polynomials \cite[Theorem~3.4]{Warnaar10}.

\begin{proposition}\label{Prop_skew-sum}
For partitions $\la$ and $\mu$, we have
\begin{equation}\label{Eq_skew-sum}
\sum_{\nu}t^{-\abs{\nu}}P_{\mu/\nu}\bigg[\frac{1-1/a}{1-t}\bigg]
Q_{\la/\nu}\bigg[\frac{1-aq/t}{1-t}\bigg]
=P_{\mu}\bigg[\frac{1-t^k/a}{1-t}\bigg]
Q_\la\bigg[\frac{1-aqt^{\ell-1}}{1-t}\bigg]
f_{\la,\mu}^{k,\ell}(a;q,t),
\end{equation}
where $k\in\mathbb{N}$ and $\ell\in\mathbb{N}\cup\{\infty\}$ may be chosen
arbitrarily, provided that $k\geq l(\la)$ and $\ell\geq l(\mu)$.
\end{proposition}

In \cite{Warnaar10} the right-hand side of \eqref{Eq_skew-sum} is stated
with $\ell=k$.
Of course, since the left-hand side does not depend on $k$ and $\ell$,
the above form of the identity is not actually more general.
Indeed,
\begin{align*}
&f_{\la,\mu}^{k,\ell}(a;q,t) \\
&\quad=t^{-(k-l(\la))\abs{\mu}} f_{\la,\mu}^{l(\la),l(\mu)}(a;q,t)
\prod_{i=1}^{l(\la)}\prod_{j=l(\mu)+1}^{\ell}
\frac{(aqt^{j-i-1};q)_{\la_i}}{(aqt^{j-i};q)_{\la_i}}
\prod_{i=l(\la)+1}^k \prod_{j=1}^{l(\mu)}
\frac{(aqt^{j-i-1};q)_{-\mu_j}}{(aqt^{j-i};q)_{-\mu_j}} \\
&\quad=f_{\la,\mu}^{l(\la),l(\mu)}(a;q,t)\,
\frac{(aqt^{l(\mu)-1};q,t)_{\la}}{(aqt^{\ell-1};q,t)_{\la}}\,
\frac{(t^{l(\la)}/a;q,t)_{\mu}}{(t^k/a;q,t)_{\mu}}.
\end{align*}
Substituting this in the right-hand side of \eqref{Eq_skew-sum} and
using \eqref{Eq_a-spec} yields the identity with $k$ and $\ell$
replaced by $l(\la)$ and $l(\mu)$.
For later use we note that from the above and \eqref{Eq_a-spec} it
follows that 
\begin{equation}\label{Eq_fQfQ}
f_{\la,\mu}^{k,\infty}(a;q,t) 
Q_{\la}\bigg[\frac{1}{1-t}\bigg]
=f_{\la,\mu}^{k,\ell}(a;q,t)
Q_{\la}\bigg[\frac{1-aqt^{\ell-1}}{1-t}\bigg],
\end{equation}
where $\ell$ is an arbitrary integer such that $\ell\geq l(\mu)$. 

For our next result we would like to specialise $a=t^{k-\ell}$
(for $k\leq \ell$ and $\ell$ finite) in \eqref{Eq_skew-sum}.
Potentially this could lead to problems with the double product on the
right, and the following lemma serves to show that such a specialisation
is in fact permitted provided the resulting double product is interpreted
correctly.

\begin{lemma}\label{Lem_f-limit}
Let $\la$ and $\mu$ be partitions and $k,\ell\in\mathbb{N}$ such that
$k\leq \ell$ and such that $k\geq l(\la)$ and $\ell\geq l(\mu)$.
Then
\begin{equation}\label{Eq_f-limit}
\lim_{b\to 1} f_{\la,\mu}^{k,\ell}(bt^{k-\ell};q,t)
\end{equation}
is well-defined. Furthermore a necessary and sufficient condition for
the nonvanishing of this limit is
\begin{equation}\label{Eq_inequalities}
\la_i\geq\mu_{i-k+\ell} \quad \text{for $1\leq i\leq k$}.
\end{equation}
\end{lemma}

The inequalities \eqref{Eq_inequalities} may conveniently be visualised as:
\begin{equation}\label{Eq_inequalities-2}
\begin{array}{ccccccccccccccccl}
&&&&\la_1&\GEQ &\la_2 &\GEQ&\cdots&\GEQ&\la_k&\GEQ&\la_{k+1}&\GEQ 
&\cdots & \GEQ & 0 \\[-1mm]
&&&&\text{\begin{turn}{270}$\geq$\end{turn}}&&
\text{\begin{turn}{270}$\geq$\end{turn}} &&&&
\text{\begin{turn}{270}$\geq$\end{turn}} &&& \\[3mm]
\mu_1&\GEQ&\cdots&\GEQ &\mu_{\ell-k+1}&\GEQ &\mu_{\ell-k+2}&\GEQ&\cdots&\GEQ
&\mu_{\ell}&\GEQ & \mu_{\ell+1} & \GEQ & \cdots & \GEQ & 0.
\end{array}
\end{equation}

\begin{remark}
It is assumed in Lemma~\ref{Lem_f-limit} that $q$ and $t$ are generic.
For the Schur case $q=t$ the equation \eqref{Eq_inequalities} has to be
replaced by
\[
\lambda_i = \mu_{j_i}+i-j_i+\ell-k
\quad \text{for $1\leq i\leq k$},
\]
where $1\leq j_1<j_2<\dots<j_k\leq\ell$.
\end{remark}

\begin{proof}
To see that the limit is well-defined, note that for fixed
$i$ the powers of $t$ in \eqref{Eq_f-limit} are zero when 
$j=i-k+\ell+1$ in the numerator and $j=i-k+\ell$ in the denominator.
Since $j\leq \ell$ this yields $k-\ell\leq i\leq k-1$ for the numerator
and $k-\ell+1\leq i\leq k$ for the denominator, with both the lower bounds
automatically satisfied since $k\leq\ell$.
Therefore, taking the product of the $t$-independent $q$-shifted factorials 
in \eqref{Eq_f-limit} and making a shift in the indices yields
\begin{equation}\label{Eq_f-lim-proof}
\frac{\prod_{i=1}^{k-1}(bq;q)_{\la_i-\mu_{i-k+\ell+1}}}
{\prod_{i=1}^k(bq;q)_{\la_i-\mu_{i-k+\ell}}}
=\frac{1}{(bq;q)_{\la_k-\mu_\ell}}\prod_{i=\ell-k+1}^{\ell-1}
{(bq^{1+\la_{i+\ell-k}-\mu_{i}};q)_{\mu_{i}-\mu_{i+1}}}.
\end{equation}
Since $\mu$ is a partition, $\mu_i\geq\mu_{i+1}$, and hence the limit 
$b\to 1$ exists.

The vanishing of the limit \eqref{Eq_f-limit} is completely determined by
the vanishing of the right-hand side of \eqref{Eq_f-lim-proof} when $b\to 1$.
Clearly the term $1/(q;q)_{\la_k-\mu_\ell}$ will vanish unless
$\la_k\geq\mu_\ell$ by \eqref{Eq_q-Poch-vanishing}.
In order for \eqref{Eq_f-lim-proof} to be nonvanishing, one of 
\begin{subequations}
\begin{align}\label{Eq_f-lim-proof-1}
\la_{i+k-\ell}&\geq\mu_i, \\
\la_{i+k-\ell}&<\mu_i=\mu_{i+1},
\label{Eq_f-lim-proof-2}
\end{align}
\end{subequations}
must hold for each $i$ such that $\ell-k+1\leq i\leq \ell-1$.
Now assume that $\la_k\geq\mu_\ell$ and one of \eqref{Eq_f-lim-proof-1}
and \eqref{Eq_f-lim-proof-2} holds.
Consider the largest $i$ for which \eqref{Eq_f-lim-proof-2} holds but
\eqref{Eq_f-lim-proof-1} does not.
We cannot have $i=\ell-1$ as this would imply
\[
\la_k<\mu_{\ell-1}=\mu_\ell
\]
contradicting $\la_k\geq\mu_\ell$.
Similarly, no such maximal $i$ exists with $i\leq \ell-1$ as we then would get
\[
\la_{i+k-\ell}<\mu_i=\mu_{i+1}.
\]
However, as \eqref{Eq_f-lim-proof-1} must now hold with $i\mapsto i+1$, 
this would give $\la_{i+k-\ell}<\la_{i+k-\ell+1}$, contradicting that
$\la$ is a partition. 
Therefore we conclude that \eqref{Eq_f-lim-proof-1} must hold for all
$\ell-k+1\leq i\leq \ell-1$. 
This is equivalent to the desired conditions by a shift of indices,
and hence we are done. 
\end{proof}

By abuse of notation, we will write $f_{\la,\mu}^{k,\ell}(t^{k-\ell};q,t)$
instead of $\lim_{b\to 1} f_{\la,\mu}^{k,\ell}(bt^{k-\ell};q,t)$
in the following.

\begin{corollary}\label{Cor_skew-sum}
Let $k,\ell\in\mathbb{N}$ such that $k\leq\ell$.
Then, for partitions $\la,\mu$ such that $l(\la)\leq k$,
\[
\sum_{\nu}t^{-\abs{\nu}}P_{\mu/\nu}\bigg[\frac{1-t^{\ell-k}}{1-t}\bigg]
Q_{\la/\nu}\bigg[\frac{1-qt^{k-\ell-1}}{1-t}\bigg]
=P_{\mu}\bigg[\frac{1-t^{\ell}}{1-t}\bigg]
Q_{\la}\bigg[\frac{1-qt^{k-1}}{1-t}\bigg]
f_{\la,\mu}^{k,\ell}(t^{k-\ell};q,t).
\]
\end{corollary}

The above corollary is essentially \cite[Theorem~4.1, $u=0$]{Warnaar08}.
It should be noted that the condition $l(\mu)\leq \ell$ has been dropped
in comparison with Proposition~\ref{Prop_skew-sum} and
\cite[Theorem~4.1]{Warnaar08}, since both sides identically vanish when
$l(\mu)>\ell$. 
To see this, note that the summand vanishes unless
$\nu\subseteq\la$, $\nu\subseteq\mu$ and $\mu_{i-k+\ell}\leq\nu_i$ for all
$i\geq 1$.
This in particular implies that the summand vanishes unless
$\mu_{i-k+\ell}\leq \la_i$ for all $i\geq 1$, in accordance 
with \eqref{Eq_inequalities}.
When $i=k+1$ this yields $\mu_{\ell+1}\leq\la_{k+1}$.
Now, since $l(\la)\leq k$, $\la_{k+1}=0$ so that $\mu_{\ell+1}=0$,
i.e., $l(\mu)\leq \ell$.
Clearly, the right-hand side has this same vanishing property.

\subsection{$\boldsymbol{\A_n}$ Cauchy-type identities}

As we will see below, the Cauchy identity for Macdonald polynomials may 
be viewed as a discrete analogue of the AFLT integral.
For the purpose of generalisation it is convenient to think of the Cauchy
identity \eqref{Eq_Cauchy} as an identity for the root system $\A_1$ in 
which the two alphabets $X$ and $Y$ are attached to the single vertex of 
the corresponding Dynkin diagram:
\medskip
\begin{center}
\begin{tikzpicture}[scale=1.1]
\draw[fill=blue] (0,0) circle (0.05cm);
\draw (0,0.4) node {$X$};
\draw (0,-0.4) node {$Y$};
\end{tikzpicture}
\end{center}
Extending this to $\A_n$, we consider sums of the form
\begin{equation}\label{Eq_An-Cauchy-lhs}
\sum_{\lar{1},\dots,\lar{n}}\prod_{r=1}^n
P_{\lar{r}}\big[X^{(r)}\big]Q_{\lar{r}}\big[Y^{(r)}\big]
\prod_{r=1}^{n-1}f_{\lar{r},\lar{r+1}}^{k_r,k_{r+1}}(a_r;q,t),
\end{equation}
where the functions $f_{\lar{r},\lar{r+1}}^{k_r,k_{r+1}}$ represent
the edges of the $\A_n$ Dynkin diagram:
\medskip
\begin{center}
\begin{tikzpicture}[scale=1.1]
\draw (0,0)--(3,0); \draw (4,0)--(6,0); \draw (7,0)--(10,0);
\draw[dashed] (3,0)--(4,0); \draw[dashed] (6,0)--(7,0);
\draw[fill=blue] (0,0) circle (0.05cm);
\draw[fill=blue] (2,0) circle (0.05cm); 
\draw[fill=blue] (5,0) circle (0.05cm);
\draw[fill=blue] (8,0) circle (0.05cm);
\draw[fill=blue] (10,0) circle (0.05cm);
\draw (0,0.4) node {$X^{(1)}$};
\draw (2,0.4) node {$X^{(2)}$};
\draw (5,0.4) node {$X^{(r-1)}$};
\draw (8,0.4) node {$X^{(n-1)}$};
\draw (10,0.4) node {$X^{(n)}$};
\draw (0,-0.4) node {$Y^{(1)}$};
\draw (2,-0.4) node {$Y^{(2)}$};
\draw (5,-0.4) node {$Y^{(r-1)}$};
\draw (8,-0.4) node {$Y^{(n-1)}$};
\draw (10,-0.4) node {$Y^{(n)}$};
\end{tikzpicture}
\end{center}
In \eqref{Eq_An-Cauchy-lhs} we choose $k_1\leq k_2\leq\cdots\leq k_{n-1}$
to be nonnegative integers and $k_n\in\mathbb{N}\cup\{\infty\}$. 
Also, $a_r$ for $1\leq r\leq n-2$ will be fixed as $a_r=t^{k_r-k_{r+1}}$,
whereas $a_{n-1}$ is an indeterminate.
In the sum \eqref{Eq_An-Cauchy-lhs} we also specialise the alphabets 
$X^{(2)},\dots,X^{(n)}$ and $Y^{(1)},\dots,Y^{(n-1)}$ as
\begin{equation}\label{Eq_XY}
X^{(r+1)}=\frac{1-t^{k_r}/a_r}{1-t},\qquad
Y^{(r)}=z_r \frac{t-a_rqt^{k_{r+1}}}{1-t}\qquad\text{for $1\leq r\leq n-1$},
\end{equation}
and fix the cardinalities of $X^{(1)}$ and $Y^{(n)}$ to be 
\[
\abs{X^{(1)}}=k_1,\qquad \abs{Y^{(n)}}=k_n.
\]
It should be noted that since $a_r=t^{k_r-k_{r+1}}$ for $r\neq n-1$ 
we have
\begin{equation}\label{Eq_Xr-finite}
X^{(r)}=\frac{1-t^{k_r}}{1-t} \qquad \text{for $2\leq r\leq n-1$}
\end{equation}
so that $\abs{X^{(r)}}=k_r$ for all $1\leq r\leq n-1$.

Recall the convention that $t^k:=0$ if $k=\infty$.
Our first $\A_n$ Cauchy-type formula may then be stated as follows.  

\begin{theorem}[$\A_n$ Cauchy-type formula I]\label{Thm_An-Cauchy}
Let $k_1\leq k_2\leq\cdots\leq k_{n-1}$ be nonnegative integers and
$k_n\in\mathbb{N}\cup\{\infty\}$.
Then for $a_{n-1}$ an indeterminate, $a_r:=t^{k_r-k_{r+1}}$ for 
$1\leq r\leq n-2$, $X^{(1)}=x_1+\dots+x_{k_1}$,
$Y^{(n)}=y_1+\dots+y_{k_n}$ and
$X^{(2)},\dots,X^{(n)},Y^{(1)},\dots,Y^{(n-1)}$ as in \eqref{Eq_XY}, and
\[
W:=z_1\cdots z_{n-1}X^{(1)} +
\sum_{r=1}^{n-1}z_{r+1}\cdots z_{n-1}\,\frac{1-1/a_r}{1-t},
\]
we have
\begin{align}\label{Eq_An-Cauchy-1}
&\sum_{\lar{1},\dots,\lar{n}} 
\prod_{r=1}^n P_{\lar{r}}\big(\big[X^{(r)}\big];q,t\big)
Q_{\lar{r}/\mur{r}}\big(\big[Y^{(r)}\big];q,t\big)
\prod_{r=1}^{n-1} f_{\lar{r},\lar{r+1}}^{k_r,k_{r+1}}(a_r;q,t) \\
&\quad\quad=P_{\mur{n}}\big([W];q,t\big)\prod_{r=1}^{n-1}\bigg(
\prod_{i=1}^{k_1}\frac{(a_rqz_1\cdots z_r x_i;q)_{\infty}}
{(tz_1\cdots z_r x_i;q)_{\infty}} 
\prod_{j=1}^{k_n} \frac{(z_{r+1}\cdots z_{n-1}y_j/a_r;q)_{\infty}}
{(z_{r+1}\cdots z_{n-1}y_j;q)_{\infty}}\bigg) \notag \\
&\quad\quad\quad\times \prod_{i=1}^{k_1} \prod_{j=1}^{k_n} 
\frac{(tz_1\cdots z_{n-1}x_iy_j;q)_{\infty}}
{(z_1\cdots z_{n-1}x_iy_j;q)_{\infty}}
\prod_{1\leq r<s\leq n-1} \prod_{i=1}^{k_{r+1}-k_r} 
\frac{(a_sqt^{i-1}z_{r+1}\cdots z_s;q)_{\infty}}
{(t^iz_{r+1}\cdots z_s;q)_{\infty}}, \notag
\end{align}
where $\mur{1},\dots,\mur{n-1}:=0$ and $\mur{n}$ is an arbitrary partition.
\end{theorem}

For $n=1$ the theorem reduces to \eqref{Eq_Cauchy} with
$(X,Y,\mu)\mapsto (X^{(1)},Y^{(1)},\mur{n})$, and for $n=2$, $k_2$
finite and $\mur{n}=0$ it coincides with~\cite[Theorem~1.2]{Warnaar10}.
When $n\geq 2$ there is some mild redundancy in \eqref{Eq_An-Cauchy-1} 
since the substitution $X^{(1)}\mapsto z_1^{-1} X^{(1)}$ eliminates any 
reference to $z_1$. 
We further remark that we do not know how to evaluate the left-hand side
of \eqref{Eq_An-Cauchy-1} in closed-form if one (or more) of the $a_r$
for $1\leq r\leq n-2$ is an indeterminate. 
Since $\abs{X^{(r)}}=k_r$ for $1\leq r\leq n-1$ the summand vanishes
unless $l(\lar{r})\leq k_r$ for this same range of $r$.
If $a_r$ for some $r\leq n-2$ is an indeterminate, then $\lar{r+1}$ can
have an arbitrarily large length, which prevents us from applying
Proposition~\ref{Prop_skew-sum} in our proof.
Requiring $a_r=t^{k_r-k_{r+1}}$ for $1\leq r\leq n-1$ allows us to use
Corollary~\ref{Cor_skew-sum} in place of the proposition.
We are, however, allowed to keep $a_{n-1}$ an indeterminate since
$Y^{(n)}$ is either a finite alphabet of cardinality $k_n$, or countably
infinite, permitting us to apply Proposition~\ref{Prop_skew-sum}.
For more details we refer to the proof of the theorem contained in the next
section.

There is a second, closely related, Cauchy-type identity, in which $k_n$ 
is finite and no longer corresponds to the cardinality of $Y^{(n)}$.

\begin{corollary}[$\A_n$ Cauchy-type formula II]\label{Cor_An-Cauchy-2}
Let $k_1\leq k_2\leq\cdots\leq k_n$ be nonnegative integers.
Then for $a_r:=t^{k_r-k_{r+1}}$ for $1\leq r\leq n-1$, 
$X^{(1)}=x_1+\dots+x_{k_1}$, $Y^{(n)}=y_1+y_2+\cdots$ and 
$X^{(2)},\dots,X^{(n)},Y^{(1)},\dots,Y^{(n-1)}$ as in \eqref{Eq_XY}, and
\[
W:=z_1\cdots z_{n-1}X^{(1)} +
\sum_{r=1}^{n-1}z_{r+1}\cdots z_{n-1}\,\frac{1-t^{k_r-k_{r+1}}}{1-t},
\]
we have
\begin{align}\label{Eq_An-Cauchy-2}
&\sum_{\lar{1},\dots,\lar{n}} 
\prod_{r=1}^n P_{\lar{r}}\big(\big[X^{(r)}\big];q,t\big)
Q_{\lar{r}/\mur{r}}\big(\big[Y^{(r)}\big];q,t\big)
\prod_{r=1}^{n-1} f_{\lar{r},\lar{r+1}}^{k_r,k_{r+1}}(a_r;q,t) \\
&\quad=P_{\mur{n}}\big([W];q,t\big)\prod_{r=1}^{n-1}\bigg(
\prod_{i=1}^{k_1}\frac{(qt^{k_r-k_{r+1}}z_1\cdots z_r x_i;q)_{\infty}}
{(tz_1\cdots z_r x_i;q)_{\infty}} 
\prod_{j\geq 1} \frac{(t^{k_{r+1}-k_r}z_{r+1}\cdots z_{n-1}y_j;q)_{\infty}}
{(z_{r+1}\cdots z_{n-1}y_j;q)_{\infty}}\bigg) \notag \\
&\quad\quad\quad\times \prod_{i=1}^{k_1} \prod_{j\geq 1} 
\frac{(tz_1\cdots z_{n-1}x_iy_j;q)_{\infty}}
{(z_1\cdots z_{n-1}x_iy_j;q)_{\infty}}
\prod_{1\leq r<s\leq n-1} \prod_{i=1}^{k_{r+1}-k_r} 
\frac{(qt^{i+k_s-k_{s+1}-1}z_{r+1}\cdots z_s;q)_{\infty}}
{(t^iz_{r+1}\cdots z_s;q)_{\infty}}, \notag
\end{align}
where $\mur{1},\dots,\mur{n-1}:=0$ and $\mur{n}$ is an arbitrary partition.
\end{corollary}

We note that in the above corollary the range for which \eqref{Eq_Xr-finite}
holds includes $r=n$.
In particular $\abs{X^{(r)}}=k_r$ for all $1\leq r\leq n$.
The corollary simplifies to the $\A_n$ $q$-binomial theorem
\cite[Theorem~3.2]{Warnaar09} if we replace $Y^{(n)}\mapsto
z_nt^{k_{n-1}}(1-a)/(1-t)$ and 
$z_r\mapsto z_r t^{k_{r-1}-1}$ for all $1\leq r\leq n-1$ where $k_0:=0$.

\begin{proof}[Proof of Corollary~\ref{Cor_An-Cauchy-2}]
We take Theorem~\ref{Thm_An-Cauchy} with $k_n=\infty$.
Then $Y^{(n)}=y_1+y_2+\cdots$ in accordance with the corollary.
Moreover, by \eqref{Eq_fQfQ} and the fact that $\mur{n-1}=0$,
\begin{align*}
&Q_{\lar{n-1}/\mur{n-1}}\big[Y^{(n-1)}\big]
f_{\lar{n-1},\lar{n}}^{k_{n-1},k_n}(a_{n-1};q,t) \\
&\qquad=Q_{\lar{n-1}}\bigg[\frac{z_nt}{1-t}\bigg]
f_{\lar{n-1},\lar{n}}^{k_{n-1},\infty}(a_{n-1};q,t) \\
&\qquad=Q_{\lar{n-1}}\bigg[z_n\frac{t-a_{n-1}qt^{\hat{k}_n}}{1-t}\bigg]
f_{\lar{n-1},\lar{n}}^{k_{n-1},\hat{k}_n}(a_{n-1};q,t) \\
&\qquad=Q_{\lar{n-1}/\mur{n-1}}[\hat{Y}^{(n-1)}]
f_{\lar{n-1},\lar{n}}^{k_{n-1},\hat{k}_n}(a_{n-1};q,t).
\end{align*} 
Here $\hat{k}_n$ is an arbitrary integer such that
$\hat{k}_n\geq l(\lar{n})$ and 
$\hat{Y}^{(n-1)}:=z_n(t-a_{n-1}qt^{\hat{k}_n})/(1-t)$,
so that $\hat{Y}^{(n-1)}$ corresponds to $Y^{(n-1)}$ in \eqref{Eq_XY}
except that $k_n$ has been replaced by $\hat{k}_n$.
Of course, since we are summing over all partitions $\lar{n}$ there
exists no integer $\hat{k}_n$ such that $\hat{k}_n\geq l(\lar{n})$ for
all $\lar{n}$.
To get around this problem we specialise $a_{n-1}=q^{k_{n-1}-\hat{k}_n}$.
Then $X^{(n)}=(1-t^{\hat{k}_n})/(1-t)$ of cardinality $\abs{\hat{k}_n}$
so that without loss of generality we may assume that
$l(\lar{n})\leq \hat{k}_n$.
Finally replacing $\hat{k}_n$ by $k_n$ completes the proof.
\end{proof}

The proof of Theorem~\ref{Thm_An-AFLT} actually requires a plethystically
substituted version of the $\mur{n}=0$ instance of \eqref{Eq_An-Cauchy-2}
obtained by replacing $Y^{(n)}\mapsto Y^{(n)}+(c-d)/(1-t)$.
This substitution can easily be carried out noting that the right-hand
side of \eqref{Eq_An-Cauchy-2} without $P_{\mur{n}}[W]$ is expressible 
in terms of $\sigma_1$ as
\begin{multline*}
\sigma_1\Bigg[\sum_{r=1}^{n-1}
\bigg(tz_1\cdots z_r \frac{1-qt^{k_r-k_{r+1}-1}}{1-q}\,X^{(1)}+
z_{r+1}\cdots z_{n-1} \frac{1-t^{k_{r+1}-k_r}}{1-q}\,Y^{(n)}\bigg)+\\
z_1\cdots z_{n-1} \frac{1-t}{1-q}\,X^{(1)}Y^{(n)}+
\sum_{1\leq r<s\leq n-1}
tz_{r+1}\cdots z_s \frac{(1-qt^{k_s-k_{s+1}-1})(1-t^{k_{r+1}-k_r})}
{(1-q)(1-t)}\Bigg].
\end{multline*}

\begin{corollary}\label{Cor_An-Cauchy}
With the same conditions as in Corollary~\ref{Cor_An-Cauchy-2},
\begin{align*}
&\sum_{\lar{1},\dots,\lar{n}} 
P_{\lar{n}}\big(\big[X^{(n)}\big];q,t) 
Q_{\lar{n}}\bigg(\bigg[Y^{(n)}+\frac{c-d}{1-t}\bigg];q,t\bigg) \\
&\qquad\qquad \times
\prod_{r=1}^{n-1} \Big( P_{\lar{r}}\big(\big[X^{(r)}\big];q,t\big)
Q_{\lar{r}}\big(\big[Y^{(r)}\big];q,t\big)
f_{\lar{r},\lar{r+1}}^{k_r,k_{r+1}}(a_r;q,t)\Big) \\
&\quad\quad=\prod_{r=1}^{n-1}\bigg(
\prod_{i=1}^{k_1}\frac{(qt^{k_r-k_{r+1}}z_1\cdots z_r x_i;q)_{\infty}}
{(tz_1\cdots z_r x_i;q)_{\infty}} 
\prod_{j\geq 1} \frac{(t^{k_{r+1}-k_r}z_{r+1}\cdots z_{n-1}y_j;q)_{\infty}}
{(z_{r+1}\cdots z_{n-1}y_j;q)_{\infty}}\bigg) \\
&\quad\quad\quad\times \prod_{i=1}^{k_1}\prod_{j\geq 1} 
\frac{(tz_1\cdots z_{n-1}x_iy_j;q)_{\infty}}
{(z_1\cdots z_{n-1}x_iy_j;q)_{\infty}}
\prod_{i=1}^{k_1} \frac{(dz_1\cdots z_{n-1} x_i;q)_\infty}
{(cz_1\cdots z_{n-1}x_i;q)_\infty} \\
&\quad\quad\quad \times
\prod_{r=1}^{n-1} \prod_{i=1}^{k_{r+1}-k_r}
\frac{(dz_{r+1}\cdots z_{n-1}t^{i-1};q)_{\infty}}
{(cz_{r+1}\cdots z_{n-1}t^{i-1};q)_{\infty}} 
\prod_{1\leq r<s\leq n-1} \prod_{i=1}^{k_{r+1}-k_r} 
\frac{(qt^{i+k_s-k_{s+1}-1}z_{r+1}\cdots z_s;q)_{\infty}}
{(t^iz_{r+1}\cdots z_s;q)_{\infty}}. 
\end{align*}
\end{corollary}

\subsection{Proof of Theorem~\ref{Thm_An-Cauchy}}
\label{Sec_Cauchy-type-proof}

We define two families of auxiliary alphabets 
$\{X^{(r,m)}\}_{0\leq m<r\leq n}$ and $\{Z^{(r)}\}_{r=1}^n$ as
\begin{align*}
X^{(r,m)}&:=\begin{cases} \displaystyle
z_1\cdots z_m X^{(1)}
+\sum_{u=1}^m z_{u+1}\cdots z_m\,\frac{1-1/a_u}{1-t} & \text{if $r=m+1$},\\
\displaystyle \frac{1-1/a_{r-1}}{1-t} & \text{otherwise},
\end{cases}
\intertext{and}
Z^{(r)}&:=\begin{cases} \displaystyle
Y^{(n)} & \text{if $r=n$}, \\[2mm]
\displaystyle  z_r\frac{t-a_rq}{1-t} & \text{otherwise}.
\end{cases}
\end{align*}
The first family satisfies the simple recursion
\begin{equation}\label{Eq_X-rec}
z_{m+1}X^{(m+1,m)}+X^{(m+2,m)}=X^{(m+2,m+1)}.
\end{equation}

\begin{lemma}\label{Lem_g-rec}
For $n,m$ integers such that $0\leq m\leq n-1$, and $\nur{n}$ a
partition, define
\begin{align}\label{Eq_gdef}
g_m&:=\prod_{r=1}^m\prod_{i=1}^{k_1}
\frac{(a_r qz_1\cdots z_rx_i;q)_{\infty}}{(tz_1\cdots z_rx_i;q)_{\infty}}
\prod_{1\leq r<s\leq m}\prod_{i=1}^{k_{r+1}-k_r}
\frac{(a_sqt^{i-1}z_{r+1}\cdots z_s;q)_{\infty}}
{(t^i z_{r+1}\cdots z_s;q)_{\infty}} \\
&\quad \times\sum_{\nur{m+1},\dots,\nur{n-1}}
\prod_{r=m+1}^n \bigg(z_r^{\abs{\nur{r}}}
\sum_{\lar{r}}P_{\lar{r}/\nur{r-1}}\big(\big[X^{(r,m)}\big];q,t\big)
Q_{\lar{r}/\nur{r}}\big(\big[Z^{(r)}\big];q,t\big)\bigg),\notag
\end{align}
where $\nur{m}:=0$.
Then $g_m=g_{m+1}$ for $0\leq m\leq n-2$.
\end{lemma}

\begin{proof}
Since $\nur{m}:=0$, the sum over $\lar{m+1}$ in \eqref{Eq_gdef} is of
the form \eqref{Eq_Cauchy} with
\[
(X,Y,\la,\mu)\mapsto \big(X^{(m+1,m)},Z^{(m+1)},\lar{m+1},\nur{m+1}\big)
\]
and hence equates to
\[
P_{\nur{m+1}}\big[X^{(m+1,m)}\big]
\sigma_1\bigg[\frac{1-t}{1-q}\,X^{(m+1,m)}Z^{(m+1)}\bigg].
\]
Since $0\leq m\leq n-2$, it follows that
\begin{align*}
&\sigma_1\bigg[\frac{1-t}{1-q}\,X^{(m+1,m)}Z^{(m+1)}\bigg] \\
&\quad=\sigma_1\bigg[z_1\cdots z_{m+1} \frac{t-a_{m+1}q}{1-q}\,X^{(1)}
+\sum_{r=1}^m z_{r+1}\cdots z_{m+1}\,
\frac{(1-t^{k_{r+1}-k_r})(t-a_{m+1}q)}{(1-t)(1-q)}\bigg] \\
&\quad=\prod_{i=1}^{k_1}\sigma_{z_1\cdots z_{m+1}}
\bigg[\frac{t-a_{m+1}q}{1-q}\,x_i\bigg]
\prod_{r=1}^m \prod_{i=1}^{k_{r+1}-k_r}\sigma_{z_{r+1}\cdots z_{m+1}}
\bigg[\frac{t^{i-1}(t-a_{m+1}q)}{1-q}\bigg] \\
&\quad=\prod_{i=1}^{k_1}\frac{(a_{m+1}qz_1\cdots z_{m+1}x_i;q)_{\infty}}
{(tz_1\cdots z_{m+1}x_i;q)_{\infty}}
\prod_{r=1}^m \prod_{i=1}^{k_{r+1}-k_r}
\frac{(a_{m+1}qt^{i-1}z_{r+1}\cdots z_{m+1};q)_{\infty}}
{(t^iz_{r+1}\cdots z_{m+1};q)_{\infty}},
\end{align*}
where the second equality follows from \eqref{Eq_sigma-sum} and the last 
equality from \eqref{Eq_sigma-kernel}.
As a result,
\begin{align*}
g_m&=\prod_{r=1}^{m+1}\prod_{i=1}^{k_1}
\frac{(a_r qz_1\cdots z_rx_i;q)_{\infty}}{(tz_1\cdots z_rx_i;q)_{\infty}}
\prod_{1\leq r<s\leq m+1}\prod_{i=1}^{k_{r+1}-k_r}
\frac{(a_sqt^{i-1}z_{r+1}\cdots z_s;q)_{\infty}}
{(t^i z_{r+1}\cdots z_s;q)_{\infty}} \\
&\quad\times \sum_{\nur{m+1},\dots,\nur{n-1}}
\bigg\{P_{\nur{m+1}}\big[z_{m+1}X^{(m+1,m)}\big] \\
&\qquad\qquad\qquad\qquad\quad\times
\prod_{r=m+2}^n\bigg(z_r^{\abs{\nur{r}}}\sum_{\lar{r}}
P_{\lar{r}/\nur{r-1}}\big[X^{(r,m)}\big]
Q_{\lar{r}/\nur{r}}\big[Z^{(r)}\big]\bigg)\bigg\}.
\end{align*}
After interchanging the order of the sum over $\nur{m+1}$ with those over
the $\lar{r}$, the former can be summed using \eqref{Eq_Mac-skew} with
\[
(X,Y,\la,\mu)\mapsto 
\big(X^{(m+2,m)},z_{m+1}X^{(m+1,m)},\lar{m+2},\nur{m+1}\big).
\]
Thanks to the recursion \eqref{Eq_X-rec} this yields
$P_{\lar{m+2}}[X^{(m+2,m+1)}]$, resulting in $g_m=g_{m+1}$.
\end{proof}

We are now ready to prove Theorem~\ref{Thm_An-Cauchy}.
As a first step we eliminate 
$f_{\lar{r},\lar{r+1}}^{k_r,k_{r+1}}(a_r;q,t)$ from the summand in
\eqref{Eq_An-Cauchy-1}
by applying Corollary~\ref{Cor_skew-sum} with
\[
(\la,\mu,\nu,k,\ell)\mapsto (\lar{r},\lar{r+1},\nur{r},k_r,k_{r+1})
\quad \text{for $1\leq r\leq n-2$}
\]
and Proposition~\ref{Prop_skew-sum} with
\[
(a,\la,\mu,\nu,k,\ell)\mapsto 
(a_{n-1},\lar{n-1},\lar{n},\nur{n-1},k_{n-1},k_n)\quad
\text{for $r=n-1$}.
\]
Then
\begin{align*}
&\sum_{\lar{1},\dots,\lar{n}} 
\prod_{r=1}^n P_{\lar{r}}\big[X^{(r)}\big]
Q_{\lar{r}/\mur{r}}\big[Y^{(r)}\big]
\prod_{r=1}^{n-1} f_{\lar{r},\lar{r+1}}^{k_r,k_{r+1}}(a_r;q,t) \\
&\qquad\quad =\sum_{\lar{1},\dots,\lar{n}}
\sum_{\nur{1},\dots,\nur{n-1}}
P_{\lar{1}}\big[X^{(1)}\big]Q_{\lar{n}/\mur{n}}\big[Y^{(n)}\big] \\
&\qquad\qquad \qquad \qquad \qquad \qquad \quad \times \prod_{r=1}^{n-1}
z_r^{\abs{\lar{r}}}Q_{\lar{r}/\nur{r}}\bigg[\frac{t-a_rq}{1-t}\bigg]
P_{\lar{r+1}/\nur{r}}\bigg[\frac{1-1/a_r}{1-t}\bigg] \\
&\qquad \quad =\sum_{\nur{1},\dots,\nur{r-1}}\prod_{r=1}^n
\bigg(z_r^{\abs{\nur{r}}}\sum_{\lar{r}}
P_{\lar{r}/\nur{r-1}}[X^{(r,0)}]
Q_{\lar{r}/\nur{r}}\big[Z^{(r)}\big]\bigg)
\bigg|_{z_n=1,\,\nur{n}=\mur{n}} \\[2mm]
& \qquad \quad =g_0\big|_{z_n=1,\,\nur{n}=\mur{n}},
\end{align*}
where in the fourth line $\nur{0}:=0$, and where $g_0$ is defined in
\eqref{Eq_gdef}.
Using Lemma~\ref{Lem_g-rec} we may replace $g_0$ by $g_{n-1}$, leading to
\begin{align*}
\sum_{\lar{1},\dots,\lar{n}} &
\prod_{r=1}^n P_{\lar{r}}\big[X^{(r)}\big]
Q_{\lar{r}/\mur{r}}\big[Y^{(r)}\big]
\prod_{r=1}^{n-1} f_{\lar{r},\lar{r+1}}^{k_r,k_{r+1}}(a_r;q,t) \\
&=\prod_{r=1}^{n-1}\prod_{i=1}^{k_1}
\frac{(a_rqz_1\cdots z_rx_i;q)_{\infty}}
{(tz_1\cdots z_rx_i;q)_{\infty}}
\prod_{1\leq r<s\leq n-1}\prod_{i=1}^{k_{r+1}-k_r}
\frac{(a_sqt^{i-1}z_{r+1}\cdots z_s;q)_{\infty}}
{(t^iz_{r+1}\cdots z_s;q)_{\infty}} \\[1mm]
&\quad\times\sum_{\lar{n}}P_{\lar{n}}\big[X^{(n,n-1)}\big]
Q_{\lar{n}/\mur{n}}\big[Y^{(n)}\big].
\end{align*}
The final sum on the right can be carried out by \eqref{Eq_Cauchy} with
\[
(X,Y,\la,\mu)\mapsto \big(X^{(n,n-1)},Y^{(n)},\lar{n},\mur{n}\big).
\]
Since $W=X^{(n,n-1)}$ and
\begin{align*}
\sigma_1\bigg[\frac{1-t}{1-q}\,X^{(n,n-1)}Y^{(n)}\bigg] 
&=\sigma_1\bigg[z_1\cdots z_{n-1}\,\frac{1-t}{1-q}\,X^{(1)}Y^{(n)}+
\sum_{r=1}^{n-1}z_{r+1}\cdots z_{n-1}\,\frac{1-1/a_r}{1-q}\,Y^{(n)}\bigg] \\
&=\prod_{i=1}^{k_1}\prod_{j=1}^{k_n}
\frac{(tz_1\cdots z_{n-1}x_iy_j;q)_{\infty}}
{(z_1\cdots z_{n-1}x_iy_j;q)_{\infty}}
\prod_{r=1}^{n-1}\prod_{j=1}^{k_n}
\frac{(z_{r+1}\cdots z_{n-1}y_j/a_r;q)_{\infty}}
{(z_{r+1}\cdots z_{n-1}y_j;q)_{\infty}},
\end{align*}
the right-hand side of the theorem results.

\section{The $\A_n$ AFLT integral}\label{Sec_AFLT}

In this section we first give a description of the domain on integration 
of the $\A_n$ AFLT integral \eqref{Eq_An-AFLT} and then
apply the $\A_n$ Cauchy identity of Corollary~\ref{Cor_An-Cauchy} 
to prove this integral.
At the end of the section we give a companion to the AFLT integral based on 
Theorem~\ref{Thm_An-Cauchy}

\subsection{The domain $\boldsymbol{C_{\gamma}^{k_1,\dots,k_n}[0,1]}$}\label{Sec_chain}

The domain of integration $C_{\gamma}^{k_1,\dots,k_n}[0,1]$ of the 
integral \eqref{Eq_I-An} takes the form of a chain in the usual sense
of algebraic topology.
Since this chain is the same as that of the $\A_n$ Selberg integral
of \cite[Theorem~1.2]{Warnaar09} (see also \cite{TV03} where this
chain first appeared in the case of $\A_2$), we refer to \cite{Warnaar09} 
for the details of exactly how it arises in the course of the proof
presented in Section~\ref{Sec_An-AFLT-proof}.

For $n=1$ the domain $C_{\gamma}^{k_1}[0,1]$ is the $k_1$-simplex given 
in \eqref{Eq_simplex}.
In order to describe $C_{\gamma}^{k_1,\dots,k_n}[0,1]$ for $n\geq 2$, 
we first consider $D^{k_1,\dots,k_n}[0,1]\subseteq[0,1]^{k_1+\dots+k_n}$
as the set of points
\[
\big(\tar{1},\tar{2},\dots,\tar{n}\big)=
\big(\tar{1}_1,\dots,\tar{1}_{k_1},
\tar{2}_1,\dots,\tar{2}_{k_2},\dots,\tar{n}_1,
\dots,\tar{n}_{k_n}\big)\in[0,1]^{k_1+\dots+k_n}
\]
subject to
\[
0<\tar{r}_1<\dots<\tar{r}_{k_r}<1 \qquad
\text{for $1\leq r\leq n$},
\]
and
\[
\tar{r}_i<\tar{r+1}_{i-k_r+k_{r+1}} \qquad
\text{for $1\leq i\leq k_r$, $1\leq r\leq n-1$}.
\]
Following \eqref{Eq_inequalities-2} this may be visualised as
\[
\begin{array}{ccccccccccccccccl}
&&&&0&\LE&\tar{r}_1&\LE&\tar{r}_2 &\LE&\cdots&\LE&
\tar{r}_{k_r}&\LE&1 \\[-1mm]
&&&&&&\text{\begin{turn}{270}$<$\end{turn}}&&
\text{\begin{turn}{270}$<$\end{turn}} &&&&
\text{\begin{turn}{270}$<$\end{turn}} &&& \\[3mm]
0&\LE&\tar{r+1}_1&\LE&\cdots&\LE&\tar{r+1}_{k_{r+1}-k_r+1}&\LE &
\tar{r+1}_{k_{r+1}-k_r+2}&\LE&\cdots&\LE
&\tar{r+1}_{k_{r+1}}&\LE& 1.
\end{array}
\]
We need to consider all possible total orderings between the integration
variables $\tar{r}$ and $\tar{r+1}$ consistent with the above partial order.
Each such total ordering may be described by a map
\[
M_r : \{1,\dots,k_r\}\longrightarrow \{1,\dots,k_{r+1}\}
\]
such that $M_r(i)\leq M_r(i+1)$ and $1\leq M_r(i)\leq i+k_{r+1}-k_r$, so that
\begin{equation}\label{Eq_chain-an}
\tar{r+1}_{M_r(i)-1}<\tar{r}_i<\tar{r+1}_{M_r(i)},
\end{equation}
where $\tar{r+1}_0:=0$.
In view of this we define the sets
\[
D_{M_1,\dots,M_{n-1}}^{k_1,\dots,k_n}\subseteq D^{k_1,\dots,k_n}[0,1]
\]
by requiring that \eqref{Eq_chain-an} holds for fixed admissible maps
$M_1,\dots,M_{n-1}$.
Then $D^{k_1,\dots,k_n}$ can be written as the chain
\[
D^{k_1,\dots,k_n}[0,1]=\sum_{M_1,\dots,M_{n-1}}
D_{M_1,\dots,M_{n-1}}^{k_1,\dots,k_n}[0,1],
\]
where the sum is over all admissible maps $M_1,\dots,M_{n-1}$.
Analytically continuing the weight function 
\begin{equation}\label{Eq_weight}
F_{M_1,\dots,M_{n-1}}^{k_1,\dots,k_n}(\gamma)
:=\prod_{r=1}^{n-1}\prod_{i=1}^{k_r}
\frac{\sin(\pi(i+k_{r+1}-k_r-M_r(i)+1)\gamma)}
{\sin(\pi(i+k_{r+1}-k_r)\gamma)}
\quad\text{for $\gamma\in\mathbb{C}\setminus\mathbb{Z}$},
\end{equation}
to include $\gamma=0$,
the chain ${C}_\gamma^{k_1,\dots,k_n}[0,1]$ is defined as
\begin{equation}\label{Eq_chain-Cn-def}
{C_\gamma^{k_1,\dots,k_n}[0,1]}
:=\sum_{M_1,\dots,M_{n-1}}F_{M_1,\dots,M_{n-1}}^{k_1,\dots,k_n}(\gamma)
D_{M_1,\dots,M_{n-1}}^{k_1,\dots,k_n}[0,1].
\end{equation}
Note that it follows from the above that
\[
C_\gamma^{0,k_2,\dots,k_n}[0,1]=C_\gamma^{k_2,\dots,k_n}[0,1].
\]

\subsection{Proof of Theorem~\ref{Thm_An-AFLT}}\label{Sec_An-AFLT-proof}

We begin with the identity of Corollary~\ref{Cor_An-Cauchy}, where we note that
the alphabets $X^{(1)}$ and $Y^{(n)}$ contain $k_1$ variables and countably 
many variables respectively.
We now fix a nonnegative integer $m$ and set $Y^{(n)}_i=0$ for $i>m$. 
Next we fix a pair of partitions $\la\in\mathscr{P}_{k_1}$ 
and $\mu\in\mathscr{P}_m$, and carry out the specialisation
\[
\big(X^{(1)},Y^{(n)},c,d\big)\mapsto
\big(\spec{\la}_{k_1},bz_nt^{1-m}\spec{\mu}_m,z_nt,bz_nt^{1-m}\big).
\]
Also replacing $\lar{1},\dots,\lar{n}$ by $\nur{1},\dots,\nur{n}$,
this leads to the identity
\begin{align*}
&\sum_{\nur{1},\dots,\nur{n}} P_{\nur{1}}[\spec{\la}_{k_1}]
Q_{\nur{n}}\bigg[bt^{-m}\spec{\mu}_m+\frac{1-bt^{-m}}{1-t}\bigg]\\
&\qquad\qquad\times
\prod_{r=1}^n(tz_r)^{\abs{\nur{r}}}
\prod_{r=1}^{n-1} \bigg( P_{\nur{r+1}}\bigg[\frac{1-t^{k_{r+1}}}{1-t}\bigg]
Q_{\nur{r}}\bigg[\frac{1-qt^{k_r-1}}{1-t}\bigg]
f_{\nur{r},\nur{r+1}}^{k_r,k_{r+1}}(t^{k_r-k_{r+1}};q,t) \bigg) \\
&\quad=\prod_{r=1}^{n-1}\bigg( \prod_{i=1}^{k_1}
\frac{(z_1\cdots z_r q^{\la_i+1}t^{k_1+k_r-k_{r+1}-i};q)_{\infty}}
{(z_1\cdots z_r q^{\la_i}t^{k_1-i+1};q)_{\infty}} 
\prod_{j=1}^m
\frac{(bz_{r+1}\cdots z_n q^{\mu_j}t^{k_{r+1}-k_r-j};q)_{\infty}}
{(bz_{r+1}\cdots z_n q^{\mu_j}t^{-j};q)_{\infty}} \bigg)\\
&\quad\quad\times \prod_{i=1}^{k_1} \prod_{j=1}^m
\frac{(bz_1\cdots z_n q^{\la_i+\mu_j}t^{k_1+1-i-j};q)_{\infty}}
{(bz_1\cdots z_n q^{\la_i+\mu_j}t^{k_1-i-j};q)_{\infty}}
\prod_{i=1}^{k_1}
\frac{(bz_1\cdots z_{n-1} q^{\la_i}t^{k_1-m-i};q)_\infty}
{(z_1\cdots z_{n-1} q^{\la_i}t^{k_1-i};q)_\infty} \\
&\quad\quad \times
\prod_{r=1}^{n-1}\prod_{i=1}^{k_{r+1}-k_r}
\frac{(bz_{r+1}\cdots z_{n-1}t^{i-m-1};q)_{\infty}}
{(z_{r+1}\cdots z_{n-1}t^{i-1};q)_{\infty}} 
\prod_{1\leq r<s\leq n-1} \prod_{i=1}^{k_{r+1}-k_r} 
\frac{(qt^{k_s-k_{s+1}+i-1}z_{r+1}\cdots z_s;q)_{\infty}}
{(t^iz_{r+1}\cdots z_s;q)_{\infty}},
\end{align*}
where we have dropped $a_r$ in favour of $t^{k_r-k_{r+1}}$ in comparison
with Corollary~\ref{Cor_An-Cauchy}.
By virtue of the evaluation symmetry \eqref{Eq_eval-sym} and the
generalised evaluation symmetry of Lemma~\ref{Lem_gen-eval-sym}, we have
\begin{align*}
P_{\nur{1}}\big[\spec{\la}_{k_1}\big]
&=\frac{P_{\nur{1}}\big[\frac{1-t^{k_1}}{1-t}\big]}
{P_{\la}\big[\frac{1-t^{k_1}}{1-t}\big]}
P_{\la}\big[\spec{\nur{1}}_{k_1}\big]
\intertext{and}
Q_{\nur{n}}\bigg[bt^{-m}\spec{\mu}_m+\frac{1-bt^{-m}}{1-t}\bigg]
&=\frac{Q_{\nur{n}}\big[\frac{1-b}{1-t}\big]}
{P_{\mu}\big[\frac{1-b}{1-t}\big]}\,
P_{\mu}\bigg[bt^{-k_n}\spec{\nur{n}}_{k_n}+\frac{1-bt^{-k_n}}{1-t}\bigg],
\end{align*}
effectively allowing us to interchange the roles of $\nur{1},\nur{n}$ and
$\la,\mu$ in the summand.
Carrying this out and multiplying both sides by 
\begin{equation}\label{Eq_PmuPnu}
P_{\la}\bigg[\frac{1-t^{k_1}}{1-t}\bigg]
P_{\mu}\bigg[\frac{1-b}{1-t}\bigg],
\end{equation}
the left-hand side of the above identity becomes
\begin{multline*}
\sum_{\nur{1},\dots,\nur{n}} P_{\la}\big[\spec{\nur{1}}_{k_1}]
P_{\mu}\bigg[bt^{-k_n}\spec{\nur{n}}_{k_n}+\frac{1-bt^{-k_n}}{1-t}\bigg]
(b;q,t)_{\nur{n}} \\
\times \prod_{r=1}^n \frac{(tz_r)^{\abs{\nur{r}}} t^{2n(\nur{r})}
(t^{k_r};q,t)_{\nur{r}}}{c_{\nur{r}}(q,t)c'_{\nur{r}}(q,t)}
\prod_{r=1}^{n-1} (qt^{k_r-1};q,t)_{\nur{r}}
f_{\nur{r},\nur{r+1}}^{k_r,k_{r+1}}(t^{k_r-k_{r+1}};q,t), 
\end{multline*}
where we have also used the specialisation formulas \eqref{Eq_princ-spec}
and \eqref{Eq_a-spec}.
The corresponding right-hand side is as before, except for the additional
factor \eqref{Eq_PmuPnu}.
Next we use \eqref{Eq_qt-Poch-def} and \eqref{Eq_hook} in the summand,
make the further substitutions
\[
b\mapsto q^{\beta+(k_n-1)\gamma},\quad t\mapsto q^\gamma\quad\text{and}
\quad z_r\mapsto q^{\alpha_r-\gamma} \quad \text{for $1\leq r\leq n$},
\]
and introduce auxiliary variables $(\tar{1},\dots,\tar{n})$ as
$\tar{r}_i:=q^{\nur{r}_i+(k_r-i)\gamma}$. 
To more simply express the resulting identity we introduce some
additional notation, and for alphabets $t=(t_1,\dots,t_k)$, 
$s=(s_1,\dots,s_{\ell})$ define
\[
\Delta_{\gamma}(t;q):=\prod_{1\leq i<j\leq k}
t_j^{2\gamma}\big(1-t_i/t_j\big)
\big(q^{1-\gamma}t_i/t_j;q\big)_{2\gamma-1}
\]
and
\[
\Delta_\gamma(t,s;q):=\prod_{i=1}^k\prod_{j=1}^{\ell}s_j^{-\gamma}
\big(qt_i/s_j;q\big)_{-\gamma}.
\]
After multiplying through by $(1-q)^{k_1+\dots+k_n}$, this yields
\begin{align*}
&(1-q)^{k_1+\dots+k_n}\sum_{\nur{1},\dots,\nur{n}}
P_{\la}\big(\tar{1};q,q^{\gamma}\big)
P_{\mu}\bigg(\bigg[q^{\beta-\gamma}\tar{n}
+\frac{1-q^{\beta-\gamma}}{1-q^\gamma}\bigg];q,q^{\gamma}\bigg)\\
&\qquad\qquad\qquad\qquad\qquad\times
\prod_{r=1}^n\Big( \Delta_\gamma\big(\tar{r};q\big)
\prod_{i=1}^{k_r}\big(\tar{r}_i\big)^{\alpha_r}
\big(q\tar{r}_i;q\big)_{\beta_r-1}\Big)
\prod_{r=1}^{n-1}\Delta_\gamma\big(\tar{r},\tar{r+1};q\big) \\
&=q^{\gamma\sum_{r=1}^n \big(\alpha_r\binom{k_r}{2}+2\gamma\binom{k_r}{3}\big)
-\gamma^2\sum_{r=1}^{n-1}k_r\binom{k_{r+1}}{2}} \\
&\quad\times
P_{\la}\bigg(\bigg[\frac{1-q^{k_1\gamma}}
{1-q^\gamma}\bigg];q,q^\gamma\bigg)
P_{\mu}\bigg(\bigg[\frac{1-q^{\beta+(k_n-1)\gamma}}{1-q^\gamma}\bigg];
q,q^{\gamma}\bigg)\\[1mm]
&\quad\times \prod_{r=1}^{n-1} \prod_{i=1}^{k_r}
\frac{\Gamma_q(1+(i-k_{r+1}-1)\gamma)\Gamma_q(i\gamma)}{\Gamma_q(\gamma)}
\prod_{i=1}^{k_n}
\frac{\Gamma_q(\beta+(i-1)\gamma)\Gamma_q(i\gamma)}{\Gamma_q(\gamma)}\\
&\quad\times \prod_{1\leq r<s\leq n-1} \prod_{i=1}^{k_{r+1}-k_r} 
\frac{\Gamma_q(\alpha_{r+1}+\dots+\alpha_s+(r-s+i)\gamma)}
{\Gamma_q(1+\alpha_{r+1}+\dots+\alpha_s+(k_s-k_{s+1}+i+r-s-1)\gamma)}\\
&\quad\times \prod_{r=1}^{n-1}\bigg(\prod_{i=1}^{k_{r+1}-k_r}
\frac{\Gamma_q(\alpha_{r+1}+\dots+\alpha_n+(r-n+i)\gamma)}
{\Gamma_q(\alpha_{r+1}+\dots+\alpha_n+\beta+(k_n-m+r-n+i-1)\gamma)}\\
&\quad\quad\qquad\;\times
\prod_{i=1}^{k_1}
\frac{\Gamma_q(\alpha_1+\dots+\alpha_r+(k_1-r-i+1)\gamma+\la_i)}
{\Gamma_q(1+\alpha_1+\dots+\alpha_r+(k_1+k_r-k_{r+1}-r-i)\gamma+\la_i)} \\
&\quad\quad\qquad\;\times \prod_{j=1}^m
\frac{\Gamma_q(\alpha_{r+1}+\dots+\alpha_n+\beta+(k_n+r-n-j)\gamma+\mu_j)}
{\Gamma_q(\alpha_{r+1}+\dots+\alpha_n+\beta+
(k_{r+1}-k_r+k_n+r-n-j)\gamma+\mu_j)}\bigg)\\
&\quad\times \prod_{i=1}^{k_1}
\frac{\Gamma_q(\alpha_1+\dots+\alpha_n+(k_1-n-i+1)\gamma+\la_i)}
{\Gamma_q(\alpha_1+\dots+\alpha_n+\beta+(k_1+k_n-m-n-i)\gamma+\la_i)} \\
&\quad\times \prod_{i=1}^{k_1} \prod_{j=1}^m
\frac{\Gamma_q(\alpha_1+\dots+\alpha_n
+\beta+(k_1+k_n-n-i-j)\gamma+\la_i+\mu_j)}
{\Gamma_q(\alpha_1+\dots+\alpha_n+\beta+
(k_1+k_n-n-i-j+1)\gamma+\la_i+\mu_j)},
\end{align*}
where $\beta_1=\dots=\beta_{n-1}:=1$ and $\beta_n:=\beta$.
The above is a restricted $q$-integral over the domain 
$D^{k_1,\dots,k_n}[0,1]$.
To complete the proof we divide the above identity by its
$\la=\mu=0$ case and then take the $q\to 1^{-}$ limit.
In this limit $(1-q^z)/(1-q^{\gamma})$ becomes the binomial 
element $z/\gamma$ and the domain of integration becomes
$C_{\gamma}^{k_1,\dots,k_n}[0,1]$, exactly as in the proof of the $\A_n$
Selberg integral (cf.~\cite[\S 5]{Warnaar09}).
The resulting $\A_n$ Selberg average is the $\ell=k_1$ case of
Theorem~\ref{Thm_An-AFLT}.
It is a trivial exercise to verify that the right-hand side of
\eqref{Eq_An-AFLT} is independent of the choice of $\ell$, as long
as $\ell\geq l(\la)$.

\subsection{A companion to the $\A_n$ AFLT integral}

In \cite{Warnaar10} the $n=2$ case of Theorem~\ref{Thm_An-Cauchy} with 
$\mur{2}=0$ is employed to prove an $\A_2$ Selberg integral with two Jack 
polynomials in the integrand, but in a different form to that of the $\A_2$
AFLT integral (Theorem~\ref{Thm_An-AFLT} for $n=2$).
The two integrals differ in that the argument of the second Jack polynomial 
in \cite[Theorem~3.1]{Warnaar10} is simply the alphabet $\tar{2}$ with 
cardinality $k_2$ and there is 
an additional parameter $\beta_1$ subject to $\beta_1+\beta_2=\gamma+1$ (here
$\beta_2$ is the $\beta$ of the $\A_2$ AFLT integral).
By the rank-$n$ case of Theorem~\ref{Thm_An-Cauchy} we obtain
an $\A_n$ analogue of \cite[Theorem~3.1]{Warnaar10} described below.

For $\alpha_1,\dots,\alpha_n,\beta_{n-1},\beta_n,\gamma\in\mathbb{C}$ 
such that
\begin{subequations}\label{Eq_conditions-alt}
\begin{equation}
\beta_{n-1}+\beta_n=\gamma+1,
\end{equation}
\begin{equation}\label{Eq_betanmin1}
\beta_{n-1}+(i-k_n-1)\gamma\not\in\mathbb{Z} 
\quad \text{for $1\leq i \leq\min\{k_{n-1},k_n\}$},
\end{equation}
and
\begin{align}
&\Re(\gamma)>-\frac{1}{\max\{k_{n-1},k_n\}}, \quad
\Re\big(\beta_r+(i-k_{r+1}-1)\gamma\big)>0
\quad \text{for $1\leq r\leq n$ and $1\leq i\leq k_r$}, \\
&\Re\big(\alpha_r+\dots+\alpha_s+(r-s+i-1)\gamma\big)>0 \\
&\qquad\qquad \text{for } \notag
\begin{cases}
\text{$1\leq r\leq s\leq n-1$ and $1\leq i\leq k_r-k_{r-1}$}, \\[1mm]
\text{$1\leq r\leq n-1$, $s=n$ and 
$1\leq i\leq \min\{k_n,k_r-k_{r-1}\}$}, \\[1mm]
\text{$r=s=n$ and $1\leq i\leq k_n$},
\end{cases}
\end{align}
where $k_{n+1}:=0$ and $\beta_1=\dots=\beta_{n-2}:=1$,
\end{subequations}
we modify the $\A_n$ Selberg average \eqref{Eq_An-Selberg-average} to
\begin{equation}\label{Eq_average-alt}
\big\langle \mathscr{O}\big\rangle_
{\alpha_1,\dots,\alpha_n,\beta_{n-1},\beta_n;\gamma}^{k_1,\dots,k_n}:=
\frac{I^{\A_n}_{k_1,\dots,k_n}
(\mathscr{O};\alpha_1,\dots,\alpha_n,\beta_{n-1},\beta_n;\gamma)}
{I^{\A_n}_{k_1,\dots,k_n}
(1;\alpha_1,\dots,\alpha_n,\beta_{n-1},\beta_n;\gamma)}.
\end{equation}
Here
\begin{align*}
I^{\A_n}_{k_1,\dots,k_n}&
(\mathscr{O};\alpha_1,\dots,\alpha_n,\beta_{n-1},\beta_n;\gamma) \\
&:=\Int_{C^{k_1,\dots,k_n}_{\beta_{n-1},\gamma}[0,1]}
\mathscr{O}\big(\tar{1},\dots,\tar{n}\big)
\prod_{r=1}^n \prod_{i=1}^{k_r}\big(\tar{r}_i\big)^{\alpha_r-1}
\big(1-\tar{r}_i\big)^{\beta_r-1} \\[-1mm]
&\qquad\qquad\qquad\quad\times
\prod_{r=1}^n \Abs{\Delta\big(\tar{r}\big)}^{2\gamma}
\prod_{r=1}^{n-1} \Abs{\Delta\big(\tar{r},\tar{r+1}\big)}^{-\gamma}\, 
\dup\tar{1}\cdots\dup\tar{n}
\end{align*}
and $C^{k_1,\dots,k_n}_{\beta;\gamma}[0,1]$ is the following
$\beta$-deformation of the chain defined in Section~\ref{Sec_chain}. 
Let 
\[
E^{k_1,\dots,k_n}[0,1]\subseteq [0,1]^{k_1+\dots+k_n}
\]
be the set of points
\[
\big(\tar{1},\dots,\tar{n}\big)\in[0,1]^{k_1+\dots+k_n}
\]
such that
\[
0<\tar{r}_1<\dots<\tar{r}_{k_r}<1 \quad \text{for $1\leq r\leq n$}
\]
and
\[
\tar{r}_i<\tar{r+1}_{i-k_r+k_{r+1}} \quad 
\text{for $1\leq i\leq k_r$, $1\leq r\leq n-2$}.
\]
$E^{k_1,\dots,k_n}[0,1]$ differs from the set $D^{k_1,\dots,k_n}[0,1]$
only in that the relative ordering between the variables $\tar{n-1}$ and 
$\tar{n}$ has been removed. 
Accordingly we replace the sum over the maps $M_{n-1}$ by a sum over maps 
\[
M_{n-1}':\{1,\dots,k_n\}\longrightarrow\{1,\dots,k_{n+1}+1\}
\]
subject to $M_{n-1}'(i)\leq M_{n-1}'(i+1)$ for $1\leq i\leq k_{n-1}$
such that
\begin{equation}\label{Eq_chain-alt}
\tar{n}_{M_{n-1}'(i)-1}<\tar{n-1}_i<\tar{n}_{M_{n-1}'(i)},
\end{equation}
where $\tar{n}_0:=0$ and $\tar{n}_{k_n+1}:=1$.
We then define 
$E_{M_1,\dots,M_{n-2};M_{n-1}'}^{k_1,\dots,k_n}[0,1]\subseteq
E^{k_1,\dots,k_n}[0,1]$ by requiring that \eqref{Eq_chain-an} holds for
$M_1,\dots,M_{n-2}$ and \eqref{Eq_chain-alt} holds for $M_{n-1}$.
Hence
\[
E^{k_1,\dots,k_n}[0,1]=\sum_{M_1,\dots,M_{n-2},M'_{n-1}}
E_{M_1,\dots,M_{n-2};M_{n-1}'}^{k_1,\dots,k_n}[0,1].
\]
We also replace the weight function \eqref{Eq_weight} by
\[
G_{M_1,\dots,M_{n-2},M_{n-1}'}^{k_1,\dots,k_n}(\gamma)
:=F_{M_1,\dots,M_{n-2}}^{k_1,\dots,k_{n-1}}(\gamma)
\prod_{i=1}^{k_{n-1}}
\frac{\sin(\pi(\beta-(i+k_n-k_{n-1}-M_{n-1}'(i)+1)\gamma))}
{\sin(\pi(\beta-(i+k_n-k_{n-1})\gamma))}.
\]
Note that the condition \eqref{Eq_betanmin1} is necessary for this
weight function to be free of poles.
The new chain is then defined as
\[
C^{k_1,\dots,k_n}_{\beta;\gamma}[0,1]
:=\sum_{M_1,\dots,M_{n-2},M_{n-1}'}
G_{M_1,\dots,M_{n-2},M_{n-1}'}^{k_1,\dots,k_n}(\gamma)
D_{M_1,\dots,M_{n-1}}^{k_1,\dots,k_n}[0,1].
\]

We are now ready to state the counterpart to Theorem~\ref{Thm_An-AFLT}.

\begin{theorem}\label{Thm_An-alt}
For $n\geq 2$, let $k_1,\dots,k_n$ be nonnegative integers 
such that $k_1\leq\cdots\leq k_{n-1}$.
Then for $\alpha_1,\dots,\alpha_n,\beta_{n-1},\beta_n,\gamma\in\mathbb{C}$
such that \eqref{Eq_conditions-alt} holds and 
$\la,\mu\in\mathscr{P}$, we have
\begin{align}\label{Eq_An-AFLT-alt}
&\Big\langle P_{\la}^{(1/\gamma)}\big(\tar{1}\big)
P_{\mu}^{(1/\gamma)}\big(\tar{n}\big)
\Big\rangle_{\alpha_1,\dots,\alpha_n,\beta_{n-1},\beta_n;\gamma}
^{k_1,\dots,k_n} \\[1mm]
&\qquad =P_{\la}^{(1/\gamma)}[k_1] P_{\mu}^{(1/\gamma)}[k_n] \,\notag 
\prod_{r=1}^{n-1}
\frac{(\alpha_1+\dots+\alpha_r+(k_1-r)\gamma;\gamma)_{\la}}
{(\alpha_1+\dots+\alpha_r+\beta_r+
(k_1+k_r-k_{r+1}-r-1)\gamma;\gamma)_{\la}} \notag \\[1mm]
&\qquad\quad\times\prod_{r=2}^n
\frac{(\alpha_r+\dots+\alpha_n+(k_n+r-n-1)\gamma;\gamma)_{\mu}}
{(1+\alpha_r+\dots+\alpha_n-\beta_{r-1}+
(k_r-k_{r-1}+k_n+r-n-1)\gamma;\gamma)_{\mu}}\notag \\
&\qquad\quad\times
\prod_{i=1}^{k_1}\prod_{j=1}^{k_n}
\frac{(\alpha_1+\dots+\alpha_n+(k_1+k_n-n-i-j+1)\gamma)_{\la_i+\mu_j}}
{(\alpha_1+\dots+\alpha_n+(k_1+k_n-n-i-j+2)\gamma)_{\la_i+\mu_j}} \notag
\end{align}
and
\begin{align}\label{Eq_An-alt}
&I^{\A_n}_{k_1,\dots,k_n}
(1;\alpha_1,\dots,\alpha_n,\beta_{n-1},\beta_n;\gamma) \\
&\quad=\prod_{r=1}^n\prod_{i=1}^{k_r}
\frac{\Gamma(\beta_r+(i-k_{r+1}-1)\gamma)\Gamma(i\gamma)}
{\Gamma(\gamma)} \notag \\
&\qquad\times \prod_{1\leq r\leq s\leq n-1}\prod_{i=1}^{k_r-k_{r-1}}
\frac{\Gamma(\alpha_r+\dots+\alpha_s+(r-s+i-1)\gamma)}
{\Gamma(\alpha_r+\dots+\alpha_s+\beta_s+(k_s-k_{s+1}+r-s+i-2)\gamma)} 
\notag \\
&\qquad\times\prod_{r=1}^n
\prod_{i=1}^{k_n}
\frac{\Gamma(\alpha_r+\dots+\alpha_n+(r-n+i-1)\gamma)}
{\Gamma(1+\alpha_r+\dots+\alpha_n-\beta_{r-1}+
(k_r-k_{r-1}+r-n+i-1)\gamma)}, \notag
\end{align}
where $k_0=k_{n+1}:=0$ and $\beta_0=\dots=\beta_{n-2}:=1$.
\end{theorem}

It should be noted that the final product in \eqref{Eq_An-alt} may 
alternatively be expressed as
\[
\prod_{r=1}^{n-1} \prod_{i=1}^{k_r-k_{r-1}}
\frac{\Gamma(\alpha_r+\dots+\alpha_n+(r-n+i-1)\gamma)}
{\Gamma(\alpha_r+\dots+\alpha_n+(k_n+r-n+i-1)\gamma)}
\prod_{i=1}^{k_n}
\frac{\Gamma(\alpha_n+(i-1)\gamma)}
{\Gamma(\alpha_n+\beta_n+(k_n-k_{n-1}+i-2)\gamma)}.
\]
When $\beta=\gamma$ in \eqref{Eq_An-AFLT} and 
$(\beta_{n-1},\beta_n)=(1,\gamma)$ in \eqref{Eq_An-AFLT-alt}
both integral evaluations coincide.
For $k_n=0$ equation \eqref{Eq_An-AFLT-alt} simplifies to the
$\A_{n-1}$ analogue of Kadell's integral \cite[Theorem~6.1]{Warnaar09}.

\begin{proof}
We start with \eqref{Eq_An-Cauchy-1} with $k_n$ finite and, 
for $\la,\mu\in\mathscr{P}$ with $l(\la)\leq k_1$ and $l(\mu)\leq k_n$, 
make the substitutions
\[
\big(X^{(1)},Y^{(n)},a_{n-1},t,\mur{n}\big)\mapsto
\big(\spec{\la}_{k_1},q^{\alpha_n}\spec{\mu}_{k_n},
q^{\beta_{n-1}+(k_{n-1}-k_n)\gamma},q^\gamma,0\big)
\]
and
\[
z_r\mapsto q^{\alpha_r-\gamma} \qquad \text{for $1\leq r\leq n-1$}.
\]
Then the resulting sum may be turned into an integral following the steps 
outlined in Section~\ref{Sec_An-AFLT-proof}.
For $\la=\mu=0$ this yields \eqref{Eq_An-alt} and for general $\la$ and $\mu$
it gives 
\[
\Big\langle P_{\la}^{(1/\gamma)}\big(\tar{1}\big)
P_{\mu}^{(1/\gamma)}\big(\tar{n}\big)
\Big\rangle_{\alpha_1,\dots,\alpha_n,\beta_{n-1},\beta_n;\gamma}
^{k_1,\dots,k_n} \times 
I^{\A_n}_{k_1,\dots,k_n}
(1;\alpha_1,\dots,\alpha_n,\beta_{n-1},\beta_n;\gamma). \qedhere
\]
\end{proof}

\section{Complex Schur functions and Selberg integrals}
\label{Sec_An-AFLT-Schur}

In this section we introduce a complex analogue of the
Schur function and show how complex Schur functions may be
utilised to prove $\A_n$ Selberg-type integrals, such as
Theorem~\ref{Thm_nplusone}.
We should remark that Kadell already observed in \cite{Kadell00} that 
for $\gamma=1$ the evaluations of the 
Kadell and Hua--Kadell integrals remain valid if one replaces
the Schur functions in the integrand by Schur functions indexed
by sequences of complex numbers.
His paper does not, however, provide the necessary tools to attack
integrals such as \eqref{Eq_nplusone}, and we will not use any
of his results for complex Schur functions, such as Pieri and
branching rules.

\subsection{Complex Schur functions}

In the following we fix the principal branch of the complex
logarithm with cut along the negative real axis and argument in
$(-\pi,\pi]$.
Accordingly we denote the cut or slit complex plane 
$\abs{\Im(\log(x))}<\pi$ by $\Omega$.

For $x=(x_1,\dots,x_n)\in\Omega^n$ and $z=(z_1,\dots,z_n)\in\mathbb{C}^n$,
we define the complex Schur function\footnote{Our definition differs
slightly from that of Kadell \cite{Kadell00}.}
\[
S^{(n)}(x;z):=\frac{\det_{1\leq i,j\leq n}\big(x_i^{z_j}\big)}{\Delta(x)},
\]
where $\Delta(x)$ is the Vandermonde product \eqref{Eq_Vandermonde}.
Clearly, for $\la\in\mathscr{P}_n$,
\begin{equation}\label{Eq_stoS}
s_{\la}(x_1,\dots,x_n)=
S^{(n)}(x_1,\dots,x_n;\la_1+n-1,\la_2+n-2,\dots,\la_n).
\end{equation}
By removing the singularities at $x_i=x_j$ we extend $S^{(n)}$ to a 
holomorphic function on $\Omega^n$. 
Since the complex Schur function is symmetric in $x$, we may employ the usual 
plethystic notation and we will sometimes write $S^{(n)}([x];z)$ instead of 
$S^{(n)}(x;z)$, where of course $x$ should always be an alphabet of 
cardinality $n$.

It follows immediately from the determinantal structure of
$S^{(n)}$ that the following expansion holds.

\begin{lemma}\label{Lem_StoSS}
For any $0\leq m\leq n$, we have
\begin{multline*}
S^{(n)}(x_1,\dots,x_n;z_1,\dots,z_n) \\
=\sum_{\substack{I\subseteq \{1,\dots,n\} \\[1pt] \abs{I}=m}}
\frac{S^{(m)}\big(\big[\sum_{i\in I} x_i\big];z_1,\dots,z_m\big)
S^{(n-m)}\big(\big[\sum_{i\notin I}x_i\big];z_{m+1},\dots,z_n\big)}
{\prod_{i\in I}\prod_{j\notin I}(x_i-x_j)}.
\end{multline*}
\end{lemma}

Like the ordinary Schur function, the complex Schur function satisfies a 
simple specialisation formula.

\begin{lemma}
We have
\begin{equation}\label{Eq_Sn-PS}
S^{(n)}(\underbrace{1,\dots,1}_{\text{$n$ times}};z_1,\dots,z_n)=
S^{(n)}\big([n];z_1,\dots,z_n\big)=
\prod_{1\leq i<j\leq n}\frac{z_i-z_j}{j-i}.
\end{equation}
\end{lemma}

\begin{proof}
Since $S^{(n)}([n];z_1,\dots,z_n)$ is a polynomial in $z$, the
claim follows from the specialisation formula \eqref{Eq_Schur-spec}
and the fact that for arbitrary $\la\in\mathscr{P}_n$,
\[
S^{(n)}\big([n];\la_1+n-1,\la_2+n-2,\dots,\la_n\big)=s_{\la}[n]. \qedhere
\]
\end{proof}

For $0<\theta<\pi$ and $r>0$, let $C_{\theta,r}$ denote the contour 
in $\mathbb{C}$ going counterclockwise around the border of the angular
sector $\abs{x}\leq r$, 
$\abs{\textup{Im}(\log(x))}\leq\theta$
as shown below.

\medskip

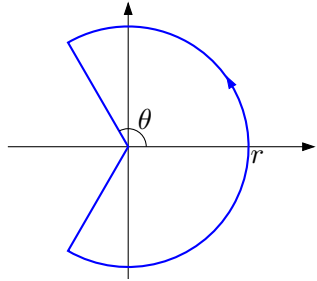
\tikzset{arrowright/.pic={
\filldraw (0,-0.05) -- (0.15,0) -- (0,0.05) -- cycle;}}
\tikzset{arrowup/.pic={
\filldraw (-0.05,0) -- (0,0.15) -- (0.05,0) -- cycle;}}

\begin{center}
\begin{tikzpicture}[scale=0.8]
\draw[thin] (-2,0)--(3,0);
\draw[thin] (0,-2.2)--(0,2.2);
\pic at (2.9,0) {arrowright};
\pic at (0,2.2) {arrowup};
\pic[rotate=30,blue] at (1.732,1) {arrowup};
\draw[blue,thick] (0,0) -- (-1,-1.732) arc (-120:120:2) -- cycle;
\draw[thin] (0.3,0) arc (0:120:0.3);
\draw (0.28,0.46) node {$\theta$};
\draw (2.15,-0.16) node {$r$};
\end{tikzpicture}
\end{center}

\noindent
Then the complex Schur function satisfies the following
integral identity.

\begin{theorem}\label{Thm_Schur}
For $\ell$ a nonnegative integer, let $y=(y_1,\dots,y_{\ell})
\in\Omega^{\ell}$, and let $0<\theta<\pi$, $r>0$ such that 
$y_i\in\interior(C_{\theta,r})$ for all $1\leq i\leq \ell$.
For $k$ a nonnegative integer, let $z=(z_1,\dots,z_k)\in\Omega^k$
such that $\Re(z_i)>-1$ for all $1\leq i\leq k$.
Then, for $\la\in\mathscr{P}$,
\begin{align}\label{Eq_S-integral}
\frac{1}{k!(2\pi\iup)^k} &
\Int_{C_{\theta,r}^k} S^{(k)}(x;z)s_{\la}[y-x] 
\prod_{1\leq i<j\leq k} (x_i-x_j)^2
\prod_{i=1}^k \prod_{j=1}^{\ell} (x_i-y_j)^{-1} \, \dup x_1\cdots \dup x_k 
\\[2mm]
&=\begin{cases}
(-1)^{\binom{k}{2}}
S^{(\ell)}\big(y;(z,\la_1+\ell-k-1,\dots,\la_{\ell-k-1}+1,\la_{\ell-k})\big)
& \text{if $l(\la)\leq\ell-k$}, \\[3mm]
0 & \text{otherwise},
\end{cases} \notag
\end{align}
where $C_{\theta,r}^k$ denotes the $k$-fold product
$C_{\theta,r}\times\dots\times C_{\theta,r}$.
\end{theorem}

We repeat our remark from the introduction
that the above integral should be understood in the
sense of indefinite integrals since the integrand is not 
defined at $(0,\dots,0)\in\mathbb{C}^k$.

\begin{proof}
The integral in \eqref{Eq_S-integral} is continuous in $y$, 
so that it suffices to prove the claim on the dense subset of 
$\interior(C_{\theta,r}^k)$ for which the $y_i$ are all distinct.
In particular, the first three factors of the integrand are holomorphic
on $\Omega^k$, so that the integrand has only simple poles along the
divisors $x_i=y_j$.

To compute the integral we proceed recursively. 
For $m$ an integer such that $0\leq m\leq k$, define 
\begin{align*}
\mathscr{I}_m^{(k,\ell)}(y)&:=
\frac{(-1)^{\binom{m}{2}}m!}{k!}
\sum_{\substack{I\subseteq \{1,\dots,\ell\} \\[1pt] \abs{I}=m}} 
\Bigg(\; \prod_{i\in I}\prod_{j\notin I}(y_i-y_j)^{-1} \\
&\quad\times 
\frac{1}{(2\pi\iup)^{k-m}}
\Int_{C_{\theta,r}^{k-m}}
S^{(k)}\bigg(\bigg[\,\sum_{i\in I} y_i+\sum_{i=m+1}^k x_i\bigg];z\bigg) 
s_{\la}\bigg(\bigg[\,\sum_{i\notin I} y_i-\sum_{i=m+1}^k x_i\bigg]\bigg) \\
&\qquad\qquad\qquad\qquad\qquad
\times \prod_{m+1\leq i<j\leq k} (x_i-x_j)^2 
\prod_{i=m+1}^k\frac{\prod_{j\in I} (x_i-y_j)}
{\prod_{j\notin I} (x_i-y_j)} \; \dup x_{m+1}\cdots \dup x_k\Bigg),
\end{align*}
where we have suppressed the dependence on $z$ and $\la$.
Observe that $\mathscr{I}_0^{(k,\ell)}(y)$ coincides with the
left-hand side of \eqref{Eq_S-integral} and that 
$\mathscr{I}_m^{(k,\ell)}(y)=0$ if $m>\ell$.

Now consider $\mathscr{I}_m^{(k,\ell)}(y)$ for some fixed $0\leq m\leq k-1$.
For a given term in the summand indexed by $I$, we compute the integral 
over $x_{m+1}$ by shrinking the radius $r$ of the corresponding contour 
$C_{\theta,r}$ to $r_0$. 
Here $r_0$ is sufficiently small so that
$y_j\in\exterior(C_{\theta,r_0})$ for all $j\notin I$. 
As a result, the integral over $x_{m+1}$ is expressed as
a sum over the residues in $x_{m+1}$ at the $\ell-m$ (distinct) poles 
$x_{m+1}=y_j$ for $j\notin I$, plus a remainder $x_{m+1}$-integral over 
$C_{\theta,r_0}$.
Since the integrand grows slower than $1/\abs{x_{m+1}}$ as $x_{m+1}$ 
approaches $0$ in the angular sector, this remainder converges to $0$ as 
$r_0\to 0$, and hence (since it is independent of $r_0$) is identically
$0$.
Thus
\begin{align*}
\mathscr{I}_m^{(k,\ell)}(y)&=\frac{(-1)^{\binom{m+1}{2}}m!}{k!}
\sum_{\substack{I\subseteq \{1,\dots,\ell\} \\[1pt] \abs{I}=m}} 
\sum_{r\notin I}
\Bigg(\; \prod_{i\in I\cup\{r\}}\prod_{j\notin I\cup\{r\}}
(y_i-y_j)^{-1} \\
&\quad\times 
\frac{1}{(2\pi\iup)^{k-m-1}}
\Int_{C_{\theta,r}^{k-m-1}}
S^{(k)}\bigg(\bigg[\,\sum_{i\in I\cup\{r\}} y_i+
\sum_{i=m+2}^k x_i\bigg];z\bigg) 
s_{\la}\bigg(\bigg[\,\sum_{i\notin I\cup\{r\}} y_i-
\sum_{i=m+2}^k x_i\bigg]\bigg) \\
&\quad\qquad\quad\qquad\qquad\qquad\times 
\prod_{m+2\leq i<j\leq k} (x_i-x_j)^2 
\prod_{i=m+2}^k\frac{\prod_{j\in I\cup\{r\}} (x_i-y_j)}
{\prod_{j\notin I\setminus\{r\}} (x_i-y_j)} \; 
\dup x_{m+2}\cdots\dup x_k\Bigg),
\end{align*}
which of course vanishes for $m\geq\ell$. 
Vanishing or not, defining $J:=I\cup\{r\}$ and noting that there are
exactly $m+1$ ways an $(m+1)$-set $J$ can be decomposed into an $m$-set and
a singleton, we obtain
$\mathscr{I}_m^{(k,\ell)}(y)=\mathscr{I}_{m+1}^{(k,\ell)}(y)$.
By iteration this yields
$\mathscr{I}_0^{(k,\ell)}(y)=\mathscr{I}_k^{(k,\ell)}(y)$, so that
\[
\mathscr{I}_0^{(k,\ell)}(y)=(-1)^{\binom{k}{2}}
\sum_{\substack{I\subseteq \{1,\dots,\ell\} \\[1pt] \abs{I}=k}} 
\frac{S^{(k)}\big(\big[\sum_{i\in I} y_i\big];z\big) 
s_{\la}\big(\big[\sum_{i\notin I} y_i\big]\big)}
{\prod_{i\in I}\prod_{j\notin I}(y_i-y_j)}.
\]
If $l(\la)>\ell-k$ this vanishes and if
$l(\la)\leq \ell-k$ it follows from \eqref{Eq_stoS} and
Lemma~\ref{Lem_StoSS} that
\[
\mathscr{I}_0^{(k,\ell)}(y)=(-1)^{\binom{k}{2}}\,
S^{(\ell)}\big(y;(z,\la_1+\ell-k-1,\dots,\la_{\ell-k-1}+1,\la_{\ell-k})\big).
\qedhere
\]
\end{proof}

To prove Theorem~\ref{Thm_nplusone} we not only need
Theorem~\ref{Thm_Schur} but also a generalisation of the
$y=(1,\dots,1)\in\mathbb{C}^{\ell}$ (or, plethystically,
$y=\ell$) case of this theorem.
For this we require a variant of the classical beta integral.

\begin{lemma}\label{Lem_beta-integral}
Let $\alpha,\beta\in\mathbb{C}$ such that $\Re(\alpha)>0$.
Then, for $r>1$ and $0<\theta<\pi$,
\[
\frac{1}{2\pi\iup}\Int_{C_{\theta,r}} x^{\alpha-1}(x-1)^{\beta-1}\,\dup x =
\frac{\Gamma(\alpha)}{\Gamma(1-\beta)\Gamma(\alpha+\beta)}.
\]
\end{lemma}

\begin{proof}
By holomorphicity, we may assume without loss of generality that 
$\Re(\beta)>0$ in the following.

Recall that $r>1$, so that the branch point $x=1$ is contained in
the interior of $C_{\theta,r}$. 
Deforming this contour to $\overline{C}_{\theta,r}$ given by
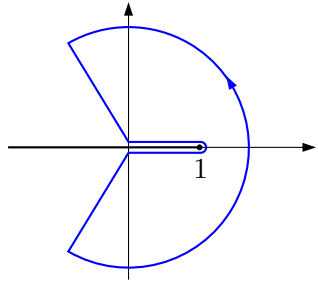
\begin{center}
\begin{tikzpicture}[scale=0.8]
\draw[thick] (-2,0)--(1.2,0);
\draw[thin] (-2,0)--(3,0);
\draw[thin] (0,-2.2)--(0,2.2);
\pic at (2.9,0) {arrowright};
\pic at (0,2.2) {arrowup};
\pic[rotate=30,blue] at (1.732,1) {arrowup};
\draw[blue,thick] (0,-0.09) -- (-1,-1.732) arc (-120:120:2) -- (0,0.09)
-- (1.2,0.09) arc (90:-90:0.09) -- cycle;
\filldraw (1.18,0) circle (0.04);
\draw (1.2,-0.35) node {$1$};
\end{tikzpicture}
\end{center}
we get, by slight abuse of notation,
\begin{align*}
\Int_{C_{\theta,r}} x^{\alpha-1}(x-1)^{\beta-1}\,\dup x 
&= 
\Bigg(\:\Int_{\overline{C}_{\theta,r}}-\lim_{r_0\to 0} \Int_{C_{r_0}}\Bigg) 
x^{\alpha-1}(x-1)^{\beta-1}\,\dup x \\[1.5mm]
&\quad -\bigg(\eup^{\pi\iup (\beta-1)} \int_0^1{+}\,
\eup^{-\pi\iup (\beta-1)}\int_0^1 \bigg) 
x^{\alpha-1}\abs{x-1}^{\beta-1}\dup x,
\end{align*}
where $C_{r_0}$ is the positively oriented semicircle $1+r_0\eup^{\iup\phi}$,
$\phi\in(-\pi/2,\pi/2)$ of radius $r_0$ around $1$.
The integral over the deformed contour $\overline{C}_{\theta,r}$ trivially
vanishes. 
Since $\Re(\beta)>0$, so does the integral over $C_{r_0}$ in the $r_0\to 0$ 
limit. 
By the standard Euler beta integral \cite[p.~5]{AAR99}
\[
\int_0^1 x^{\alpha-1}(1-x)^{\beta-1}\,\dup x=
\frac{\Gamma(\alpha)\Gamma(\beta)}{\Gamma(\alpha+\beta)}
\]
for $\Re(\alpha),\Re(\beta)>0$, we thus find
\[
\frac{1}{2\pi\iup}
\Int_{C_{\theta,r}} x^{\alpha-1}(x-1)^{\beta-1}\,\dup x 
=\frac{\sin(\pi\beta)}{\pi}\:
\frac{\Gamma(\alpha)\Gamma(\beta)}{\Gamma(\alpha+\beta)}.
\]
The claim now follows by the reflection formula for the gamma function.
\end{proof}

\begin{theorem}\label{Thm_beta-Schur}
Let $\beta\in\mathbb{C}$ and $z=(z_1,\dots,z_k)\in\mathbb{C}^k$ such that 
$\Re(z_i)>-1$ for all $i$.
Then, for $r>1$, $0<\theta<\pi$ and $\la\in\mathscr{P}$,
\begin{align}\label{Eq_beta-Schur}
\frac{1}{(2\pi\iup)^k}
\Int_{C_{\theta,r}^k}& S^{(k)}(x;z)s_{\la}[1-\beta-x] 
\prod_{1\leq i<j\leq k} (x_i-x_j)^2
\prod_{i=1}^k (x_i-1)^{\beta-1} \, \dup x_1\cdots \dup x_k \\
&=(-1)^{\binom{k}{2}} S^{(k)}([k];z) s_{\la}[1-\beta-k] \notag \\[2mm]
& \quad \times
\prod_{i=1}^k \bigg(\frac{i!\,\Gamma(z_i+1)}
{\Gamma(2-i-\beta)\Gamma(z_i+\beta+k)}
\prod_{j\geq 1} \frac{z_i-\la_j+\beta+k+j-1}{z_i+\beta+k+j-1}\bigg). \notag
\end{align}
\end{theorem}

\begin{proof}
First assume that $\beta=1-\ell$ with $\ell$ an integer.
Then \eqref{Eq_beta-Schur} simplifies to
\begin{align}\label{Eq_integer-case}
\frac{1}{(2\pi\iup)^k}
\Int_{C_{\theta,r}^k}& S^{(k)}(x;z)s_{\la}[\ell-x] 
\prod_{1\leq i<j\leq k} (x_i-x_j)^2
\prod_{i=1}^k (x_i-1)^{-\ell} \, \dup x_1\cdots \dup x_k \\
&=(-1)^{\binom{k}{2}} S^{(k)}([k];z) s_{\la}[\ell-k] \notag \\[2mm]
& \quad \times \prod_{i=1}^k \bigg( \frac{i!\,\Gamma(z_i+1)}
{\Gamma(\ell-i+1)\Gamma(z_i-\ell+k+1)}
\prod_{j\geq 1} \frac{z_i-\la_j-\ell+k+j}{z_i-\ell+k+j}\bigg) \notag \\
&=k! (-1)^{\binom{k}{2}} S^{(k)}([k];z) s_{\la}[\ell-k] 
\prod_{i=1}^k \prod_{j=1}^{\ell-k} \frac{z_i-\la_j-\ell+k+j}{k+j-i}.  \notag 
\end{align}
By the principal specialisation formulas \eqref{Eq_Schur-spec} and
\eqref{Eq_Sn-PS} it follows that
\begin{multline}\label{Eq_double-PS}
S^{(k)}([k];z)s_{\la}[\ell-k]
\prod_{i=1}^k \prod_{j=1}^{\ell-k} \frac{z_i-\la_j-\ell+k+j}{k+j-i} \\
=\begin{cases}
S^{(\ell)}\big([\ell];
(z,\la_1+\ell-k-1,\dots,\la_{\ell-k-1}+1,\la_{\ell-k})\big)
& \text{if $l(\la)\leq\ell-k$}, \\[3mm]
0 & \text{otherwise},
\end{cases}
\end{multline}
so that \eqref{Eq_integer-case} is precisely \eqref{Eq_S-integral} 
for $y=1+\dots+1=\ell$.

One possible approach to showing the general complex $\beta$ case would 
be to appeal to Carlson's theorem.
To avoid having to show the required bounded growth at infinity of the
left-hand side of \eqref{Eq_beta-Schur}, we will more simply show that
after appropriate normalisation both sides of \eqref{Eq_beta-Schur}
become rational functions.
As a first step towards proving this we make the substitution 
$\la\mapsto\la'$ using \eqref{Eq_Schur-duality} and
\[
\prod_{j\geq 1}\frac{w-\la'_j+j-1}{w+j-1}=
\prod_{j\geq 1}\frac{w-j}{w+\la_j-j}
\]
to obtain the equivalent identity
\begin{align}\label{Eq_dualize}
\frac{1}{(2\pi\iup)^k}
\Int_{C_{\theta,r}^k} & S^{(k)}(x;z)s_{\la}[x+\beta-1] 
\prod_{1\leq i<j\leq k} (x_i-x_j)^2
\prod_{i=1}^k (x_i-1)^{\beta-1} \, \dup x_1\cdots \dup x_k \\
&=(-1)^{\binom{k}{2}} S^{(k)}([k];z) s_{\la}[\beta+k-1] \notag \\[2mm]
&\quad\times 
\prod_{i=1}^k \frac{i!\,\Gamma(z_i+1)}{\Gamma(2-i-\beta)\Gamma(z_i+\beta+k)}
\prod_{i=1}^k \prod_{j\geq 1} \frac{z_i+\beta+k-j}{z_i+\la_j+\beta+k-j}. \notag
\end{align}
By \cite[p. 72]{Macdonald95}
\[
s_{\la}[X+Y]=\sum_{\mu}s_{\la/\mu}[X]s_{\mu}[Y]
\]
and the determinantal nature of both types of Schur functions in the 
integrand, the left hand side of \eqref{Eq_dualize} may also be written as
\begin{align*}
k! \sum_{\mu} & s_{\la/\mu}[\beta-1]\,
\frac{1}{(2\pi\iup)^k}
\Int_{C_{\theta,r}^k} \det_{1\leq i,j\leq k}
\big( x_i^{\mu_j+k-j}\big)
\prod_{i=1}^k x_i^{z_i} (x_i-1)^{\beta-1} \, \dup x_1\cdots \dup x_k \\
&=k! \sum_{\mu} s_{\la/\mu}[\beta-1]
\det_{1\leq i,j\leq k} \Bigg(\: \frac{1}{2\pi\iup}
\Int_{C_{\theta,r}} x^{z_i+\mu_j+k-j}(x-1)^{\beta-1}\, \dup x \Bigg) \\
&=k! \sum_{\mu} s_{\la/\mu}[\beta-1]
\det_{1\leq i,j\leq k} \bigg( \frac{\Gamma(z_i+\mu_j+k-j+1)}
{\Gamma(1-\beta)\Gamma(z_i+\mu_j+\beta+k-j+1)} \bigg) \\
&=k! \sum_{\mu} s_{\la/\mu}[\beta-1]
\prod_{i=1}^k \frac{\Gamma(z_i+1)}{\Gamma(1-\beta)\Gamma(z_i+\beta+1)} \:
\det_{1\leq i,j\leq k} \bigg( 
\frac{(z_i+1)_{\mu_j+k-j}}{(z_i+\beta+1)_{\mu_j+k-j}} \bigg),
\end{align*}
where the penultimate line follows from Lemma~\ref{Lem_beta-integral}.
This shows that the left-hand side of \eqref{Eq_dualize} divided by the
right-hand side is given by the following rational function:
\begin{multline*}
\frac{(-1)^{\binom{k}{2}}}{S^{(k)}([k];z)} 
\sum_{\mu} \frac{s_{\la/\mu}[\beta-1]}{s_{\la}[\beta+k-1]}
\prod_{i=1}^k\bigg(\frac{(z_i+\beta+1)_{k-1}}{(i-1)!(2-i-\beta)_{i-1}} 
\prod_{j\geq 1} \frac{z_i+\la_j+\beta+k-j}{z_i+\beta+k-j}\bigg) \\
\times \det_{1\leq i,j\leq k}
\bigg(\frac{(z_i+1)_{\mu_j+k-j}}{(z_i+\beta+1)_{\mu_j+k-j}}\bigg),
\end{multline*}
completing the proof.
We remark that for $\la=0$ the above collapses into
\[
\frac{(-1)^{\binom{k}{2}}}{S^{(k)}([k];z)} 
\prod_{i=1}^k\frac{(z_i+\beta+1)_{k-1}}{(i-1)!(2-i-\beta)_{i-1}}\:
\det_{1\leq i,j\leq k} 
\bigg(\frac{(z_i+1)_{k-j}}{(z_i+\beta+1)_{k-j}}\bigg).
\]
In this instance the determinant follows as a special case of a well-known
determinant evaluation due to Krattenthaler, see e.g., 
\cite[Lemma~3]{Krattenthaler99} with 
$(n,X_i,A_i,B_i)\mapsto (k,z_i,k-i+1,\beta+k-i+1)$.
For general $\la$ it does not seem so easy to show by direct means
that the above rational function is in fact equal to $1$.
\end{proof}

\subsection{Proof of Theorem~\ref{Thm_nplusone}}

Recall definitions \eqref{Eq_A-def} of $A_r$ and $A_{r,s}$,
\eqref{Eq_I-def} of the integral 
\[
I^{\A_n}_{k_1,\dots,k_n}(\mathscr{O};\alpha_1,\dots,\alpha_n,\beta)
\]
and \eqref{Eq_contour} of the integration contour $C^{k_1,\dots,k_n}$.
In our next theorem, which may be viewed as a complex analogue of
Theorem~\ref{Thm_nplusone}, 
we will slightly extend the definition of the above integral to allow for 
functions $\mathscr{O}(\tar{1},\dots,\tar{n})$ that are not polynomial in
$\tar{1}$.

\begin{theorem}\label{Thm_complex-An-AFLT}
For $n$ a positive integer, let $k_1,\dots,k_n$ be nonnegative
integers, $\alpha_1,\dots,\alpha_n,\beta\in\mathbb{C}$, $z\in\mathbb{C}^{k_1}$
and $\lar{2},\dots,\lar{n+1}\in\mathscr{P}$, such that
\begin{align*}
\Re(z_i+\alpha_1+\dots+\alpha_s)>s-1 & \quad
\text{for $1\leq s\leq n$ and $1\leq i\leq k_1$}, \\
\Re\big(\lar{r}_{k_r-k_{r-1}}+\alpha_r+\dots+\alpha_s\big)>s-r & \quad
\text{for $2\leq r\leq s\leq n$}.
\end{align*}
Further let
$\tar{1},\dots,\tar{n}$ be alphabets of cardinality $k_1,\dots,k_n$,
and set
$\lar{1}_i:=z_i-k_1+i$ for $1\leq i\leq k_1$ and $\tar{n+1}:=1-\beta$.
Then,
\begin{align}\label{Eq_complex-An-AFLT}
&I^{\A_n}_{k_1,\dots,k_n}
\bigg(S^{(k_1)}\big(\tar{1};z\big) 
\prod_{r=2}^{n+1} s_{\lar{r}}\big[\tar{r}-\tar{r-1}\big] 
;\alpha_1,\dots,\alpha_n,\beta\bigg) \\
&\quad=S^{(k_1)}([k_1];z)
\prod_{r=1}^n \bigg(
(-1)^{\binom{k_r}{2}} s_{\lar{r+1}}[k_{r+1}-k_r] 
\prod_{i=1}^{k_r} \frac{i!}{\Gamma(k_{r+1}-i+1)} \bigg) \notag \\
& \quad\quad \times 
\prod_{1\leq r<s\leq n} 
\prod_{i=1}^{k_r-k_{r-1}} \prod_{j=1}^{k_s-k_{s-1}} 
\big(\lar{r}_i-\lar{s}_j+A_{r,s}+j-i\big) \notag \\
& \quad\quad \times
\prod_{r=1}^n \prod_{i=1}^{k_r-k_{r-1}} \bigg(
\frac{\Gamma(\lar{r}_i+A_r-i-n)}{\Gamma(\lar{r}_i+A_{r,n+1}-i+1)} 
\prod_{j\geq 1} 
\frac{\lar{r}_i-\lar{n+1}_j+A_{r,n+1}+j-i}{\lar{r}_i+A_{r,n+1}+j-i}\bigg),
\notag
\end{align}
where $\beta_1=\dots=\beta_{n-1}:=1$, $\beta_n:=\beta$, $k_0:=0$ and 
$k_{n+1}:=1-\beta$.
\end{theorem}

We note that 
\[
\prod_{r=1}^n\prod_{i=1}^{k_r} \frac{1}{\Gamma(k_{r+1}-i+1)}
\]
vanishes unless $k_{r+1}-k_r\geq 0$ for $1\leq r\leq n-1$, i.e.,
unless $k_1\leq k_2\leq\cdots\leq k_n$.
Assuming this holds it then follows from
\[
\prod_{r=1}^n s_{\lar{r+1}}[k_{r+1}-k_r]
\]
that the right-hand side vanishes unless
$l(\lar{r})\leq k_r-k_{r-1}$ for all $2\leq r\leq n$.

If we take $\lar{2}=\lar{n+1}=0$ in
Theorem~\ref{Thm_complex-An-AFLT} and also choose $z_i=k_1-i$ 
for all $1\leq i\leq k_1$ (so that on the right $\lar{1}=0$), 
we obtain
\begin{multline}\label{Eq_An-one}
I^{\A_n}_{k_1,\dots,k_n}(1;\alpha_1,\dots,\alpha_n,\beta) \\
=\prod_{r=1}^n \bigg( (-1)^{\binom{k_r}{2}} 
\prod_{i=1}^{k_r} \frac{i!}{\Gamma(k_{r+1}-i+1)} \bigg) 
\prod_{1\leq r<s\leq n+1} \prod_{i=1}^{k_r-k_{r-1}} 
\frac{\Gamma(A_{r,s}+k_s-k_{s-1}-i+1)}{\Gamma(A_{r,s}-i+1)} ,
\end{multline}
where the $\alpha_r$ must satisfy \eqref{Eq_conv}.
This is equivalent to \eqref{Eq_An-One} where it should be noted
that in the above $k_{n+1}$ is defined as $1-\beta$ whereas 
$k_{n+1}:=0$ in \eqref{Eq_An-One}. 

Let $\lar{1}\in\mathscr{P}_{k_1}$.
If we take \eqref{Eq_complex-An-AFLT} for arbitrary 
partitions $\lar{2},\dots,\lar{n+1}$, specialise
$z_i=\lar{1}_i+k_1-i$ for $1\leq i\leq k_1$
and divide the resulting identity by \eqref{Eq_An-one},
then Theorem~\ref{Thm_nplusone} follows with
$\ell_r$ given by $k_r-k_{r-1}$ for all $1\leq r\leq n$
and $\ell_{n+1}$ any integer such that $\ell_{n+1}\geq l(\lar{n+1})$.
Since the right-hand side of \eqref{Eq_nplusone} is independent
of the choice of $\ell_r$ provided that $\ell_r\geq l(\lar{r})$, 
this completes the proof of Theorem~\ref{Thm_nplusone}.

\begin{proof}[Proof of Theorem~\ref{Thm_complex-An-AFLT}]
We may deform the contour $C^{k_1,\dots,k_n}$ to
\[
C_{\theta_1,r_1}^{k_1}\times\dots\times C_{\theta_n,r_n}^{k_n},
\]
for arbitrary real $r_1,\dots,r_n$ and $\theta_1,\dots,\theta_n$ 
such that
$r_1>\dots>r_n>1$ and $\pi>\theta_1>\dots>\theta_n>0$, and in the
following we assume this new deformed contour.

We now proceed by induction on $n$.
The base case, $n=1$, corresponds to Theorem~\ref{Thm_beta-Schur} 
with $(k,\la,\beta)\mapsto (k_1,\lar{2},\beta_1=\beta)$
and $z_i\mapsto z_i+\alpha_1-1$ for all $1\leq i\leq k_1$.

To carry out the inductive step we denote the left-hand side of 
\eqref{Eq_complex-An-AFLT} by
\[
\mathscr{L}^{k_1,\dots,k_n}_{\lar{2},\dots,\lar{n+1}}
(z;\alpha_1,\dots,\alpha_n,\beta).
\]
Assuming $n\geq 2$, we integrate over $\tar{1}$ using 
Theorem~\ref{Thm_Schur} with 
$(k,\ell,x,y,\la)\mapsto (k_1,k_2,\tar{1},\tar{2},\lar{2})$ and
$z_i\mapsto z_i+\alpha_1-1$ for all $1\leq i\leq k_1$.
This yields
\[
\mathscr{L}^{k_1,\dots,k_n}_{\lar{2},\dots,\lar{n+1}}
(z;\alpha_1,\dots,\alpha_n,\beta)=
\begin{cases}
k_1! (-1)^{\binom{k_1}{2}} 
\mathscr{L}^{k_2,\dots,k_n}_{\lar{3},\dots,\lar{n+1}}
(z';\alpha_2,\dots,\alpha_n,\beta) & 
\text{if $l\big(\lar{2}\big)\leq k_2-k_1$},\\[3mm]
0 & \text{otherwise},
\end{cases}
\]
where
\[
z':=\big(z_1+\alpha_1-1,\dots,z_{k_1}+\alpha_1-1,
\lar{2}_1+k_2-k_1-1,\dots,\lar{2}_{k_2-k_1-1}+1,\lar{2}_{k_2-k_1}\big).
\]
By induction on $n$ this gives (in the 
$l\big(\lar{2}\big)\leq k_2-k_1$ case)
\begin{align*}
&\mathscr{L}^{k_1,\dots,k_n}_{\lar{2},\dots,\lar{n+1}}
(z;\alpha_1,\dots,\alpha_n,\beta) \\
&=S^{(k_2)}\big([k_2];z'\big) 
k_1!(-1)^{\binom{k_1}{2}} 
\prod_{r=2}^n \bigg( k_r!(-1)^{\binom{k_r}{2}} 
 s_{\lar{r+1}}[k_{r+1}-k_r]
\prod_{i=1}^{k_r} \frac{i!}{\Gamma(k_{r+1}-i+1)}\bigg) \\
& \quad \times \prod_{s=3}^n
\prod_{i=1}^{k_2} \prod_{j=1}^{k_s-k_{s-1}} 
\big(z'_i-\lar{s}_j+\alpha_2+\dots+\alpha_{s-1}
+k_{s-1}-k_s-s+j+2\big) \\
& \quad \times \prod_{3\leq r<s\leq n} 
\prod_{i=1}^{k_r-k_{r-1}} \prod_{j=1}^{k_s-k_{s-1}} 
\big(\lar{r}_i-\lar{s}_j+A_{r,s}+j-i\big) \\
& \quad \times
\prod_{i=1}^{k_2} \bigg(
\frac{\Gamma(z'_i+\alpha_2+\dots+\alpha_n-n+2)}
{\Gamma(z'_i+\alpha_2+\dots+\alpha_n+k_n-k_{n+1}-n+2)} \\
& \quad\qquad\qquad \times
\prod_{j\geq 1} 
\frac{z'_i-\lar{n+1}_j+\alpha_2+\dots+\alpha_n+k_n-k_{n+1}+j-n+1}
{z'_i+\alpha_2+\dots+\alpha_n+k_n-k_{n+1}+j-n+1}\bigg) \\
& \quad \times
\prod_{r=3}^n \prod_{i=1}^{k_r-k_{r-1}} \bigg(
\frac{\Gamma(\lar{r}_i+A_r-i-n)}{\Gamma(\lar{r}_i+A_{r,n+1}-i+1)} 
\prod_{j\geq 1} 
\frac{\lar{r}_i-\lar{n+1}_j+A_{r,n+1}+j-i}{\lar{r}_i+A_{r,n+1}+j-i}\bigg).
\end{align*}
By the definition of $z'$ and by \eqref{Eq_double-PS} with 
$(\ell,k,z,\la)\mapsto (k_2,k_1,z+\alpha_1-1,\lar{2})$, this finally yields
\begin{align*}
&\mathscr{L}^{k_1,\dots,k_n}_{\lar{2},\dots,\lar{n+1}}
(z;\alpha_1,\dots,\alpha_n,\beta) \\
&=S^{(k_1)}\big([k_1];z\big) 
\prod_{r=1}^n \bigg( (-1)^{\binom{k_r}{2}} s_{\lar{r+1}}[k_{r+1}-k_r] 
\prod_{i=1}^{k_r} \frac{i!}{\Gamma(k_{r+1}-i+1)} \bigg) \\
& \quad \times \prod_{1\leq r<s\leq n} 
\prod_{i=1}^{k_r-k_{r-1}} \prod_{j=1}^{k_s-k_{s-1}} 
\big(\lar{r}_i-\lar{s}_j+A_{r,s}+j-i\big) \\
& \quad \times
\prod_{r=1}^n \prod_{i=1}^{k_r-k_{r-1}} 
\bigg(\frac{\Gamma(\lar{r}_i+A_r+r-i-n)}{\Gamma(\lar{r}_i+A_{r,n+1}+r-i+1)} 
\prod_{j\geq 1} 
\frac{\lar{r}_i-\lar{n+1}_j+A_{r,n+1}+j-i}{\lar{r}_i+A_{r,n+1}+j-i}\bigg).
\end{align*}
As noted previously, the right-hand side vanishes if $k_2-k_1<0$, and
assuming $k_2-k_1\geq 0$ it then vanishes unless 
$l(\lar{2})\leq k_2-k_1$.
Hence the above holds regardless of the ordering of $k_1$ and $k_2$
or the length of $\lar{2}$.
\end{proof}

\section{An elliptic AFLT integral}\label{Sec_E-Selberg}

In this section we prove the elliptic AFLT integral of Theorem~\ref{Thm_eAFLT}.
Throughout it is assumed that $p,q\in\mathbb{C}$ such that $\abs{p},\abs{q}<1$. 

\subsection{The elliptic Selberg integral}

We begin with a brief review of the elliptic Selberg integral and its
relation to the ordinary Selberg integral \eqref{Eq_Selberg}. 

The elliptic Selberg integral\footnote{An elliptic Selberg integral of a 
very different type, which arises as a solution of the
Knizhnik--Zamolodchikov--Bernard heat equation, may be found in
\cite{TV19} and references therein.} 
was first conjectured by van Diejen and Spiridonov \cite{vDS00,vDS01} and
then proved by the second author in \cite{Rains10}. 
(See also \cite{IN17,Spiridonov07} for alternative proofs).
It corresponds to the $\bla=\bmu=\boldsymbol{0}$ case of 
Theorem~\ref{Thm_eAFLT}, viz.\
\begin{align}\label{Eq_E-Selberg}
&S_n(t_1,\dots,t_6;t;p,q) \\
&\qquad:=\kappa_n
\Int_{\mathbb{T}^n} \prod_{1\leq i<j\leq n}
\frac{\Gamma(tz_i^{\pm}z_j^{\pm};p,q)}
{\Gamma(z_i^{\pm}z_j^{\pm};p,q)}
\prod_{i=1}^n \frac{\prod_{r=1}^6 \Gamma(t_r z_i^{\pm};p,q)}
{\Gamma(z_i^{\pm 2};p,q)} \,
\frac{\dup z_1}{z_1}\cdots \frac{\dup z_n}{z_n} \notag \\
&\qquad\hphantom{:}=\prod_{i=1}^n 
\bigg(\Gamma(t^i;p,q)
\prod_{1\leq r<s\leq 6} \Gamma(t^{i-1}t_rt_s;p,q)\bigg), \notag 
\end{align}
where $t,t_1,\dots,t_6\in\mathbb{C}$ are such that
$\abs{t},\abs{t_1},\dots,\abs{t_6}<1$ and 
the balancing condition \eqref{Eq_balancing} holds.
For $n=1$ the elliptic Selberg integral corresponds to Spiridonov's 
elliptic analogue of the Euler beta integral \cite{Spiridonov01} and
for general $n$ it was shown in \cite{Rains09} that by taking an 
appropriate limit the full Selberg integral \eqref{Eq_Selberg} 
(with $k\mapsto n$) arises.
Specifically, setting $q=\exp(2\pi\iup\tau)$ for $\tau\in\mathbb{H}$
\begin{align*}
\lim_{p\to 0} \lim_{\tau\to 0} & \,
\prod_{i=1}^n \frac{\Gamma(q^{\alpha+\beta+(2n-i-1)\gamma};p,q)}
{\Gamma(q^{i\gamma};p,q)\Gamma(q^{\alpha+(i-1)\gamma};p,q)
\Gamma(q^{\beta+(i-1)\gamma};p,q)
\Gamma^4(-q^{(\alpha+\beta)/2+(i-1)\gamma};p,q)} \\
& \times S_n\big(q^{\alpha/2},q^{\alpha/2},-q^{\beta/2},-q^{\beta/2},
p^{1/2}q^{1/2-\zeta/2},p^{1/2}q^{1/2-\zeta/2};q^{\gamma};p,q\big) \\
&\qquad =\prod_{i=1}^n
\frac{\Gamma(\alpha+\beta+(2n-i-1)\gamma)\Gamma(1+\gamma)}
{\Gamma(\alpha+(i-1)\gamma)\Gamma(\beta+(i-1)\gamma)\Gamma(1+i\gamma)} \\
&\qquad \quad \times
\Int_{[0,1]^n} \prod_{i=1}^n t_i^{\alpha-1} (1-t_i)^{\beta-1} 
\prod_{1\leq i<j\leq n} \abs{t_i-t_j}^{2\gamma}\,
\dup t_1\cdots \dup t_n,
\end{align*}
where $\textrm{Re}(\alpha),\textrm{Re}(\beta),\textrm{Re}(\gamma)>0$,
$2(n-1)\gamma+\alpha+\beta-\zeta=0$, see \cite[Theorem 7.4]{Rains09}.
This same limit, taken on the right-hand side of \eqref{Eq_E-Selberg}
yields $1$, so that the Selberg integral follows, albeit for a slightly
more restricted range of the parameter $\gamma$ than strictly necessary.

\subsection{Elliptic interpolation functions}

To lift the elliptic Selberg integral \eqref{Eq_E-Selberg} to the
elliptic AFLT integral of Theorem~\ref{Thm_eAFLT} we require a number
of results from the theory of elliptic interpolation functions on root
systems.
For a detailed discussion of these functions we refer the reader to
\cite{CG06,Rains06,Rains10,Rains12,RW17}, and below we only state
some of the key results needed in the proof of Theorem~\ref{Thm_eAFLT}.

In order to describe the elliptic skew interpolation functions we 
need the $\mathrm{BC}_n$-symmetric elliptic interpolation functions 
(or well-poised Macdonald functions)
\[
R^{\ast}_{\mu}(x_1,\dots,x_n;a,b;q,t;p),
\]
where $\mu\in\mathscr{P}_n$.
For $\la\in\mathscr{P}_n$ such that $\mu\not\subseteq\la$ 
these functions satisfy the vanishing property
\[
R^{\ast}_{\mu}\big(a\spec{\la}_{n;q,t};a,b;q,t;p\big)=0.
\]
The elliptic interpolation function
$R^{\ast}_{\mu}(a,b;q,t;p)$ generalises Okounkov's 
$\mathrm{BC}_n$-symmetric Macdonald
interpolation polynomial $P^{\ast}_{\mu}(q,t,s)$
\cite{Okounkov98} which has the same
vanishing property and has the ordinary Macdonald polynomial $P_{\mu}(q,t)$
as top-homogeneous degree component.
When $t^nab=pq$, the elliptic interpolation functions are referred
to as being of Cauchy-type and factorise as
\begin{equation}\label{Eq_factorisation}
R^{\ast}_{\mu}(x_1,\dots,x_n;a,b;q,t;p)
=\Delta^0_{\mu}(t^{n-1}a/b\vert 
t^{n-1}ax_1^{\pm},\dots,t^{n-1}ax_n^{\pm};q,t;p),
\end{equation}
where $\Delta^0_{\mu}(a\vert\dots,x_i^{\pm},\dots;q,t;p):=
\Delta^0_{\mu}(a\vert\dots,x_i,x_i^{-1},\dots;q,t;p)$.
They also factor under principal specialisation:
\begin{equation}\label{Eq_R-PS}
R^{\ast}_{\mu}(v\spec{0}_{n;q,t};a,b;q,t;p)
=\Delta^0_{\mu}(t^{n-1}a/b\vert t^{n-1}av,a/v;q,t;p),
\end{equation}
and satisfy the evaluation symmetry
\begin{equation}\label{Eq_evaluation}
\frac{R^{\ast}_{\mu}(v\spec{\la}_{n;q,t}/a;a,b/a;q,t;p)}
{R^{\ast}_{\mu}(v\spec{0}_{n;q,t}/a;a,b/a;q,t;p)} =
\frac{R^{\ast}_{\la}(v\spec{\mu}_{n;q,t}/a';a',b/a';q,t;p)}
{R^{\ast}_{\la}(v\spec{0}_{n;q,t}/a';a',b/a';q,t;p)},
\end{equation}
where $\la,\mu\in\mathscr{P}_n$ and $aa'=v\sqrt{t^{n-1}b}$.

Using the elliptic interpolation functions one can define
elliptic binomial coefficients as
\begin{align*}
\binom{\la}{\mu}_{[a,b];q,t;p}&:=
\frac{(pqa;q,t;p)_{2\la^2}}
{C^{-}_{\la}(pq;q,t;p)C^{-}_{\la}(t;q,t;p)
C^{+}_{\la}(a;q,t;p)C^{+}_{\la}(pqa/t;q,t;p)} \\[1mm]
& \quad \times
\Delta_{\mu}^0(a/b\vert t^n,1/b;q,t;p) \,
R^{\ast}_{\mu}\big(t^{1-n}a^{1/2}
\spec{\la}_{n;q,t};t^{1-n}a^{1/2},ba^{-1/2};q,t;p\big),
\end{align*}
where $2\la^2$ is shorthand for the partition 
$(2\la_1,2\la_1,2\la_2,2\la_2,\dots)$ and $n$ is an integer such
that $n\geq l(\la),l(\mu)$.
This definition is independent of the choice of $a^{1/2}$ and $n$.
The elliptic binomial coefficients vanish unless $\mu\subseteq\la$
and trivialise to $1$ when $\mu=0$ (but not when $\mu=\la$). 
It will be convenient to also use the normalised elliptic binomial
coefficients 
\begin{equation}\label{Eq_ebinom}
\obinomE{\la}{\mu}_{[a,b](v_1,\dots,v_k);q,t;p}:=
\frac{\Delta^0_{\la}(a\vert b,v_1,\dots,v_k;q,t;p)}
{\Delta^0_{\mu}(a/b\vert 1/b,v_1,\dots,v_k;q,t;p)}
\binom{\la}{\mu}_{[a,b];q,t;p}.
\end{equation}
In terms of these, the elliptic $\mathrm{C}_n$ analogue of Jackson's
summation \cite[Theorem 4.1]{Rains06} 
(or \cite[Equation~(3.7)]{CG06}) is given by
\begin{equation}\label{Eq_Jackson}
\sum_{\mu} \Delta^0_{\mu}(a/b\vert d,e;q,t;p)\,
\obinomE{\la}{\mu}_{[a,b];q,t;p} \obinomE{\mu}{\nu}_{[a/b,c/b];q,t;p}
=\obinomE{\la}{\nu}_{[a,c](bd,be);q,t;p},
\end{equation}
where $bcde=apq$.

We can now define the elliptic skew interpolation functions as
\begin{align*}
R^{\ast}_{\la/\nu}([v_1,\dots,v_{2n}];a,b;q,t;p) 
&:=\sum_{\mu} \Delta_{\mu}^0(pq/b^2\vert pq/bv_1,\dots,pq/bv_{2n};q,t;p) \\ 
&\quad\qquad \times \obinomE{\la}{\mu}_{[a/b,ab/pq];q,t;p}
\obinomE{\mu}{\nu}_{[pq/b^2,pqV/ab];q,t;p},
\end{align*}
where $V:=v_1\cdots v_{2n}$.
Note that unlike the $\mathrm{BC}_n$-symmetric interpolation functions
the skew interpolation functions are $\Symm_{2n}$-symmetric. 
The rationale for using plethystic brackets around the $v_i$ is
that in the $p\to 0$ limit they relate to the variables $x_i$ of the 
$\mathrm{BC}_n$-symmetric interpolation functions via the plethystic 
substitution
\[
x_i+x_i^{-1}\mapsto 
\frac{t^{1/2}(v_{2i-1}^{-1}+v_{2i}^{-1}-v_{2i-1}-v_{2i})}{1-t}
\]
for all $1\leq i\leq n$.
In particular, for real $\alpha,\beta$ such that $0<\alpha<\beta<1$
and $\alpha+\beta<1$,
\begin{multline}\label{Eq_skewR-limit}
\lim_{p\to 0} p^{\alpha \abs{\la}}
R^{\ast}_{\la/0}
\Big(\big[t^{1/2}(p^{-\alpha}x_1)^{\pm},\dots,t^{1/2}(p^{-\alpha}x_n)^{\pm},
p^{-\alpha}t^{-1/2}c,p^{\alpha}t^{1/2}/d\big];
a,p^{\beta}b;q,t;p\Big)\\
=\big({-}at^{-1/2}\big)^{\abs{\la}}q^{n(\la')} t^{-2n(\la)} c_{\la}(q,t)
P_{\la}\Big(\Big[X+\frac{d-c}{1-t}\Big];q,t\Big),
\end{multline}
where $X:=x_1+\dots+x_n$.

By \eqref{Eq_Jackson} it immediately follows that
\begin{equation}\label{Eq_nisone}
R^{\ast}_{\la/\nu}([v_1,v_2];a,b;q,t;p) 
=\obinomE{\la}{\nu}_{[a/b,v_1v_2](a/v_1,a/v_2);q,t;p}.
\end{equation}
The skew interpolation functions further satisfy the branching rule
\begin{align}\label{Eq_branching}
&R^{\ast}_{\la/\nu}
\big([v_1,\dots,v_{2n},w_1,\dots,w_{2m}];a,b;q,t;p\big) \\
&\qquad \quad = \sum_{\mu} R^{\ast}_{\la/\mu}
\big([v_1,\dots,v_{2n}];a,b;q,t;p\big) 
R^{\ast}_{\mu/\nu}
\big([w_1,\dots,w_{2m}];a/V,b;q,t;p\big), \notag
\end{align}
and for $\nu=0$ generalise the $\mathrm{BC}_n$-symmetric interpolation
functions 
\begin{align}\label{Eq_interpolation-skew}
&R^{\ast}_{\la/0}([t^{1/2}x_1^{\pm},\dots,t^{1/2}x_n^{\pm}];
t^{n-1/2}a,t^{1/2}b;q,t;p) \\[1mm]
&\qquad\quad
=\begin{cases}
\Delta_{\la}^0(t^{n-1}a/b\vert t^n;q,t;p)\,
R^{\ast}_{\la}(x_1,\dots,x_n;a,b;q,t;p)
& \text{for $\la\in\mathscr{P}_n$}, \\[2mm]
0 & \text{otherwise},
\end{cases} 
\notag
\end{align}
where we recall the convention \eqref{Eq_Rastpm}.

The elliptic Selberg integral \eqref{Eq_E-Selberg} is $p,q$-symmetric.
In order to maintain this symmetry in the AFLT case, we need to generalise
all of the above elliptic functions.
Let $\bla=(\lar{1},\lar{2})$ and $\bmu=(\mur{1},\mur{2})$
be bipartitions.
Then
\begin{align*}
R^{\ast}_{\bmu}(a,b;t;p,q)&:=
R^{\ast}_{\mur{1}}(a,b;q,t;p)
R^{\ast}_{\mur{2}}(a,b;p,t;q) \\[2mm]
R^{\ast}_{\bla/\bmu}(a,b;t;p,q)&:=
R^{\ast}_{\lar{1}/\mur{1}}(a,b;q,t;p)
R^{\ast}_{\lar{2}/\mur{2}}(a,b;p,t;q).
\end{align*}
In much the same way we define $\Delta^0_{\bla}(a\vert b_1,\dots,b_k;t;p,q)$
and $\obinomE{\bla}{\bmu}_{[a,b](v_1,\dots,v_k);t;p,q}$.
Finally, we recall the definition of 
$\spec{\bla}_{n;t;p,q}$, the spectral vector indexed by a 
bipartition, from \eqref{Eq_e-spec}.

For the proof of the elliptic AFLT integral we require two key identities
for the interpolation functions.
The first of these is the following generalised elliptic Selberg integral,
see \cite[Theorem~9.2]{Rains10}.
For $\bla,\bmu\in\mathscr{P}_n^2$ and $t,t_1,t_2,t_3,t_4,t_5,t_6\in\mathbb{C}$ 
such that \eqref{Eq_balancing} holds 
\begin{align}\label{Eq_thm92}
&\kappa_n \Int_{C_{\bla\bmu}} 
R^{\ast}_{\bla}(z_1,\dots,z_n;t_1,t_2;t;p,q) 
R^{\ast}_{\bmu}(z_1,\dots,z_n;t_3,t_6;t;p,q) \\[-2mm]
& \quad\qquad \times
\prod_{1\leq i<j\leq n}
\frac{\Gamma(tz_i^{\pm}z_j^{\pm};p,q)}
{\Gamma(z_i^{\pm}z_j^{\pm};p,q)}
\prod_{i=1}^n \frac{\prod_{r=1}^6 \Gamma(t_r z_i^{\pm};p,q)}
{\Gamma(z_i^{\pm 2};p,q)}\,
\frac{\dup z_1}{z_1}\cdots\frac{\dup z_n}{z_n} \notag \\
&\quad=
\prod_{i=1}^n 
\bigg(\Gamma(t^i;p,q)
\prod_{1\leq r<s\leq 6} \Gamma(t^{i-1}t_rt_s;p,q)\bigg) \notag \\[1mm]
& \qquad \times
\Delta^0_{\bla}(t^{n-1}t_1/t_2\vert t^{n-1}t_1t_4,t^{n-1}t_1t_5;t;p,q) 
\Delta^0_{\bmu}(t^{n-1}t_3/t_6\vert t^{n-1}t_3t_4,t^{n-1}t_3t_5;t;p,q)
\notag \\[2mm]
& \qquad \times
R_{\bla}^{\ast}(t_3\spec{\boldsymbol{0}}_{n;t;p,q}/\zeta';
t_1 \zeta',t_2\zeta';t;p,q)
R_{\bmu}^{\ast}(t_1\spec{\bla}_{n;t;p,q}/\zeta;
t_3 \zeta,t_6 \zeta;t;p,q), \notag
\end{align}
where $\zeta:=\sqrt{t^{n-1}t_1t_2}$ and $\zeta':=\sqrt{t^{n-1}t_3t_6}$,
and where $C_{\bla\bmu}$ is a deformation of $\mathbb{T}^n$ separating
sequences of poles of the integrand tending to zero from sequences of
poles tending to infinity. 
Note that by the evaluation symmetry \eqref{Eq_evaluation}
this result is invariant under the simultaneous substitution
$(\bla,\bmu,t_1,t_2,t_3,t_6)\mapsto (\bmu,\bla,t_3,t_6,t_1,t_2)$.
We also note that for $\bmu=0$ the above integral may be viewed as
an elliptic analogue of Kadell's integral.

The second key result we need is the connection coefficient identity 
\cite[Corollary 4.14]{Rains06}
\begin{equation}\label{Eq_CC}
R^{\ast}_{\bla}(x_1,\dots,x_n;a,b;t;p,q)
=\sum_{\bmu}
\obinomE{\bla}{\bmu}_{[t^{n-1}a/b,a/a'](t^{n-1}aa');t;p,q}
R^{\ast}_{\bmu}(x_1,\dots,x_n;a',b;t;p,q),
\end{equation}
where $\bla\in\mathscr{P}_n^2$.

\subsection{Proof of Theorem~\ref{Thm_eAFLT}}

Denote the elliptic AFLT integral by $I_{\bla\bmu}$.
As a first step towards proving the theorem we apply the branching
rule \eqref{Eq_branching} to expand 
$R^{\ast}_{\bmu/\boldsymbol{0}}$ as
\begin{align*}
&R^{\ast}_{\bmu/\boldsymbol{0}}
\big([t^{1/2}z_1^{\pm},\dots,t^{1/2}z_n^{\pm},
t^{-1/2}t_4,t^{-1/2}t_5];t^{n-3/2}t_3t_4t_5,t^{1/2}t_6;t;p,q\big) \\[2mm]
&\qquad=\sum_{\bnu\subseteq\bmu}
R^{\ast}_{\bmu/\bnu}
\big([t^{-1/2}t_4,t^{-1/2}t_5];
t^{n-3/2}t_3t_4t_5,t^{1/2}t_6;t;p,q\big) \\
&\qquad\qquad\quad\times
R^{\ast}_{\bnu/\boldsymbol{0}}
\big([t^{1/2}z_1^{\pm},\dots,t^{1/2}z_n^{\pm}];
t^{n-1/2}t_3,t^{1/2}t_6;t;p,q\big).
\end{align*}
By \eqref{Eq_nisone} and \eqref{Eq_interpolation-skew} this may be 
further rewritten as
\begin{align*}
&R^{\ast}_{\bmu/\boldsymbol{0}}
\big([t^{1/2}z_1^{\pm},\dots,t^{1/2}z_n^{\pm},
t^{-1/2}t_4,t^{-1/2}t_5];t^{n-3/2}t_3t_4t_5,t^{1/2}t_6;t;p,q\big) \\[2mm]
&\qquad =\sum_{\bnu\in\mathscr{P}_n^2}
\Delta^0_{\bnu}(t^{n-1}t_3/t_6\vert t^n;t;p,q)
\obinomE{\bmu}{\bnu}_{[t^{n-2}t_3t_4t_5/t_6,t_4t_5/t]
(t^{n-1}t_3t_4,t^{n-1}t_3t_5);t;p,q} \\[2mm]
&\qquad\qquad\quad\times
R^{\ast}_{\bnu}(z_1,\dots,z_n;t_3,t_6;t;p,q).
\end{align*}
If we substitute this into the elliptic AFLT integral, 
rewrite $R^{\ast}_{\bla/\boldsymbol{0}}$
in terms of $R^{\ast}_{\bla}$ using 
\eqref{Eq_interpolation-skew} and then interchange
the order of the sum and integral, we obtain exactly the integral
\eqref{Eq_thm92} (with $\bmu\mapsto\bnu$) in the summand.
Hence
\begin{align*}
I_{\bla\bmu}&=S_n(t_1,t_2,t_3,t_4,t_5,t_6;t;p,q) \\[2mm]
&\quad\times
\Delta^0_{\bla}(t^{n-1}t_1/t_2\vert t^n,t^{n-1}t_1t_4,t^{n-1}t_1t_5;t;p,q) 
R_{\bla}^{\ast}(t_3\spec{\boldsymbol{0}}_{n;t;p,q}/\zeta';
t_1 \zeta',t_2\zeta';t;p,q) \notag \\[2mm]
&\quad\times \sum_{\bnu\in\mathscr{P}_n^2}
\Delta^0_{\bnu}(t^{n-1}t_3/t_6\vert t^n,t^{n-1}t_3t_4,t^{n-1}t_3t_5;t;p,q)\\
& \qquad\qquad \times
\obinomE{\bmu}{\bnu}_{[t^{n-2}t_3t_4t_5/t_6,t_4t_5/t]
(t^{n-1}t_3t_4,t^{n-1}t_3t_5);t;p,q} 
R_{\bnu}^{\ast}(t_1\spec{\bla}_{n;t;p,q}/\zeta;t_3 \zeta,
t_6 \zeta;t;p,q),
\end{align*}
where, as before, $\zeta:=\sqrt{t^{n-1}t_1t_2}$ and
$\zeta':=\sqrt{t^{n-1}t_3t_6}$.
By \eqref{Eq_balancing}, \eqref{Eq_Delta_sym} and \eqref{Eq_ebinom}
\begin{multline*}
\Delta^0_{\bnu}(t^{n-1}t_3/t_6\vert t^n,t^{n-1}t_3t_4,t^{n-1}t_3t_5;t;p,q)
\obinomE{\bmu}{\bnu}_{[t^{n-2}t_3t_4t_5/t_6,t_4t_5/t]
(t^{n-1}t_3t_4,t^{n-1}t_3t_5);t;p,q} \\[1mm]
=\Delta^0_{\bmu}(t^{n-2}t_3t_4t_5/t_6\vert t^{n-1}t_3t_4,t^{n-1}t_3t_5,
t^{n-1}t_4t_5;t;p,q) 
\obinomE{\bmu}{\bnu}_{[t^{n-2}t_3t_4t_5/t_6,t_4t_5/t](pqt_3/tt_6);t;p,q}.
\end{multline*}
As a result,
\begin{align*}
I_{\bla\bmu}&=S_n(t_1,t_2,t_3,t_4,t_5,t_6;t;p,q) \\[2mm]
&\quad\times
\Delta^0_{\bla}(t^{n-1}t_1/t_2\vert t^n,t^{n-1}t_1t_4,t^{n-1}t_1t_5;t;p,q) 
R_{\bla}^{\ast}(t_3\spec{\boldsymbol{0}}_{n;t;p,q}/\zeta';
t_1 \zeta',t_2\zeta';t;p,q)
\notag \\[2mm]
&\quad\times\Delta^0_{\bmu}(t^{n-2}t_3t_4t_5/t_6\vert 
t^{n-1}t_3t_4,t^{n-1}t_3t_5,t^{n-1}t_4t_5;t;p,q) \\
&\quad\times \sum_{\bnu\in\mathscr{P}_n^2}
\obinomE{\bmu}{\bnu}_{[t^{n-2}t_3t_4t_5/t_6,t_4t_5/t]
(pqt_3/tt_6);t;p,q} 
R_{\bnu}^{\ast}( t_1\spec{\bla}_{n;t;p,q}/\zeta;t_3 \zeta,
t_6 \zeta;t;p,q).
\end{align*}
The sum over $\bnu$ can be carried out by \eqref{Eq_CC} with
$(a,a',b)\mapsto (t_3t_4t_5\zeta/t,t_3\zeta,t_6\zeta)$, so that
\begin{align*}
I_{\bla\bmu}&=S_n(t_1,t_2,t_3,t_4,t_5,t_6;t;p,q) \\[2mm]
&\quad\times
\Delta^0_{\bla}(t^{n-1}t_1/t_2\vert t^n,t^{n-1}t_1t_4,t^{n-1}t_1t_5;t;p,q) 
R_{\bla}^{\ast}(t_3\spec{\boldsymbol{0}}_{n;t;p,q}/\zeta';
t_1 \zeta',t_2\zeta';t;p,q)
\notag \\[2mm]
&\quad\times\Delta^0_{\bmu}(t^{n-2}t_3t_4t_5/t_6\vert 
t^{n-1}t_3t_4,t^{n-1}t_3t_5,t^{n-1}t_4t_5;t;p,q) \\[2mm]
&\quad\times 
R_{\bmu}^{\ast}(t_1\spec{\bla}_{n;t;p,q}/\zeta;
t_3t_4t_5\zeta/t,t_6 \zeta;t;p,q).
\end{align*}
We finally observe that 
\[
t^n(t_3t_4t_5\zeta/t)(t_6 \zeta)=t^{2n-2}t_1t_2t_3t_4t_5t_6=pq
\]
so that the interpolation function in the last line is of Cauchy type
and hence factors by \eqref{Eq_factorisation}. 
Thus
\begin{align*}
I_{\bla\bmu}&=S_n(t_1,t_2,t_3,t_4,t_5,t_6;t;p,q) \\[2mm]
&\quad\times
\Delta^0_{\bla}(t^{n-1}t_1/t_2\vert t^n,t^{n-1}t_1t_4,t^{n-1}t_1t_5;t;p,q) 
R_{\bla}^{\ast}(t_3\spec{\boldsymbol{0}}_{n;t;p,q}/\zeta';
t_1 \zeta',t_2\zeta';t;p,q)
\notag \\[2mm]
&\quad\times\Delta^0_{\bmu}(t^{n-2}t_3t_4t_5/t_6\vert 
t^{n-1}t_3t_4,t^{n-1}t_3t_5,t^{n-1}t_4t_5;t;p,q) \\[2mm]
&\quad\times 
\frac{\Delta^0_{\bmu}(t^{n-2}t_3t_4t_5/t_6\vert
t^{n-2}t_1t_3t_4t_5 \spec{\boldsymbol{\bla}}_{n;t;p,q})}
{\Delta^0_{\bmu}(t^{n-2}t_3t_4t_5/t_6\vert
t^{n-1}t_1t_3t_4t_5 \spec{\boldsymbol{\bla}}_{n;t;p,q})},
\end{align*}
where we have again used \eqref{Eq_balancing} and
\eqref{Eq_Delta_sym}.
Using the principal specialisation formula \eqref{Eq_R-PS} completes the 
proof.

\subsection{Proof of Corollary~\ref{Cor_Mac-limit}}

Throughout the proof we used condensed notation for theta and elliptic
gamma functions:
\begin{align*}
\theta(z_1,\dots,z_k;p)&=\theta(z_1;p)\cdots\theta(z_k;p) \\
\Gamma(z_1,\dots,z_k;p,q)&=\Gamma(z_1;p,q)\cdots\Gamma(z_k;p,q).
\end{align*}

Denote the integral identity obtained from \eqref{Eq_eAFLT} by 
restricting $\bla$ and $\bmu$ to $\bla=(\la,0)$ and $\bmu=(\mu,0)$ as
\[
\mathscr{L}_{\la,\mu}(t_1,\dots,t_6;q,t;p)=
\mathscr{R}_{\la,\mu}(t_1,\dots,t_6;q,t;p),
\]
where $t^{2n-2}t_1t_2t_3t_4t_5t_6=pq$.
In order to take the $p\to 0$ limit we adopt the symmetry breaking procedure
of \cite{Rains09} and multiply the integrand of $\mathscr{L}_{\la,\mu}$ by
\cite[Corollary~1.2]{Rains06}
\[
1=\sum_{\sigma\in\{\pm 1\}^n} \prod_{i=1}^n
\frac{\theta(t_1 z_i^{\sigma_i},t_3 z_i^{\sigma_i},t_4 z_i^{\sigma_i},
t^{n-1}t_1t_3t_4z_i^{-\sigma_i};q)}
{\theta(z_i^{2\sigma_i},t^{i-1}t_1t_3,t^{i-1}t_1t_4,t^{i-1}t_3t_4;q)}
\prod_{1\leq i<j\leq n} \frac{\theta(t z_i^{\sigma_i}z_j^{\sigma_j};q)}
{\theta(z_i^{\sigma_i}z_j^{\sigma_j};q)}.
\]
By $\mathrm{BC}_n$ symmetry, all $2^n$ integrals that arise are 
identical.
Using $\Gamma(pz;p,q)=\theta(z;q)\Gamma(z;p,q)$ this yields
\begin{align*}
&\mathscr{L}_{\la,\mu}(t_1,\dots,t_6;q,t;p) \\
&\quad=2^n \kappa_n\int
R^{\ast}_{\la/0}\big([t^{1/2}z^{\pm}];
t^{n-1/2}t_1,t^{1/2}t_2;q,t;p\big) \\
& \quad\qquad\qquad\times
R^{\ast}_{\mu/0}
\big([t^{1/2}z^{\pm},
t^{-1/2}t_4,t^{-1/2}t_5];t^{n-3/2}t_3t_4t_5,t^{1/2}t_6;q,t;p\big) \\[1mm]
& \quad\qquad\qquad\times
\prod_{1\leq i<j\leq n} 
\frac{\Gamma(p t z_i z_j,t/z_iz_j,t(z_i/z_j)^{\pm 1};p,q)}
{\Gamma(p z_iz_j,1/z_iz_j,(z_i/z_j)^{\pm 1};p,q)} \\
& \quad\qquad\qquad\times
\prod_{i=1}^n \frac{\theta(t^{n-1}t_1t_3t_4/z_i;q)
\prod_{r\in I} \Gamma(p t_r z_i,t_r/z_i;p,q)
\prod_{r\in J} \Gamma(t_r z_i^{\pm};p,q)}
{\theta(t^{i-1}t_1t_3,t^{i-1}t_1t_4,t^{i-1}t_3t_4;q)
\Gamma(p z_i^2,z_i^{-2};p,q)} \, \frac{\dup z}{z},
\end{align*}
where $I:=\{1,3,4\}$, $J:=\{2,5,6\}$, $z=(z_1,\dots,z_n)$ and
$\dup z/z:=(\dup z_1/z_1)\cdots(\dup z_n/z_n)$.
We now scale the parameters as
\[
(t_1,t_2,t_3,t_4,t_5,t_6) \mapsto
(t_1,p^{1/2}t_2,t_3,p^{-1/4}t_4,p^{1/4}t_5,p^{1/2}t_6),
\]
resulting in the $p$-independent balancing condition 
$t^{2n-2}t_1t_2t_3t_4t_5t_6=q$.
We can now also scale the integration contour by a factor of $p^{-1/4}$
without passing over any poles.
After also replacing $z_i\mapsto p^{-1/4}z_i$, 
all elliptic gamma functions in the resulting integral 
have arguments scaling as $p^{\gamma}$ for $\gamma\in [0,1]$.
Then taking the $p\to 0$ limit using \eqref{Eq_skewR-limit} 
(with $\alpha=1/4$ and $\beta=1/2$) and finally
replacing $z\mapsto z/t_5$, we obtain
\begin{align*}
\lim_{p\to 0} & \bigg(p^{(\abs{\la}+\abs{\mu})/4}
\mathscr{L}_{\la,\mu}(t_1,p^{1/2}t_2,t_3,p^{-1/4}t_4,p^{1/4}t_5,p^{1/2}t_6) 
\prod_{i=1}^n \theta(p^{-1/4}t^{i-1}t_1t_4,p^{-1/4}t^{i-1}t_3t_4;q)\bigg) \\
&=(-t^{n-1}t_1/t_5)^{\abs{\la}} (-t^{n-2}t_3t_4)^{\abs{\mu}}
q^{n(\la')+n(\mu')} t^{-2n(\la)-2n(\mu)} \\
& \quad \times c_{\la}(q,t) c_{\mu}(q,t) \,
\frac{(q;q)_{\infty}^n}{(t;q)_{\infty}^n}
\prod_{i=1}^n \frac{1}{\theta(t^{i-1}t_1t_3;q)} \\
&\quad \times \frac{1}{n!(2\pi\iup)^n}
\int P_{\la}(z;q,t) P_{\mu}\Big(\Big[z+\frac{t-t_4t_5}{1-t}\Big];q,t\Big) \\
&\qquad\qquad\qquad\quad \times 
\prod_{i=1}^n \frac{\theta(t^{n-1}t_1t_3t_4t_5/z_i;q)}
{(t_4t_5/z_i,z_i;q)_{\infty}} 
\prod_{1\leq i<j\leq n} \frac{(z_i/z_j,z_j/z_i;q)_{\infty}}
{(tz_i/z_j,tz_j/z_i;q)_{\infty}} \, \frac{\dup z}{z},
\end{align*}
where the contour can be chosen as $\mathbb{T}^n$ provided
$\abs{t_4t_5}<1$.
Since the limit of the right-hand side of \eqref{Eq_eAFLT} is given by
\begin{align*}
\lim_{p\to 0} & \bigg(p^{(\abs{\la}+\abs{\mu})/4}
\mathscr{R}_{\la,\mu}(t_1,p^{1/2}t_2,t_3,p^{-1/4}t_4,p^{1/4}t_5,p^{1/2}t_6) 
\prod_{i=1}^n \theta(p^{-1/4}t^{i-1}t_1t_4,p^{-1/4}t^{i-1}t_3t_4;q)\bigg) \\
&=(-t^{n-1}t_1t_4)^{\abs{\la}} (-t^{n-1}t_3t_4)^{\abs{\mu}}
q^{n(\la')+n(\mu')} t^{-n(\la)-n(\mu)} \\
&\quad \times 
\frac{(t^n,t^{n-1}t_1t_3;q,t)_{\la}(t^{n-1}t_4t_5;q,t)_{\mu}}
{(t^{2n-m-2}t_1t_3t_4t_5;q,t)_{\la}}
\prod_{i=1}^n \frac{(t^{2n-i-1}t_1t_3t_4t_5;q)_{\infty}}
{(t^i,t^{i-1}t_1t_3,t^{i-1}t_4t_5;q)_{\infty}} \\
&\quad \times 
\prod_{i=1}^n \prod_{j=1}^m
\frac{(t^{2n-i-j-1}t_1t_3t_4t_5;q)_{\la_i+\mu_j}}
{(t^{2n-i-j}t_1t_3t_4t_5;q)_{\la_i+\mu_j}},
\end{align*}
the claim follows with $(a,b)=(t^{n-1}t_1t_3t_4t_5,t_4t_5)$.

\subsection{An equivalent form of Corollary~\ref{Cor_Mac-limit}}

For functions $f,g:\mathbb{C}^n\to\mathbb{C}$, symmetric in the 
$n$ variables, define the scalar product \cite[p.~372]{Macdonald95}
\begin{equation}\label{Eq_scalar-product}
\langle f,g\rangle'_n:=
\frac{1}{n!(2\pi\iup)^n} \Int_{\mathbb{T}^n} f(z) g(z^{-1})
\prod_{1\leq i<j\leq n} \frac{(z_i/z_j,z_j/z_i;q)_{\infty}}
{(tz_i/z_j,tz_j/z_i;q)_{\infty}} \, 
\frac{\dup z_1}{z_1}\cdots \frac{\dup z_n}{z_n}.
\end{equation}
Then, for $\la\in\mathscr{P}_n$, \cite[pp.~369\&370]{Macdonald95}
\begin{equation}\label{Eq_ortho}
\big\langle P_{\la}(q,t),Q_{\mu}(q,t)\big\rangle'_n=
\delta_{\la,\mu} \, \frac{(t^n;q,t)_{\la}}{(qt^{n-1};q,t)_{\la}}
\prod_{i=1}^n\frac{(t,qt^{i-1};q)_{\infty}}{(q,t^i;q)_{\infty}}.
\end{equation}
In terms of $\langle\cdot,\cdot\rangle'_n$
Corollary~\ref{Cor_Mac-limit} may be rewritten as follows.

\begin{corollary}
For $\la,\mu\in\mathscr{P}$ and $a,b,q,t\in\mathbb{C}$ such 
that $\abs{b},\abs{q},\abs{t}<1$,
\begin{align*}
&\bigg\langle 
P_{\la}(z;q,t) \prod_{i=1}^n \frac{(az_i;q)_{\infty}}{(bz_i;q)_{\infty}},
Q_{\mu}\bigg(\bigg[z+\frac{t-b}{1-t}\bigg];q,t\bigg)
\prod_{i=1}^n \frac{(qz_i/a;q)_{\infty}}{(z_i;q)_{\infty}}
\bigg\rangle'_n \\
&\quad= t^{(1-n)\abs{\mu}} 
P_{\la}\bigg(\bigg[\frac{1-t^n}{1-t}\bigg];q,t\bigg)
Q_{\mu}\bigg(\bigg[\frac{1-bt^{n-1}}{1-t}\bigg];q,t\bigg) \\[1mm]
&\quad\quad \times 
\frac{(qt^m/a;q,t)_{\la}}{(bqt^{n-1}/a;q,t)_{\la}}
\prod_{i=1}^n 
\frac{(t,at^{i-1},qt^{i-1}b/a;q)_{\infty}}
{(q,t^i,bt^{i-1};q)_{\infty}}
\prod_{i=1}^n \prod_{j=1}^m
\frac{(qt^{j-i}/a;q)_{\la_i-\mu_j}}{(qt^{j-i+1}/a;q)_{\la_i-\mu_j}}.
\end{align*}
where $m$ is an arbitrary integer such that $m\geq l(\mu)$.
\end{corollary}

For $b=t$ this is equivalent to \cite[Theorem 1.7]{Warnaar05},
which, as shown on page~261 of that paper, is a generalisation
of \eqref{Eq_ortho}.

\begin{proof}
Given a partition $\la\subseteq (N^n)$, we denote the complement of
$\la$ with respect to $(N^n)$ by $\hat{\la}$, i.e., 
$\hat{\la}=(N-\la_n,\dots,N-\la_1)$.

In Corollary~\ref{Cor_Mac-limit} we replace $\la$ by its complement
with respect to $(N^n)$, where $N$ is an arbitrary integer such that 
$\la_1\leq N$, and also scale $a\mapsto aq^{-N}$.
Using \cite[Equation (4.3)]{BF99}
\[
P_{\hat{\la}}(x;q,t)=(x_1\cdots x_n)^N P_{\la}(x^{-1};q,t)
\]
and
\[
(aq^{-N},q^{N+1}/a;q)_{\infty}=(-a)^N q^{-\binom{N+1}{2}}(a,q/a;q)_{\infty}
\]
this yields
\begin{align*}
&\frac{1}{n!(2\pi\iup)^n} \Int_{\mathbb{T}^n} P_{\la}\big(z^{-1};q,t\big) 
P_{\mu}\bigg(\bigg[z+\frac{t-b}{1-t}\bigg];q,t\bigg) \\[-1mm]
&\qquad\qquad\quad \times 
\prod_{i=1}^n \frac{(a/z_i,qz_i/a;q)_{\infty}}
{(b/z_i,z_i;q)_{\infty}} 
\prod_{1\leq i<j\leq n} \frac{(z_i/z_j,z_j/z_i;q)_{\infty}}
{(tz_i/z_j,tz_j/z_i;q)_{\infty}} \, 
\frac{\dup z_1}{z_1}\cdots \frac{\dup z_n}{z_n} \\[1mm]
&\quad=b^{-\abs{\la}} t^{\abs{\mu}-(n-1)\abs{\la}} 
P_{\la}\bigg(\bigg[\frac{1-t^n}{1-t}\bigg];q,t\bigg)
P_{\mu}\bigg(\bigg[\frac{1-bt^{n-1}}{1-t}\bigg];q,t\bigg) \\[1mm]
&\quad\quad \times 
\prod_{i=1}^n 
\frac{(t,at^{i-m-1}q^{-\la_i},at^{i-n}/b,qt^{i-1}b/a;q)_{\infty}}
{(q,t^i,bt^{i-1},at^{i-n}q^{-\la_i}/b;q)_{\infty}}
\prod_{i=1}^n \prod_{j=1}^m
\frac{(at^{i-j}q^{-\la_i+\mu_j};q)_{\infty}}
{(at^{i-j-1}q^{-\la_i+\mu_j};q)_{\infty}},
\end{align*}
which is independent of $N$.
By
\begin{multline*}
\prod_{i=1}^n \frac{(at^{i-m-1}q^{-\la_i};q)_{\infty}}
{(at^{i-n}q^{-\la_i}/b;q)_{\infty}}
\prod_{i=1}^n \prod_{j=1}^m
\frac{(at^{i-j}q^{-\la_i+\mu_j};q)_{\infty}}
{(at^{i-j-1}q^{-\la_i+\mu_j};q)_{\infty}} \\
=(bt^{n-1})^{\abs{\la}} t^{-n\abs{\mu}}
\frac{(qt^m/a;q,t)_{\la}}{(bqt^{n-1}/a;q,t)_{\la}}
\prod_{i=1}^n \frac{(at^{i-1};q)_{\infty}}
{(at^{i-n}/b;q)_{\infty}}
\prod_{i=1}^n \prod_{j=1}^m
\frac{(qt^{j-i}/a;q)_{\la_i-\mu_j}}
{(qt^{j-i+1}/a;q)_{\la_i-\mu_j}}
\end{multline*}
and the substitution $z\mapsto z^{-1}$ in the integral, the claim follows.
\end{proof}

\section{Open problems}\label{Sec_open}

To conclude the paper we will discuss a number of open problems.

\subsection{Generalising Theorem~\ref{Thm_nplusone} to 
$\boldsymbol{\gamma\neq 1}$}

The main open problem is to generalise Theorem~\ref{Thm_nplusone}
to the case of Jack polynomials.
Using \eqref{Eq_Schur-duality} this theorem can be rewritten as
\begin{align}\label{Eq_gamma-one}
& \bigg\langle 
\bigg(\prod_{r=1}^n 
s_{\lar{r}}\big[\tar{r}-\tar{r-1}\big] \bigg)
s_{\lar{n+1}}\big[\tar{n}+\beta-1\big]
\bigg\rangle_{\alpha_1,\dots,\alpha_n,\beta}^{k_1,\dots,k_n} \\[2mm]
&=\bigg(\prod_{r=1}^n s_{\lar{r}}[k_r-k_{r-1}] \bigg)
s_{\lar{n+1}}[k_n+\beta-1] 
\prod_{\substack{r,s=1\\[1pt] r\neq s}}^{n+1}
\prod_{i=1}^{\ell_r}
\frac{\big(\varepsilon_r (A_{r,s}-k_{s-1}+k_s)-i+1\big)_{\lar{r}_i}}
{\big(\varepsilon_r(A_{r,s}+\varepsilon_s \ell_s)-i+1\big)_{\lar{r}_i}} 
\notag \\
&\quad\times
\prod_{1\leq r<s\leq n+1}
\prod_{i=1}^{\ell_r}\prod_{j=1}^{\ell_s}
\frac{(A_{r,s}-i+\varepsilon_s j+1)_{\lar{r}_i-\varepsilon_s\lar{s}_j}}
{(A_{r,s}-i+\varepsilon_s (j-1)+1)_
{\lar{r}_i-\varepsilon_s\lar{s}_j}}, \notag
\end{align}
where $\ell_1,\dots,\ell_{n+1}$ are arbitrary integers
such that $\ell_r\geq l(\lar{r})$ for $1\leq r\leq n+1$,
$\varepsilon_1=\dots=\varepsilon_n=1$, $\varepsilon_{n+1}=-1$,
$k_0:=0$ and $k_{n+1}:=1-\beta$.
It is not difficult to define a function, say
\[
R^{k_1,\dots,k_n}_{\lar{1},\dots,\lar{n+1}}
(\alpha_1,\dots,\alpha_n,\beta;\gamma),
\]
such that for $\gamma=1$ it gives the right-hand side of
\eqref{Eq_gamma-one} and such that for $\lar{1}=\la$,
$\lar{2}=\dots=\lar{n}=0$ and $\lar{n+1}=\mu$ it yields the 
right-hand side of \eqref{Eq_An-AFLT}.
To describe this function, we generalise our earlier definition
\eqref{Eq_A-def} of $A_r$ and $A_{r,s}$ to include $\gamma$:
\[
A_r:=\alpha_r+\dots+\alpha_n+(k_r-k_{r-1}+r)\gamma
\quad\text{and}\quad
A_{r,s}:=A_r-A_s,
\]
for $1\leq r,s\leq n+1$.
Hence $A_{r,s}=-A_{r,s}$ and
\begin{equation}\label{Eq_Ars}
A_{r,s}=\alpha_r+\dots+\alpha_{s-1}+(k_r-k_{r-1}-k_s+k_{s-1}+r-s)\gamma
\end{equation}
for $1\leq r\leq s\leq n+1$.

\begin{lemma}\label{Lem_RHS}
Let $A_{r,s}$ be as in \eqref{Eq_Ars}, where
$0=k_0\leq k_1\leq k_2\leq\cdots\leq k_n$ are integers and 
$k_{n+1}:=1-\beta/\gamma$.
Set $\varepsilon_1=\dots=\varepsilon_n=1$ and $\varepsilon_{n+1}=-1$,
and define
\begin{align}\label{Eq_RHS}
&R^{k_1,\dots,k_n}_{\lar{1},\dots,\lar{n+1}}
(\alpha_1,\dots,\alpha_n,\beta;\gamma) \\
&\quad :=\bigg(\prod_{r=1}^n P^{(1/\gamma)}_{\lar{r}}[k_r-k_{r-1}]\bigg)
P^{(1/\gamma)}_{\lar{n+1}}[k_n+\beta/\gamma-1] \notag \\
&\qquad\times
\prod_{1\leq r<s\leq n+1}
\frac{(-\varepsilon_s A_{r,s}-
\varepsilon_s(k_{r-1}-k_r)\gamma;\gamma)_{\lar{s}}} 
{(-\varepsilon_s A_{r,s}+\varepsilon_s \ell_r\gamma;\gamma)_{\lar{s}}} \notag\\
&\qquad\times
\prod_{1\leq r<s\leq n}
\Bigg( \frac{(A_{r,s}-(k_{s-1}-k_s)\gamma;\gamma)_{\lar{r}}}
{(1+A_{r,s}+(\varepsilon_s\ell_s-1)\gamma;\gamma)_{\lar{r}}} 
\prod_{i=1}^{\ell_r} \prod_{j=1}^{\ell_s}
\frac{(1+A_{r,s}+(j-i)\gamma)_{\lar{r}_i-\lar{s}_j}}
{(1+A_{r,s}+(j-i-1)\gamma)_{\lar{r}_i-\lar{s}_j}} \Bigg) \notag \\
&\qquad\times
\prod_{r=1}^n \Bigg(
\frac{(A_{r,n+1}-(k_n-k_{n+1})\gamma;\gamma)_{\lar{r}}}
{(A_{r,n+1}-\ell_{n+1}\gamma;\gamma)_{\lar{r}}} 
\prod_{i=1}^{\ell_r} \prod_{j=1}^{\ell_{n+1}}
\frac{(A_{r,n+1}-(i+j-1)\gamma)_{\lar{r}_i+\lar{n+1}_j}}
{(A_{r,n+1}-(i+j-2)\gamma)_{\lar{r}_i+\lar{n+1}_j}} \Bigg), \notag
\end{align}
where $\ell_1,\dots,\ell_{n+1}$ are arbitrary integers such that
$\ell_r\geq l(\lar{r})$ and $\ell_1\leq k_1$.
Then 
\[
R^{k_1,\dots,k_n}_{\lar{1},\dots,\lar{n+1}}
(\alpha_1,\dots,\alpha_n,\beta;\gamma)
\]
is well-defined \textup{(}i.e., independent of the choice of 
the $\ell_r$\textup{)},
\[
R^{0,k_2,\dots,k_n}_{\lar{1},\dots,\lar{n+1}}
(\alpha_1,\dots,\alpha_n,\beta;\gamma)=
\begin{cases}
R^{k_2,\dots,k_n}_{\lar{2},\dots,\lar{n+1}}
(\alpha_2,\dots,\alpha_n,\beta;\gamma)
& \text{if $\lar{1}=0$}, \\[2mm]
0 & \text{otherwise}
\end{cases}
\]
and 
\[
R^{k_1,\dots,k_n}_{\la,\!\!\!\!
\underbrace{\scriptstyle{0,\dots,0}}_{n-1 \text{ times}}
\!\!\!\!,\mu}(\alpha_1,\dots,\alpha_n,\beta;\gamma)
\quad\text{and}\quad
R^{k_1,\dots,k_n}_{\lar{1},\dots,\lar{n+1}}
(\alpha_1,\dots,\alpha_n,\beta;1)
\]
agree with the right-hand side of the $\A_n$ AFLT integral \eqref{Eq_An-AFLT}
and the right-hand side of \eqref{Eq_gamma-one} respectively.
Moreover, when $l(\lar{1})<k_1$ and $k_1,\dots,k_n\geq 1$,
\begin{align}\label{Eq_recR}
& R^{k_1,\dots,k_n}_{\lar{1},\dots,\lar{n+1}}
(\alpha_1,\alpha_2,\dots,\alpha_n,\beta;\gamma) \\
&\quad = R^{k_1-1,\dots,k_n-1}_{\lar{1},\dots,\lar{n+1}}
(\alpha_1+\gamma,\alpha_2,\dots,\alpha_n,\beta+\gamma,\gamma) 
\prod_{s=1}^{n+1} \frac{(-\varepsilon_s A_{1,s}+\varepsilon_s k_1\gamma;
\gamma)_{\lar{s}}}
{(-\varepsilon_s A_{1,s}+\varepsilon_s (k_1-1)\gamma;\gamma)_{\lar{s}}}. 
\notag 
\end{align}
\end{lemma}

We note that the assumption that $l(\lar{1})<k_1$ is not actually
a restriction since
\[
R^{k_1,\dots,k_n}_{\lar{1},\dots,\lar{n+1}}
(\alpha_1,\dots,\alpha_n,\beta;\gamma)
\]
depends on $\lar{1}+\alpha_1$ only, where, for $m$ a scalar,
$\la+m:=(\la_1+m,\la_2+m,\dots)$.
For $n=2$ and $\lar{2}=0$ the recursion \eqref{Eq_recR} agrees with 
\cite[Equation (A.15)]{FL12} (provided $[k_2\gamma]/[(k_2-1)\gamma]$
in the latter is corrected to $[k_2\gamma]_{\la}/[(k_2-1)\gamma]_{\la}$).
For $n>2$ and $\lar{2}=\dots=\lar{n}=0$, however, \eqref{Eq_recR} and
the recursion at the bottom of page~36 of \cite{FL12} are inconsistent.

\begin{proof}[Proof of Lemma~\ref{Lem_RHS}]
To see that the right-hand side of \eqref{Eq_RHS} is independent of the 
$\ell_r$, fix a $t$ such that $1\leq t\leq n+1$. 
Then, assuming that $\lar{t}_{\ell_t}=0$, it follows from elementary 
manipulations and the use of
\[
\frac{(a)_{-n}}{(b)_{-n}}=\frac{(1-b)_n}{(1-a)_n}
\]
that the right-hand side of \eqref{Eq_RHS} reduces to the same expression with
$\ell_t\mapsto \ell_t-1$.

For \eqref{Eq_recR}, write
\[
k_{n+1}=k_{n+1}(\beta;\gamma) \quad\text{and}\quad
A_{r,s}=A_{r,s}^{k_1,\dots,k_n}(\alpha_1,\dots,\alpha_n,\beta;\gamma).
\]
It is then readily checked that 
\begin{gather*}
k_{n+1}(\beta+\gamma;\gamma)=k_{n+1}(\beta;\gamma)-1 \\
A_{r,s}^{k_1-1,\dots,k_n-1}(\alpha_1+\gamma,\alpha_2,\dots,\alpha_n,
\beta+\gamma;\gamma)=
A_{r,s}^{k_1,\dots,k_n}(\alpha_1,\dots,\alpha_n,\beta;\gamma).
\end{gather*}
Hence, for $l(\lar{1})\leq k_1-1$,
\[
\frac{R^{k_1,\dots,k_n}_{\lar{1},\dots,\lar{n+1}}
(\alpha_1,\alpha_2,\dots,\alpha_n,\beta;\gamma)}
{R^{k_1-1,\dots,k_n-1}_{\lar{1},\dots,\lar{n+1}}
(\alpha_1+\gamma,\alpha_2,\dots,\alpha_n,\beta+\gamma,\gamma)}
=\frac{P_{\lar{1}}^{(1/\gamma)}[k_1]}{P_{\lar{1}}^{(1/\gamma)}[k_1-1]}
\prod_{s=2}^{n+1} \frac{(-\varepsilon_s A_{1,s}+\varepsilon_s k_1\gamma;
\gamma)_{\lar{s}}}
{(-\varepsilon_s A_{1,s}+\varepsilon_s (k_1-1)\gamma;\gamma)_{\lar{s}}}.
\]
By the specialisation formula \eqref{Eq_Jack-z} the recursion \eqref{Eq_recR}
follows.

The remaining claims of the lemma are immediate and left to the reader.
\end{proof}

An obvious guess would be that
\begin{align}\label{Eq_guess}
&\bigg\langle \bigg(\prod_{r=1}^n 
P^{(1/\gamma)}_{\lar{r}}\big[\tar{r}-\tar{r-1}\big] \bigg)
P^{(1/\gamma)}_{\lar{n+1}}\big[\tar{n}+\beta/\gamma-1\big]
\bigg\rangle_{\alpha_1,\dots,\alpha_n,\beta;\gamma}^{k_1,\dots,k_n} \\[1mm]
&\qquad\qquad =R^{k_1,\dots,k_n}_{\lar{1},\dots,\lar{n+1}}
(\alpha_1,\dots,\alpha_n,\beta;\gamma), \notag
\end{align}
but this is easily shown to be false unless $\lar{2}=\dots=\lar{n}=0$.
For example, by a direct computation using Theorem~\ref{Thm_An-AFLT} and
the Jack polynomial limit of \eqref{Eq_Prxminy} given by
\[
P_{(r)}^{(1/\gamma)}[x-y]=x^r \, {_2F_1}
\bigg(\genfrac{}{}{0pt}{}{-\gamma,-r}{1-\gamma-r};\frac{y}{x}\bigg),
\]
it follows that for $k_1=\dots=k_n=1$ and 
\[
\big(\lar{1},\dots,\lar{n},\lar{n+1}\big)=\big((u_1),\dots,(u_n),\mu\big)
\]
the left-hand side of \eqref{Eq_guess} evaluates as a product of $n-1$
terminating $_3F_2$ series. 
Specifically,
\begin{align*}
&\bigg\langle 
\bigg(\prod_{r=1}^n 
P^{(1/\gamma)}_{(u_r)}\big[t_r-t_{r-1}\big] \bigg)
P^{(1/\gamma)}_{\mu}\big[t_n+\beta/\gamma-1\big]
\bigg\rangle_{\alpha_1-u_1,\dots,\alpha_n-u_n,\beta;\gamma}^{1,\dots,1} \\
&\qquad =P_{\mu}^{(1/\gamma)}[\beta/\gamma]\,
\frac{(\alpha_1+\dots+\alpha_n+\beta-n\gamma;\gamma)_{\mu}}
{(\alpha_1+\dots+\alpha_n+\beta-(n-1)\gamma;\gamma)_{\mu}} \\
&\quad\qquad \times 
\prod_{r=1}^n
\frac{(1-\alpha_1-\dots-\alpha_r+(r-1)\gamma)_{u_1+\dots+u_r}}
{(1-\alpha_1-\dots-\alpha_r-\beta_r+(r-\delta_{r,n})\gamma)_{u_1+\dots+u_r}} \\
&\quad\qquad \times 
\prod_{r=1}^{n-1} {_3F_2}\bigg(\genfrac{}{}{0pt}{}
{-\gamma,\alpha_1+\dots+\alpha_r-(r-1)\gamma,-u_{r+1}}
{1-\gamma-u_{r+1},1+\alpha_1+\dots+\alpha_r-r\gamma};1\bigg),
\end{align*}
where $t_0:=0$ and $\beta_1=\dots=\beta_{n-1}:=1$.
For $\gamma=1$ the $r$th $_3F_2$ series simplifies
to $\delta_{u_{r+1},0}$ in accordance with the $k_1=\dots=k_n$
case of \eqref{Eq_gamma-one}.
We do not know how to modify the product of Jack polynomials
on the left of \eqref{Eq_guess} so that equality holds.

\subsection{Generalising Theorem~\ref{Thm_An-alt} for  
$\boldsymbol{\gamma=1}$}

Another open problem is to generalise the $\gamma=1$ case of
\eqref{Eq_An-AFLT-alt} to include a product of $n$ Schur functions.
For $\beta_{n-1}+\beta_n=2$, denote by
\[ 
\big\langle \mathscr{O}\big\rangle_
{\alpha_1,\dots,\alpha_n,\beta_{n-1},\beta_n}^{k_1,\dots,k_n}
\]
the $\gamma=1$ case of the $\A_n$ Selberg average
\eqref{Eq_average-alt} (this again requires a complex
integration contour).
Then the problem is to extend the method of Section~\ref{Sec_An-AFLT-Schur}
to prove that
\begin{align*}
& \bigg\langle 
\bigg(\prod_{r=1}^{n-1} 
s_{\lar{r}}\big[\tar{r}-\tar{r-1}\big] \bigg)
s_{\lar{n}}\big[\tar{n}\big]
\bigg\rangle_{\alpha_1,\dots,\alpha_n,\beta_{n-1},\beta_n}
^{k_1,\dots,k_n} \\[2mm]
&\quad\stackrel{?}{=} \prod_{r=1}^n
\prod_{1\leq i<j\leq \ell_r}\frac{\lar{r}_i-\lar{r}_j+j-i}{j-i}
\prod_{1\leq r<s\leq n-1} \prod_{i=1}^{\ell_r}\prod_{j=1}^{\ell_s}
\frac{\lar{r}_i-\lar{s}_j+A_{r,s}+j-i}{A_{r,s}+j-i} \\
&\qquad\times
\prod_{r=1}^{n-1} \prod_{i=1}^{\ell_r}\prod_{j=1}^{k_n}
\frac{A_{r,n+1}-i-j+1}{\lar{r}_i+\lar{n}_j+A_{r,n+1}-i-j+1} \\[1mm]
&\qquad\times
\prod_{r,s=1}^{n-1}\prod_{i=1}^{\ell_r}
\frac{(A_{r,s}-k_{s-1}+k_s-i+1)_{\lar{r}_i}}
{(A_{r,s}+\ell_s-i+1)_{\lar{r}_i}}
\prod_{r=1}^{n-1}\prod_{i=1}^{\ell_r} 
\frac{(A_{r,n}-k_{n-1}+k_n-i+1)_{\lar{r}_i}}
{(A_{r,n}+\beta_{n-1}-i)_{\lar{r}_i}} \\[1mm]
&\qquad\times
\prod_{r=2}^{n-1}\prod_{i=1}^{k_n}
\frac{(A_{r,n+1}+k_{r-1}-k_r-i+1)_{\lar{n}_i}}
{(A_{r,n+1}-\ell_r-i+1)_{\lar{n}_i}}
\prod_{i=1}^{k_n}
\frac{(A_{n,n+1}+k_{n-1}-k_n-i+1)_{\lar{n}_i}}
{(A_{n,n+1}+\beta_n-i)_{\lar{n}_i}},
\end{align*}
where $\tar{0}:=0$, $k_0=k_{n+1}:=0$, $\ell_1=k_1$, $\ell_n=k_n$,
$\ell_r$ for $2\leq r\leq n-1$ are arbitrary nonnegative integers
such that $\ell_r\geq l(\lar{r})$, and the $A_{r,s}$ are defined as in 
\eqref{Eq_A-def}.
The more general average
\[
\bigg\langle \bigg(\prod_{r=1}^n 
s_{\lar{r}}\big[\tar{r}-\tar{r-1}\big] \bigg)
s_{\lar{n}}\big[\tar{n}\big]
\bigg\rangle_{\alpha_1,\dots,\alpha_n,\beta_{n-1},\beta_n}^{k_1,\dots,k_n}
\]
appears not to have a similarly simple evaluation.
For example, for $n=2$ and $\beta_1+\beta_2=\gamma+1$ it follows that
\begin{align*}
&\bigg\langle 
P^{(1/\gamma)}_{(u_1)}\big[t_1\big]
P^{(1/\gamma)}_{(u_2)}\big[t_2-t_1\big]
P^{(1/\gamma)}_{(u_3)}\big[t_2\big]
\bigg\rangle_{\alpha_1,\alpha_2,\beta_1,\beta_2;\gamma}^{1,1} \\[2mm]
&\qquad =
\frac{(\alpha_1)_{u_1}(\alpha_2)_{u_2+u_3}
(\alpha_1+\alpha_2-\gamma)_{u_1+u_2+u_3}}
{(\alpha_1+\beta_1-\gamma)_{u_1}(\alpha_2+\beta_2-\gamma)_{u_2+u_3}
(\alpha_1+\alpha_2)_{u_1+u_2+u_3}} \\[2mm]
&\quad\qquad \times 
{_4F_3}\bigg(\genfrac{}{}{0pt}{}
{-\gamma,\alpha_1+u_1,-\alpha_2+\beta_1-u_2-u_3,-u_2}
{1-\gamma-u_2,\alpha_1+\beta_1-\gamma+u_1,1-\alpha_2-u_2-u_3};1\bigg).
\end{align*}
For $\gamma=1$ this does not vanish when $u_2=0$, but instead yields the
non-uniform expression
\begin{align*}
&\Big\langle s_{(u_1)}\big[t_1\big] s_{(u_2)}\big[t_2-t_1\big]
s_{(u_3)}\big[t_2\big] \Big\rangle_{\alpha_1,\alpha_2,\beta_1,\beta_2}^{1,1} 
\\[2mm]
&\qquad =
\begin{cases}
\displaystyle 
\frac{(\alpha_1)_{u_1}(\alpha_2)_{u_3}(\alpha_1+\alpha_2-1)}
{(\alpha_1+\beta_1-1)_{u_1}(\alpha_2+\beta_2-1)
(\alpha_1+\alpha_2-1+u_1+u_3)} & \text{if $u_2=0$}, \\[5mm]
\displaystyle 
\frac{(\alpha_1)_{u_1}(\alpha_2)_{u_2+u_3-1}
(\alpha_1+\alpha_2-1)(\beta_1-1)}
{(\alpha_1+\beta_1-1)_{u_1+1}(\alpha_2+\beta_2-1)_{u_2+u_3}} 
& \text{if $u_2\geq 1$}.
\end{cases}
\end{align*}

\subsection{An elliptic Selberg and AFLT integral for $\boldsymbol{\A_n}$}

Theorem~\ref{Thm_eAFLT} gives an elliptic generalisation of 
the AFLT integral \eqref{Eq_AFLT}. This integral includes the elliptic
Selberg integral as special case.
Two obvious open problems are to generalise the $\A_n$ Selberg 
integral \eqref{Eq_An-Selberg} to the elliptic level and to then further 
extend this to an elliptic analogue of the $\A_n$ AFLT integral
\eqref{Eq_An-AFLT}.
We intend to address these problems in a future paper \cite{ARW}.

\end{document}